\documentclass[11pt,onecolumn,draftcls]{IEEEtran}

\usepackage{graphicx}
\usepackage{mathtools}
\usepackage{amsmath,amssymb, amsthm, bbm}
\usepackage{epsfig}
\usepackage{varioref}
\usepackage{subfigure}
\usepackage[sort&compress,square,numbers]{natbib}
\usepackage{arydshln}
\usepackage{paralist}
\usepackage[colorlinks]{hyperref}
\usepackage{algpseudocode}
\usepackage{algorithm}
\usepackage{subfigure}
\usepackage{color}
\usepackage{float}
\usepackage{multirow}

\def\Csp{{\mathbb{C}}}

\def\Rsp{{\mathbb{R}}}

\def\Thetasp{{\mbox{\boldmath $\Theta$}}}

\def\Thetaspsc{{\mbox{\boldmath $\Thetasc$}}}

\def\bOmega{{\mbox{\boldmath $\Omega$}}}

\def\pxtheta{P_{\Xmatsc;\thetavecsc}}

\def\pxthetazero{P_{\Xmatsc;\thetavecsc_{0}}}

\def\qxtheta{Q^{\left(u\right)}_{\Xmatsc;\thetavecsc}}

\def\uGausss{u_{\Gausssc}}

\newcommand{\Gausssc}{{\mbox{\tiny${\rm{G}}$}}}

\def\thetaveczero{\thetavec_{0}}
\def\thetaveczerosc{\thetavecsc_{0}}

\def\thetavec{{\bf{\theta}}}

\newcommand{\avec}{{\bf{a}}}

\newcommand{\bvec}{{\bf{b}}}

\newcommand{\yvec}{{\bf{y}}}

\newcommand{\xvec}{{\bf{x}}}

\newcommand{\zvec}{{\bf{z}}}

\newcommand{\hvec}{{\bf{h}}}

\newcommand{\etavec}{{\bf{\eta}}}

\newcommand{\zerovec}{{\bf{0}}}

\newcommand{\alphavec}{{\bf{\alpha}}}

\newcommand{\Gammamat}{{\bf{\Gamma}}}

\newcommand{\Amat}{{\bf{A}}}

\newcommand{\Bmat}{{\bf{B}}}

\newcommand{\Cmat}{{\bf{C}}}

\newcommand{\Fmat}{{\bf{F}}}

\newcommand{\Gmat}{{\bf{G}}}

\newcommand{\Imat}{{\bf{I}}}

\newcommand{\Pmat}{{\bf{P}}}

\newcommand{\Smat}{{\bf{S}}}

\newcommand{\Rmat}{{\bf{R}}}

\newcommand{\Wmat}{{\bf{W}}}

\newcommand{\Xmat}{{\bf{X}}}

\newcommand{\Xmatsc}{{\mbox{\boldmath \tiny $\Xmat$}}}

\newcommand{\Zmatsc}{{\mbox{\boldmath \tiny $\Zmat$}}}

\newcommand{\Amatsc}{{\mbox{\boldmath \tiny $\Amat$}}}

\newcommand{\Wmatsc}{{\mbox{\boldmath \tiny $\Wmat$}}}

\newcommand{\Zmat}{{\bf{Z}}}

\newcommand{\Psimat}{\mbox{\boldmath $\Psi$}}

\newcommand{\bzeta}{\mbox{\boldmath $\zeta$}}

\def\bzeta{{\mbox{\boldmath $\zeta$}}}

\def\bomega{{\mbox{\boldmath $\omega$}}}

\def\bOmega{{\mbox{\boldmath $\Omega$}}}

\def\bxi{{\mbox{\boldmath $\xi$}}}

\def\bSigma{{\mbox{\boldmath $\Sigma$}}}

\def\bomega{{\mbox{\boldmath $\omega$}}}

\def\bxi{{\mbox{\boldmath $\xi$}}}

\def\psivec{{\mbox{\boldmath $\psi$}}}

\def\bSigmasc{{\mbox{\boldmath \tiny{$\bSigma$}}}}

\def\alphavec{{\mbox{\boldmath $\alpha$}}}

\def\alphavecsm{{\mbox{\boldmath \small $\alpha$}}}

\def\etavec{{\mbox{\boldmath $\eta$}}}

\def\thetavec{{\mbox{\boldmath $\theta$}}}

\def\thetavecsc{{\mbox{\boldmath \tiny {$\theta$}}}}

\def\Thetasc{{\mbox{\tiny {$\Theta$}}}}

\def\XCalsc{{\mbox{\tiny $\mathcal{X}$}}}

\def\XCal{{\mbox{$\mathcal{X}$}}}

\def\muvec{{\mbox{\boldmath $\mu$}}}

\newcommand{\be}{\begin{equation}}

\newcommand{\ee}{\end{equation}}

\newcommand{\beqna}{\begin{eqnarray}}

\newcommand{\eeqna}{\end{eqnarray}}

\begin{document}
\title{Measure-Transformed Quasi Maximum Likelihood Estimation}
\author{{Koby Todros and Alfred O. Hero}
\thanks{This work was partially supported by United States-Israel Binational Science Foundation grant 2014334 and United States MURI grant W911NF-15-1-0479.}}

\maketitle
\newtheorem{Theorem}{Theorem}
\newtheorem{Lemma}{Lemma}
\newtheorem{Corollary}{Corollary}
\newtheorem{Conclusion}{Conclusion}
\newtheorem{Proposition}{Proposition}
\newtheorem{Definition}{Definition}
\newtheorem{Remark}{Remark}
\newtheorem{Identity}{Identity}
%%%%%%%%%%%%%%%%%%%%%%%%%%%%%%%%%%%%%%%%%%%%%%%%%%%%%%%%%%%%%%%%%%%%%%%%%%%%%%%%%%%%%%%%%%%%%%%%%%%%%%%%%%%%%%%%%%%%%%%%%%%%%%%%%%%%%%%%%%%%%%%%%%%%%%%%%%%%%%%%%%%%%%%%%%%%%%%%%%%%%%%%%%%%%%%%%%%%%%%%%%%%%%%%%%%%%%%%%%%%%%%%%%%%%%%%%%
\begin{abstract}
In this paper the Gaussian quasi maximum likelihood estimator (GQMLE) is generalized by applying a transform to the probability distribution of the data. The proposed estimator, called measure-transformed GQMLE (MT-GQMLE), minimizes the empirical Kullback-Leibler divergence between a transformed probability distribution of the data and a hypothesized Gaussian probability measure. By judicious choice of the transform we show that, unlike the GQMLE, the proposed estimator can gain sensitivity to higher-order statistical moments and resilience to outliers leading to significant mitigation of the model mismatch effect on the estimates. Under some mild regularity conditions we show that the MT-GQMLE is consistent, asymptotically normal and unbiased. Furthermore, we derive a necessary and sufficient condition for asymptotic efficiency. A data driven procedure for optimal selection of the measure transformation parameters is developed that minimizes the trace of an empirical estimate of the asymptotic mean-squared-error matrix. The MT-GQMLE is applied to linear regression and source localization and numerical comparisons
illustrate its robustness and resilience to outliers. 
\end{abstract}
%%%
%%%%%%%%%%%%%%%%%%%%%%%%%%%%%%%%%%%%%%%%%%%%%%%%%%%%%%%%%%%%%%%%%%%%%%%%%%%%%%%%%%%%%%%%%%%%%%%%%%%%%%%%%%%%%%%%%%%%%
%%%%%%%%%%%%%%%%%%%%%%%%%%%%%%%%%%%%%%%%%%%%%%%%%%%%%%%%%%%%%%%%%%%%%%%%%%%%%%%%%%%%%%%%%%%%%%%%%%%%%%%%%%%%%%%%%%%%%
\section{Introduction}
\label{sec:intro}
Classical multivariate estimation \cite{Fisher1}, \cite{Fisher2} deals with the problem of estimating a deterministic vector parameter using a sequence of multivariate samples from an underlying probability distribution. 
When the probability distribution is known to lie in a specified parametric family of probability measures, the maximum likelihood estimator (MLE) \cite{Lehmann}, \cite{KayBook} can be implemented. In many practical scenarios an accurate parametric model is not available, i.e., the underlying parametric family of probability distributions is unknown. In these cases, alternatives to the MLE that do not require complete knowledge of the likelihood function become attractive.

A popular alternative to the MLE is the quasi maximum likelihood estimator (QMLE) \cite{White}, \cite{Kay}, that similarly to the MLE belongs to the wider classes of M-estimators \cite{Serfling}, \cite{HuberRob} and the generalized method of moments \cite{GMM}. The QMLE minimizes the empirical Kulback-Leibler divergence \cite{Kullback} to the underlying probability distribution of the data over a misspecified parametric family of hypothesized probability measures. When the hypothesized parametric family is the class of Gaussian probability distributions the well established Gaussian QMLE (GQMLE) \cite{White}, \cite{rousseaux2005gaussian}-\cite{Viberg} is obtained that reduces to fitting the mean vector and covariance matrix of the underlying parametric distribution to the sample mean vector (SMV) and sample covariance matrix (SCM) of the data. The GQMLE has gained popularity due to its low implementation complexity, simplicity of performance analysis, and the computational advantages that are analogous to the first and second-order methods of moments \cite{Lehmann}, \cite{KayBook}, \cite{Pearson}. Despite the possible model mismatch, relative to the Gaussian distribution, the GQMLE is a consistent estimator of the true parameters of interest when the mean vector and covariance matrix are correctly specified \cite{White}, \cite{bollerslev1992quasi}, \cite{Trognon}. However, in some circumstances, such as for certain types of non-Gaussian data, deviation from normality inflicts poor estimation performance of the GQMLE even though it may be consistent. This can occur when the mean vector and covariance matrix are weakly dependent on the parameter to be estimated and stronger dependency is present in higher-order cumulants, or in case of heavy-tailed data when the non-robust SMV and SCM provide poor estimates in the presence of outliers.  

In order to overcome this limitation several non-Gaussian QMLEs have been proposed in the econometric and signal processing literature that assume more complex distributional models \cite{bollerslev1987conditionally}-\cite{georgiou2006maximum}. Although non-Gaussian QMLEs can be successful in reducing the effect of model mismatch, they may suffer from the several drawbacks. First, unlike the GQMLE, they may have cumbersome implementations with increased computational and sample complexities and their performance analysis may not be tractable. Second, to be consistent and asymptotically normal they may require considerably more restrictive regularity conditions as compared to the GQMLE \cite{Trognon}, \cite{fan2014quasi}, \cite{Huber}, \cite{newey1997asymptotic}. Furthermore, when the hypothesized parametric family largely deviates from normality they may have low asymptotic relative efficiency (ARE) \cite{Serfling} w.r.t. the MLE under nominal Gaussian distribution of the data.

In this paper, a new generalization of the GQMLE is proposed that applies Gaussian quasi maximum likelihood estimation under a transformed probability distribution of the data. Under the proposed generalization we show that new estimators can be obtained that can gain sensitivity to higher-order statistical moments, resilience to outliers, and yet have the computational advantages of the first and second-order methods of moments. 
This generalization, called measure-transformed GQMLE (MT-GQMLE), is based on the probability measure transformation framework that was recently applied to canonical correlation analysis \cite{MTCCA}, \cite{MTCCAAPP}, and multiple signal classification \cite{Todros2}, \cite{Todros3}. 

The transformation is structured by a non-negative function, called the MT-function, and maps the true probability distribution into a set of new probability measures on the observation space. By modifying the MT-function, classes of measure transformations can be obtained that have different useful properties that mitigate the effect of model mismatch on the estimates. Under the considered transform we define the parametric measure-transformed (MT) mean vector and covariance matrix, establish their sensitivity to higher-order statistical moments, and develop their strongly consistent estimates. These quantities are then used for constructing the MT-GQMLE.

The proposed estimator minimizes the empirical Kullback-Leibler divergence to the transformed probability distribution of the data over a family of hypothesized Gaussian probability measures.
This hypothesized class of Gaussian distributions is characterized by the MT-mean vectors and MT-covariance matrices corresponding to the true underlying parametric family of transformed probability measures. It is worthwhile noting that knowledge of these MT-mean vectors and MT-covariance matrices, which establishes partial information about the underlying distribution, is analogous to the side information in the minimax estimation problem considered in \cite{Ahlswede}. Under some mild regularity conditions we show that the MT-GQMLE is consistent, asymptotically normal and unbiased. Additionally, a necessary and sufficient condition for asymptotic efficiency is derived. Robustness of the proposed MT-GQMLE to outliers is studied using its vector valued influence function \cite{Hampel}. A sufficient condition on the MT-function that guarantees outlier resilience is derived. Furthermore, a data-driven procedure for optimal selection of the MT-function within some parametric class of functions is developed that minimizes the trace of an empirical estimate of the asymptotic mean-squared-error (MSE) matrix. 

We illustrate the MT-GQMLE for linear regression and source localization in the presence of spherically contoured noise. By specifying the MT-function within the family of zero-centered Gaussian-shaped functions parameterized by a scale parameter, we show that the MT-GQMLE can significantly mitigate the model mismatch effect introduced by the normality assumption. More specifically, we show that the proposed MT-GQMLE outperforms the non-robust  GQMLE and other robust alternatives and attains MSE performance that are significantly closer to those obtained by the MLE that, unlike the proposed estimator, requires complete knowledge of the likelihood function.

The basic idea behind the proposed MT-GQMLE was first presented in the conference paper \cite{MTQMLConf}. The contribution of the present paper relative to \cite{MTQMLConf} includes:
\begin{inparaenum}
\item
detailed derivation of the MT-GQMLE,
\item
robustness analysis using its vector valued influence function,
\item
rigorous proofs of the propositions and theorems stating its properties, and
\item
more complete simulation studies.
\end{inparaenum}

The paper is organized as follows. In Section \ref{GQMLERev}, the GQMLE  is reviewed. Section \ref{PMT} reviews the principles of the considered probability measure transformation. In Section \ref{QMLEst} we use this transformation to construct the MT-GQMLE. The proposed estimator is applied to linear regression and source localization in Section \ref{NumExamp}. In Section \ref{Conclusion}, we conclude by summarizing the main points of the paper. The proofs of the propositions and theorems stated throughout the paper are given in the Appendix.
%%%%%%%%%%%%%%%%%%%%%%%%%%%%%%%%%%%%%%%%%%%%%%%%%%%%%%%%%%%%%%%%%%%%%%%%%%%%%%%%%%%%%%%%%%%%%%%%%%%%%%%%%%%%%%%%%%%%%
%%%%%%%%%%%%%%%%%%%%%%%%%%%%%%%%%%%%%%%%%%%%%%%%%%%%%%%%%%%%%%%%%%%%%%%%%%%%%%%%%%%%%%%%%%%%%%%%%%%%%%%%%%%%%%%%%%%%%
\section{The Gaussian quasi maximum likelihood estimator}
\label{GQMLERev}
%%%%%%%%%%%%%%%%%%%%%%%%%%%%%%%%%%%%%%%%%%%%%%%%%%%%%%%%%%%%%%%%%%%%%%%%%%%%%%%%%%%%%%%%%%%%%%%%%%%%%%%%%%%%%%%%%%%%%
\subsection{Preliminaries}
We define the measure space $\left(\XCal,\mathcal{S}_{\XCalsc},\pxthetazero\right)$, where $\XCal\subseteq\Csp^{p}$ is the observation space of a complex-valued random vector $\Xmat$, $\mathcal{S}_{\XCalsc}$ is a $\sigma$-algebra over $\XCal$, and $\pxthetazero$ is a probability measure that belongs to some unknown parametric family of probability measures 
\begin{equation}
\label{PMClass}
\mathcal{P}\triangleq\left\{\pxtheta:\thetavec\in\Thetasp\subseteq\Rsp^{m}\right\}
\end{equation}
on $\mathcal{S}_{\XCalsc}$. The vector $\thetaveczero$ denotes a fixed unknown parameter value to be estimated and the vector $\thetavec$ is a free parameter that indexes the parameter space $\Thetasp$. 
The vector $\thetavec_{0}$ will be called the true vector parameter and $\pxthetazero$ will be called the true probability distribution of $\Xmat$. It is assumed that the family $\left\{\pxtheta: \thetavec\in\Thetasp\right\}$ is absolutely continuous w.r.t. a dominating $\sigma$-finite measure $\rho$ on $\mathcal{S}_{\XCalsc}$, such that the Radon-Nykonym derivative \cite{Folland} 
\begin{equation}
\label{likelihoodfunc}
f\left(\xvec;\thetavec\right)\triangleq\frac{d\pxtheta\left(\xvec\right)}{d\rho\left(\xvec\right)},
\end{equation}
exists for all $\thetavec\in\Thetasp$. The function $f\left(\xvec;\thetavec\right)$ is called the likelihood function of $\thetavec$ observed by the vector $\xvec\in\XCal$.
Let $g:\XCal\rightarrow\Csp$ denote an integrable scalar function on $\XCal$. The expectation of $g\left(\Xmat\right)$ under $\pxtheta$ is defined as
\begin{equation}
\label{ExpDef} 
{\rm{E}}\left[g\left(\Xmat\right);\pxtheta\right]\triangleq\int\limits_{\XCalsc}g\left(\xvec\right)d\pxtheta\left(\xvec\right).
\end{equation}
The empirical probability measure $\hat{P}_{\Xmatsc}$ given a sequence of samples $\Xmat_{n}$, $n=1,\ldots,N$ from $\pxtheta$ is specified by
\begin{equation}
\label{EmpProbMes}
\hat{P}_{\Xmatsc}\left(A\right)=\frac{1}{N}\sum\limits_{n=1}^{N}\delta_{\Xmatsc_{n}}\left(A\right),
\end{equation}
where $A\in\mathcal{S}_{\XCalsc}$, and $\delta_{\Xmatsc_{n}}\left(\cdot\right)$ is the Dirac probability measure at $\Xmat_{n}$ \cite{Folland}.
%%%%%%%%%%%%%%%%%%%%%%%%%%%%%%%%%%%%%%%%%%%%%%%%%%%%%%%%%%%%%%%%%%%%%%%%%%%%%%%%%%%%%%%%%%%%%%%%%%%%%%%%%%%%%%%%%%%%%
\subsection{The Gaussian QMLE}
\label{GQMLE}
Given a sequence of samples from $\pxthetazero$, the GQMLE of $\thetaveczero$ minimizes the empirical Kulback-Leibler divergence (KLD) between the probability measure $\pxthetazero$ and a hypothesized complex circular Gaussian probability distribution \cite{Schreier} $\Phi_{\Xmatsc;\thetavecsc}$ that is characterized by the mean vector $\muvec_{\Xmatsc}\left(\thetavec\right)\triangleq{\rm{E}}\left[\Xmat;\pxtheta\right]$ and the covariance matrix $\bSigma_{\Xmatsc}\left(\thetavec\right)\triangleq{\rm{cov}}\left[\Xmat;\pxtheta\right]$. The KLD between $\pxthetazero$ and $\Phi_{\Xmatsc;\thetavecsc}$ is defined as \cite{Kullback}:
\begin{equation}
\label{KL0}
D_{\rm{KL}}\left[\pxthetazero||\Phi_{\Xmatsc;\thetavecsc}\right]\triangleq{\rm{E}}\left[\log\frac{{f\left(\Xmat;\thetaveczero\right)}}
{\phi\left(\Xmat;\thetavec\right)};\pxthetazero\right],
\end{equation}
where $f\left(\xvec;\thetaveczero\right)$ and 
\begin{equation}
\label{GaussDens0}
\phi\left(\xvec;\thetavec\right)\triangleq{\det}^{-1}\left[\pi\bSigma_{\Xmatsc}\left(\thetavec\right)\right]
\exp\left({-\left(\xvec-\muvec_{\Xmatsc}\left(\thetavec\right)\right)^{H}
\bSigma^{-1}_{\Xmatsc}\left(\thetavec\right)
 \left(\xvec-\muvec_{\Xmatsc}\left(\thetavec\right)\right)}\right)
\end{equation}
are the density functions associated with $\pxthetazero$ and $\Phi_{\Xmatsc;\thetavecsc}$, respectively, w.r.t. the dominating $\sigma$-finite measure $\rho$ on $\mathcal{S}_{\XCalsc}$. An empirical estimate of (\ref{KL0}) given a sequence of samples $\Xmat_{n}$, $n=1,\ldots,N$ from $\pxthetazero$ is defined as:
\begin{equation}
\label{KLEmp0}
\hat{D}_{\rm{KL}}\left[\pxthetazero||\Phi_{\Xmatsc;\thetavecsc}\right]
\triangleq\frac{1}{N}\sum\limits_{n=1}^{N}\log\frac{f\left(\Xmat_{n};\thetaveczero\right)}{\phi\left(\Xmat_{n};\thetavec\right)}.
\end{equation}
The GQMLE of $\thetaveczero$ is obtained by minimization of (\ref{KLEmp0}) w.r.t. $\thetavec$. Since the density $f\left(\cdot;\thetaveczero\right)$ is $\thetavec$-independent this minimization is equivalent to maximization of the term $\frac{1}{N}\sum_{n=1}^{N}\log{\phi\left(\Xmat_{n};\thetavec\right)}$, which by (\ref{GaussDens0}) amounts to maximization of the following objective function:
\begin{equation}
\label{ObjFun0}
J\left(\thetavec\right)\triangleq
-D_{\rm{LD}}\left[\hat{\bSigma}_{\Xmatsc}||\bSigma_{\Xmatsc}\left(\thetavec\right)\right] -
\left\|\hat{\muvec}_{\Xmatsc}-{\muvec}_{\Xmatsc}\left(\thetavec\right)\right\|^{2}_{{\left(\bSigmasc_{\xvec}\left(\thetavecsc\right)\right)}^{-1}},
\end{equation}
where $D_{\rm{LD}}\left[\Amat||\Bmat\right]\triangleq
{\rm{tr}}\left[\Amat\Bmat^{-1}\right]-\log\det\left[\Amat\Bmat^{-1}\right]-p$
is the log-determinant divergence \cite{LogDetDiv} between positive definite matrices $\Amat,\Bmat\in\Csp^{p\times{p}}$, $\left\|\avec\right\|_{\Cmat}\triangleq\sqrt{\avec^{H}\Cmat\avec}$ denotes the weighted Euclidian norm of a vector $\avec\in\Csp^{p}$ with positive-definite weighting matrix $\Cmat\in\Csp^{p\times{p}}$, and $\hat{\muvec}_{\Xmatsc}\triangleq\frac{1}{N}\sum_{n=1}^{N}\Xmat_{n}$ and $\hat{\bSigma}_{\Xmatsc}\triangleq\frac{1}{N}\sum_{n=1}^{N}\left(\Xmat_{n}-\hat{\muvec}_{\Xmatsc}\right)\left(\Xmat_{n}-\hat{\muvec}_{\Xmatsc}\right)^{H}$ denote the standard SMV and SCM.
Hence, the GQMLE of $\thetavec_{0}$ is given by:
\begin{equation}
\label{PropEst0}
\hat{\thetavec}=\arg\max\limits_{\thetavecsc\in\Thetaspsc}J\left(\thetavec\right).
\end{equation}
%%%%%%%%%%%%%%%%%%%%%%%%%%%%%%%%%%%%%%%%%%%%%%%%%%%%%%%%%%%%%%%%%%%%%%%%%%%%%%%%%%%%%%%%%%%%%%%%%%%%%%%%%%%%%%%%%%%%%
%%%%%%%%%%%%%%%%%%%%%%%%%%%%%%%%%%%%%%%%%%%%%%%%%%%%%%%%%%%%%%%%%%%%%%%%%%%%%%%%%%%%%%%%%%%%%%%%%%%%%%%%%%%%%%%%%%%%%
\section{Probability measure transformation}
\label{PMT}
In this section, we review the principles of the probability measure transform \cite{MTCCA}-\cite{Todros3} in the new context of the parametric family (\ref{PMClass}).
We define the parametric measure-transformed mean vector and covariance matrix and show their relation to higher-order statistical moments. Furthermore, we formulate the empirical measured transformed mean and covariance and state the conditions for strong consistency. These quantities are used in the following section to construct the proposed measure-transformed extension of the GQMLE (\ref{PropEst0}). 
%%%%%%%%%%%%%%%%%%%%%%%%%%%%%%%%%%%%%%%%%%%%%%%%%%%%%%%%%%%%%%%%%%%%%%%%%%%%%%%%%%%%%%%%%%%%%%%%%%%%%%%%%%%%%%%%%%%%%
\subsection{Probability measure transform}
The following definition of a transformation on the parametric probability measure $P_{\Xmatsc;\thetavecsc}$ parallels the transformation on a non-parametric distribution stated as Definition 1 in \cite{Todros2}.
\begin{Definition}  
\label{Def1}
Given a non-negative function $u:\Csp^{p}\rightarrow\Rsp_{+}$ satisfying 
\begin{equation} 
\label{Cond}
0<{{{\rm{E}}}\left[u\left(\Xmat\right);P_{\Xmatsc;\thetavecsc}\right]}<\infty,
\end{equation}
a transform on $P_{\Xmatsc;\thetavecsc}$ is defined via the relation:
\begin{equation}
\label{MeasureTransform} 
\qxtheta\left(A\right)\triangleq{\rm{T}}_{u}\left[\pxtheta\right]\left(A\right)=\int\limits_{A}\varphi_{u}\left(\xvec;\thetavec\right)d\pxtheta\left(\xvec\right),
\end{equation}
where $A\in\mathcal{S}_{\XCalsc}$ and
\begin{equation} 
\label{VarPhiDef}  
\varphi_{u}\left(\xvec;\thetavec\right)\triangleq\frac{u\left(\xvec\right)}{{{\rm{E}}}\left[u\left(\Xmat\right);\pxtheta\right]}.
\end{equation}
The function $u\left(\cdot\right)$ is called the MT-function.
\end{Definition}  
%%%
Similarly to  the non-parametric case (Proposition 1 in \cite{Todros2}), the parameterized transformation (\ref{MeasureTransform}) has the following properties.
\begin{Proposition}[Properties of the transform]
\label{Prop1}
Let $\qxtheta$ be defined by relation (\ref{MeasureTransform}). 
Then
\begin{enumerate}[1)]
\item
\label{P1}
$\qxtheta$ is a probability measure on $\mathcal{S}_{\XCalsc}$.
\item
\label{P2}
$\qxtheta$ is absolutely continuous w.r.t. $\pxtheta$, with Radon-Nikodym derivative \cite{Folland}:
\begin{equation}
\label{MeasureTransformRadNik}      
\frac{d\qxtheta\left(\xvec\right)}{d\pxtheta\left(\xvec\right)}=\varphi_{u}\left(\xvec;\thetavec\right).
\end{equation}
%%%
\end{enumerate} 
[Proof: see Appendix A in \cite{Todros2}]
\end{Proposition}
As in \cite{Todros2} we say that the probability measure $\qxtheta$ is ``generated by the MT-function $u\left(\cdot\right)$''. 
%%%%%%%%%%%%%%%%%%%%%%%%%%%%%%%%%%%%%%%%%%%%%%%%%%%%%%%%%%%%%%%%%%%%%%%%%%%%%%%%%%%%%%%%%%%%%%%%%%%%%%%%%%%%%%%%%%%%%
\subsection{The measure-transformed mean and covariance}
\label{MT_MEAN_COV}
According to (\ref{MeasureTransformRadNik}) the mean vector and covariance matrix of $\Xmat$ under $\qxtheta$, called the MT-mean and MT-covariance, are given by the parameterized functions:
\begin{equation}
\label{MTMean}    
\muvec^{\left(u\right)}_{\Xmatsc}\left(\thetavec\right)\triangleq{\rm{E}}\left[\Xmat\varphi_{u}\left(\Xmat;\thetavec\right);\pxtheta\right]
\end{equation}
and
\begin{equation}  
\label{MTCovZ}     
\bSigma^{\left(u\right)}_{\Xmatsc}\left(\thetavec\right)\triangleq{\rm{E}}\left[
\left(\Xmat-\muvec^{\left(u\right)}_{\Xmatsc}\left(\thetavec\right)\right)
\left(\Xmat-\muvec^{\left(u\right)}_{\Xmatsc}\left(\thetavec\right)\right)^{H}\varphi_{u}\left(\Xmat;\thetavec\right);\pxtheta\right],
\end{equation}
respectively. Equations  (\ref{MTMean})  and (\ref{MTCovZ}) imply that $\muvec^{\left(u\right)}_{\Xmatsc}\left(\thetavec\right)$ and $\bSigma^{\left(u\right)}_{\Xmatsc}\left(\thetavec\right)$ are weighted mean and covariance of $\Xmat$ under $\pxtheta$, with the weighting function $\varphi_{u}\left(\cdot;\cdot\right)$ defined in (\ref{VarPhiDef}). By modifying the MT-function $u\left(\cdot\right)$, such that the condition (\ref{Cond}) is satisfied, the MT-mean and MT-covariance are modified. In particular, by choosing $u\left(\cdot\right)$ to be any non-zero constant valued function we have $\qxtheta=\pxtheta$, for which the standard mean vector $\muvec_{\Xmatsc}\left(\thetavec\right)$ and covariance matrix $\bSigma_{\Xmatsc}\left(\thetavec\right)$ are obtained. Alternatively, when $u\left(\cdot\right)$ is non-constant analytic function, which has a convergent Taylor series expansion, the resulting MT-mean and MT-covariance involve \textit{higher-order statistical moments} of $\pxtheta$.
%%%%%%%%%%%%%%%%%%%%%%%%%%%%%%%%%%%%%%%%%%%%%%%%%%%%%%%%%%%%%%%%%%%%%%%%%%%%%%%%%%%%%%%%%%%%%%%%%%%%%%%%%%%%%%%%%%%%%
\subsection{The empirical measure-transformed mean and covariance}
According to  (\ref{MTMean})  and (\ref{MTCovZ}) the values of the parameterized MT-mean and MT-covariance functions can be estimated using only samples from the distribution $\pxtheta$. In the following Proposition, which is a straightforward parametric extension of the non-parametric case stated as Proposition 2 in \cite{Todros2}, strongly consistent estimates of $\muvec^{\left(u\right)}_{\Xmatsc}\left(\thetavec\right)$ and $\bSigma^{\left(u\right)}_{\Xmatsc}\left(\thetavec\right)$ are presented based on i.i.d. samples from $\pxtheta$. 
%%%
\begin{Proposition}[Strongly consistent estimates of the MT-mean and MT-covariance]
\label{ConsistentEst}
Let $\Xmat_{n}$, $n=1,\ldots,N$ denote a sequence of i.i.d. samples from $\pxtheta$. Define the empirical MT-mean and MT-covariance, respectively:
\begin{equation}   
\label{Mu_u_Est} 
\hat{\muvec}^{\left(u\right)}_{\Xmatsc}\triangleq\sum_{n=1}^{N}\Xmat_{n}\hat{\varphi}_{u}\left(\Xmat_{n}\right)
\end{equation}
and
\begin{equation}   
\label{Rx_u_Est}   
\hat{\bSigma}^{\left(u\right)}_{\Xmatsc}\triangleq\sum\limits_{n=1}^{N}\left(\Xmat_{n}-\hat{\muvec}^{\left(u\right)}_{\Xmatsc}\right)\left(\Xmat_{n}-\hat{\muvec}^{\left(u\right)}_{\Xmatsc}\right)^{H}\hat{\varphi}_{u}\left(\Xmat_{n}\right),
\end{equation}
where 
\begin{equation}
\label{hat_varphi}  
\hat{\varphi}_{u}\left(\Xmat_{n}\right)\triangleq\frac{u\left(\Xmat_{n}\right)}{\sum_{n=1}^{N}u\left(\Xmat_{n}\right)}.
\end{equation}
%%%
If 
\begin{equation}
\label{Cond11}
{\rm{E}}\left[\left\|\Xmat\right\|^{2}u\left(\Xmat\right);\pxtheta\right]<\infty
\end{equation}
then
$\hat{\muvec}^{\left(u\right)}_{\Xmatsc}\xrightarrow{\textit{w.p. 1}}{\muvec}^{\left(u\right)}_{\Xmatsc}\left(\thetavec\right)$ and $\hat{\bSigma}^{\left(u\right)}_{\Xmatsc}\xrightarrow{\textit{w.p. 1}}{\bSigma}^{\left(u\right)}_{\Xmatsc}\left(\thetavec\right)$ as $N\rightarrow\infty$, where ``$\xrightarrow{\textit{w.p. 1}}$'' denotes convergence with probability (w.p.) 1 \cite{MeasureTheory}. 
[Proof: see Appendix B in \cite{Todros2}]
\end{Proposition}
Notice that when the MT-function $u(\cdot)$ is non-zero constant valued, the estimators $\hat{\muvec}^{\left(u\right)}_{\Xmatsc}$ and $\frac{N}{N-1}\hat{\bSigma}^{(u)}_{\Xmatsc}$ reduce to the standard unbiased sample mean vector (SMV) and sample covariance matrix (SCM), respectively. 
%%%%%%%%%%%%%%%%%%%%%%%%%%%%%%%%%%%%%%%%%%%%%%%%%%%%%%%%%%%%%%%%%%%%%%%%%%%%%%%%%%%%%%%%%%%%%%%%%%%%%%%%%%%%%%%%%%%%%
%%%%%%%%%%%%%%%%%%%%%%%%%%%%%%%%%%%%%%%%%%%%%%%%%%%%%%%%%%%%%%%%%%%%%%%%%%%%%%%%%%%%%%%%%%%%%%%%%%%%%%%%%%%%%%%%%%%%%
\section{The measure-transformed Gaussian quasi maximum likelihood estimator}
\label{QMLEst}
In this section we extend the GQMLE (\ref{PropEst0}) by applying the transformation (\ref{MeasureTransform}) to the probability measure $\pxtheta$. Similarly to the GQMLE, 
given a sequence of samples from $\pxthetazero$, the proposed MT-GQMLE of $\thetaveczero$ minimizes the empirical Kullback-Leibler divergence between the transformed probability measure $Q^{\left(u\right)}_{\Xmatsc;\thetaveczerosc}$ and a hypothesized complex circular Gaussian probability distribution. This Gaussian measure denoted here as $\Phi^{\left(u\right)}_{\Xmatsc;\thetavecsc}$ is characterized by the MT-mean $\muvec^{(u)}_{\Xmatsc}\left(\thetavec\right)$ and MT-covariance $\bSigma^{(u)}_{\Xmatsc}\left(\thetavec\right)$ corresponding to the transformed probability measure $Q^{\left(u\right)}_{\Xmatsc;\thetavecsc}$. The minimization is carried out w.r.t. the free vector parameter $\thetavec$.
Regularity conditions for consistency, asymptotic normality and unbiasedness are derived. Additionally, we provide a closed-form expression for the asymptotic MSE and obtain a necessary and sufficient condition for asymptotic efficiency. Robustness of the MT-GQMLE to outliers is studied using its influence function. Optimal selection of the MT-function $u\left(\cdot\right)$ out of some parametric class of functions is also discussed.  
%%%%%%%%%%%%%%%%%%%%%%%%%%%%%%%%%%%%%%%%%%%%%%%%%%%%%%%%%%%%%%%%%%%%%%%%%%%%%%%%%%%%%%%%%%%%%%%%%%%%%%%%%%%%%%%%%%%%%
\subsection{The MT-GQMLE}
The KLD between $Q^{\left(u\right)}_{\Xmatsc;\thetaveczerosc}$ and $\Phi^{\left(u\right)}_{\Xmatsc;\thetavecsc}$ is defined as \cite{Kullback}:
\begin{equation}
\label{KL1}
D_{\rm{KL}}\left[Q^{(u)}_{\Xmatsc;\thetaveczerosc}||\Phi^{(u)}_{\Xmatsc;\thetavecsc}\right]\triangleq{\rm{E}}\left[\log\frac{{q^{(u)}\left(\Xmat;\thetaveczero\right)}}
{\phi^{(u)}\left(\Xmat;\thetavec\right)};Q^{\left(u\right)}_{\Xmatsc;\thetaveczerosc}\right],
\end{equation}
where $q^{(u)}\left(\xvec;\thetaveczero\right)$ and 
\begin{equation}
\label{GaussDens}
\phi^{(u)}\left(\xvec;\thetavec\right)\triangleq{\det}^{-1}\left[\pi\bSigma^{(u)}_{\Xmatsc}\left(\thetavec\right)\right]
\exp\left({-\left(\xvec-\muvec^{(u)}_{\Xmatsc}\left(\thetavec\right)\right)^{H}
\left(\bSigma^{(u)}_{\Xmatsc}\left(\thetavec\right)\right)^{-1}
 \left(\xvec-\muvec^{(u)}_{\Xmatsc}\left(\thetavec\right)\right)}\right)
\end{equation}
are the density functions associated with $Q^{\left(u\right)}_{\Xmatsc;\thetaveczerosc}$ and $\Phi^{\left(u\right)}_{\Xmatsc;\thetavecsc}$, respectively. According to  (\ref{MeasureTransformRadNik}), the divergence $D_{\rm{KL}}\left[Q^{(u)}_{\Xmatsc;\thetaveczerosc}||\Phi^{(u)}_{\Xmatsc;\thetavecsc}\right]$ can be estimated using only samples from $\pxthetazero$.
Hence, similarly to (\ref{Mu_u_Est}) and  (\ref{Rx_u_Est}), an empirical estimate of (\ref{KL1}) given a sequence of samples $\Xmat_{n}$, $n=1,\ldots,N$ from $\pxthetazero$ is defined as:
\begin{equation}
\label{KLEmp}
\hat{D}_{\rm{KL}}\left[Q^{(u)}_{\Xmatsc;\thetaveczerosc}||\Phi^{(u)}_{\Xmatsc;\thetavecsc}\right]
\triangleq\sum\limits_{n=1}^{N}\hat{\varphi}_{u}\left(\Xmat_{n}\right)\log\frac{q^{(u)}\left(\Xmat_{n};\thetaveczero\right)}{\phi^{(u)}\left(\Xmat_{n};\thetavec\right)},
\end{equation}
where $\hat{\varphi}_{u}\left(\cdot\right)$ is defined in (\ref{hat_varphi}). The proposed estimator of $\thetaveczero$ is obtained by minimization of (\ref{KLEmp}) w.r.t. $\thetavec$. 
Since the density $q^{(u)}\left(\cdot;\thetaveczero\right)$ is $\thetavec$-independent this minimization is equivalent to maximization of the term $\sum_{n=1}^{N}\hat{\varphi}_{u}\left(\Xmat_{n}\right)\log{\phi^{(u)}\left(\Xmat_{n};\thetavec\right)}$, which by  (\ref{Mu_u_Est}), (\ref{Rx_u_Est}) and (\ref{GaussDens}) amounts to maximization of the following objective function:
\begin{equation}
\label{ObjFun}
J_{u}\left(\thetavec\right)\triangleq
-D_{\rm{LD}}\left[\hat{\bSigma}^{\left(u\right)}_{\Xmatsc}||\bSigma^{\left(u\right)}_{\Xmatsc}\left(\thetavec\right)\right] -
\left\|\hat{\muvec}^{\left(u\right)}_{\Xmatsc}-{\muvec}^{\left(u\right)}_{\Xmatsc}\left(\thetavec\right)\right\|^{2}_{{\left(\bSigmasc^{\left(u\right)}_{\xvec}\left(\thetavecsc\right)\right)}^{-1}},
\end{equation}
where the operators ${D}_{\rm{LD}}\left[\cdot||\cdot\right]$ and $\left\|\cdot\right\|_{(\cdot)}$ are defined below (\ref{ObjFun0}). Thus, the proposed MT-GQMLE is given by:
\begin{equation}
\label{PropEst}
\hat{\thetavec}_{u}=\arg\max\limits_{\thetavecsc\in\Thetaspsc}J_{u}\left(\thetavec\right).
\end{equation}
By modifying the MT-function $u\left(\cdot\right)$, such that the condition (\ref{Cond}) is satisfied the MT-GQMLE is modified, resulting in a family of estimators
generalizing the GQMLE described in Subsection \ref{GQMLE}. In particular, if $u\left(\xvec\right)$ is chosen to be any non-zero constant function over $\XCal$, then $\qxtheta=\pxtheta$ and the standard GQMLE (\ref{PropEst0})  is obtained.
%%%%%%%%%%%%%%%%%%%%%%%%%%%%%%%%%%%%%%%%%%%%%%%%%%%%%%%%%%%%%%%%%%%%%%%%%%%%%%%%%%%%%%%%%%%%%%%%%%%%%%%%%%%%%%%%%%%%%
%%%%%%%%%%%%%%%%%%%%%%%%%%%%%%%%%%%%%%%%%%%%%%%%%%%%%%%%%%%%%%%%%%%%%%%%%%%%%%%%%%%%%%%%%%%%%%%%%%%%%%%%%%%%%%%%%%%%%
\subsection{Asymptotic performance analysis}
Here, we study the asymptotic performance of the proposed estimator (\ref{PropEst}). For simplicity, we assume a sequence of i.i.d. samples $\Xmat_{n}$,  $n=1,\ldots,N$ from $\pxthetazero$. 
\begin{Theorem}[Strong consistency of $\hat{\thetavec}_{u}$]
\label{ConsistencyTh}
Assume that the following conditions are satisfied:
\begin{enumerate}[({A}-1)]
\item
\label{AS_1}
The parameter space $\Thetasp$ is compact.
\item
\label{AS_2}
$\muvec^{(u)}_{\Xmatsc}\left(\thetaveczero\right)\neq\muvec^{(u)}_{\Xmatsc}\left(\thetavec\right)$ or $\bSigma^{(u)}_{\Xmatsc}\left(\thetaveczero\right)\neq\bSigma^{(u)}_{\Xmatsc}\left(\thetavec\right)$ $\forall\thetavec\neq\thetaveczero$. 
\item
\label{AS_3} 
$\bSigma^{(u)}_{\Xmatsc}\left(\thetavec\right)$ is non-singular $\forall\thetavec\in\Thetasp$.
\item
\label{AS_4} 
$\muvec^{\left(u\right)}_{\Xmatsc}\left(\thetavec\right)$ and $\bSigma^{\left(u\right)}_{\Xmatsc}\left(\thetavec\right)$ are continuous over $\Thetasp$.
\item
\label{AS_5}
${\rm{E}}\left[\left\|\Xmat\right\|^{2}u\left(\Xmat\right);\pxthetazero\right]<\infty$.
\end{enumerate}
Then, 
\begin{equation}
\label{StrongConsistency}
\hat{\thetavec}_{u}\xrightarrow{\textit{w.p. 1}}\thetaveczero\hspace{0.2cm}{\text{as $N\rightarrow\infty$}}.
\end{equation}
[A proof is given in Appendix \ref{ConsistencyThProof}]
\end{Theorem}
We note that Assumption A-\ref{AS_2} states an identifiability condition of $\thetaveczero$ under the Gaussian family of probability measures $\left\{\Phi^{(u)}_{\Xmatsc;\thetavecsc}:\thetavec\in\Thetasp\right\}$. This is because $\Phi^{(u)}_{\Xmatsc;\thetavecsc_{0}}\neq\Phi^{(u)}_{\Xmatsc;\thetavecsc}$ $\forall\thetavec\neq\thetaveczero$ if and only if $\muvec^{(u)}_{\Xmatsc}\left(\thetaveczero\right)\neq\muvec^{(u)}_{\Xmatsc}\left(\thetavec\right)$ or $\bSigma^{(u)}_{\Xmatsc}\left(\thetaveczero\right)\neq\bSigma^{(u)}_{\Xmatsc}\left(\thetavec\right)$ $\forall\thetavec\neq\thetaveczero$. If this condition is not satisfied, then by Eq. (\ref{ObjFunDet}) the asymptotic objective function does not have a unique global maximum, leading to inconsistent estimates.  
%%%
\begin{Theorem}[Asymptotic normality and unbiasedness of $\hat{\thetavec}_{u}$]
\label{ANor} 
Assume that the following conditions are satisfied:
\begin{enumerate}[({B}-1)]
\item
\label{C1}
$\hat{\thetavec}_{u}\xrightarrow{P}\thetaveczero$ as $N\rightarrow\infty$, where ``$\xrightarrow{P}$'' denotes convergence in probability \cite{MeasureTheory}.
\item
\label{C2}
$\thetaveczero$ lies in the interior of $\Thetasp$ which is assumed to be compact.
\item
\label{C3}  
$\muvec^{\left(u\right)}_{\Xmatsc}\left(\thetavec\right)$ and $\bSigma^{\left(u\right)}_{\Xmatsc}\left(\thetavec\right)$ are twice continuously differentiable in $\Thetasp$.
\item
\label{C4}
${\rm{E}}\left[u^{2}\left(\Xmat\right);\pxthetazero\right]<\infty$ and ${\rm{E}}\left[\left\|\Xmat\right\|^{4}u^{2}\left(\Xmat\right);\pxthetazero\right]<\infty$.
\end{enumerate}
Then, 
\begin{equation}
\label{Asymp}
\sqrt{N}\left(\hat{\thetavec}_{u}-\thetaveczero\right)\xrightarrow{D}\mathcal{N}\left(\zerovec,\Rmat_{u}\left(\thetaveczero\right)\right)
\hspace{0.2cm}{\textrm{as}}\hspace{0.2cm}N\rightarrow\infty,
\end{equation}
where ``$\xrightarrow{D}$'' denotes convergence in distribution \cite{MeasureTheory},
\begin{equation}
\label{AMSE}  
\Rmat_{u}\left(\thetaveczero\right)=\Fmat^{-1}_{u}\left(\thetaveczero\right)\Gmat_{u}\left(\thetaveczero\right)\Fmat^{-1}_{u}\left(\thetaveczero\right),
\end{equation}
%%% 
\begin{equation}   
\label{GDef}
\Gmat_{u}\left(\thetavec\right)\triangleq{\rm{E}}\left[u^2\left(\Xmat\right)\psivec_{u}\left(\Xmat;\thetavec\right)\psivec^{T}_{u}\left(\Xmat;\thetavec\right);P_{\Xmatsc;\thetaveczerosc}\right],
\end{equation}
%%%
\begin{equation}
\label{PsiDef}
\psivec_{u}\left(\Xmat;\thetavec\right)\triangleq\nabla_{\thetavecsc}\log\phi^{(u)}\left(\Xmat;\thetavec\right),
\end{equation}
%%%
\begin{equation}
\label{FDef}   
\Fmat_{u}\left(\thetavec\right)\triangleq-{\rm{E}}\left[u\left(\Xmat\right)\Gammamat_{u}\left(\Xmat;\thetavec\right);P_{\Xmatsc;\thetaveczerosc}\right],
\end{equation}
%%%
\begin{equation}
\label{GammaDef}  
\Gammamat_{u}\left(\Xmat;\thetavec\right)\triangleq\nabla^{2}_{\thetavecsc}\log\phi^{(u)}\left(\Xmat;\thetavec\right)
\end{equation}
%%%
and it is assumed that $\Fmat_{u}\left(\thetavec\right)$ is non-singular at $\thetavec=\thetaveczero$. [A proof is given in Appendix \ref{ANorProof}]
\end{Theorem}
Theorem \ref{ANor} implies that, similarly to the standard MLE \cite{Lehmann}, \cite{KayBook} and the side-information based MLE considered in \cite{Ahlswede}, the proposed MT-GQMLE converges to $\thetaveczero$ at a rate of $\frac{1}{\sqrt{N}}$.
Clearly, by (\ref{Asymp}) the asymptotic MSE of $\hat{\thetavec}_{u}$ is given by:
\begin{equation}
\label{AMSEN}  
\Cmat_{u}\left(\thetavec_{0}\right)=N^{-1}\Rmat_{u}\left(\thetaveczero\right).
\end{equation}
The following proposition relates (\ref{AMSEN}) and the Cram\'{e}r-Rao lower bound (CRLB) \cite{Cramer}, \cite{Rao}. 
\begin{Proposition}[Relation to the CRLB]
\label{EffTh} 
Let
\begin{equation}
\label{logfGrad}
\etavec(\Xmat;\thetavec)\triangleq\nabla_{\thetavecsc}\log{f}(\Xmat;\thetavec)
\end{equation}
denote the gradient of the logarithm over the likelihood function (\ref{likelihoodfunc}) w.r.t. $\thetavec$. Assume that the following conditions are satisfied:
\begin{enumerate}[({C}-1)]
\item
\label{D1}
For any $\thetavec\in\Thetasp$, the partial derivatives $\frac{\partial{f}\left(\xvec;\thetavecsc\right)}{\partial\theta_{k}}$ and $\frac{\partial^{2}\phi^{(u)}\left(\xvec,\thetavecsc\right)}{\partial\theta_{k}\partial\theta_{j}}$ $k,j=1,\ldots,m$ exist $\rho-\text{a.e.}$, where $\theta_{k}$, $k=1,\ldots,m$ denote the entries of the vector parameter $\thetavec$.
\item
\label{D2}
For any $\thetavec\in\Thetasp$, $v_{k}\left(\xvec;\thetavec\right)\triangleq\frac{\partial\log\phi^{(u)}\left(\xvec;\thetavecsc\right)}{\partial\theta_{k}}f\left(\xvec;\thetavec\right){u}\left(\xvec\right)\in\mathcal{L}_{1}\left(\XCal,\rho\right)$, $k=1,\ldots,m$, where $\mathcal{L}_{1}\left(\XCal,\rho\right)$ denotes the space of absolutely integrable functions, defined on $\XCal$, w.r.t. the measure $\rho$.
\item
\label{D3}
There exist dominating functions $r_{k,j}\left(\xvec\right)\in\mathcal{L}_{1}\left(\XCal,\rho\right)$, $k,j=1,\ldots,m$, such that for any $\thetavec\in\Thetasp$ $\left|\frac{\partial{v}_{k}\left(\xvec;\thetavecsc\right)}{\partial\theta_{j}}\right|\leq{r_{k,j}}\left(\xvec\right)$ $\rho-\text{a.e.}$
\item
The matrix $\Gmat_{u}\left(\thetavec\right)$ (\ref{GDef}) and the Fisher information matrix of $\pxtheta$
$$\Imat_{\rm{FIM}}\left[\pxtheta\right]\triangleq{\rm{E}}\left[\etavec\left(\Xmat;\thetavec\right)\etavec^{T}\left(\Xmat;\thetavec\right);P_{\Xmatsc;\thetavecsc}\right]$$ are non-singular at $\thetavec=\thetaveczero$.
%%%
\end{enumerate}
Then,
\begin{equation}
\label{CRBRel}
\nonumber
\Cmat_{u}\left(\thetaveczero\right)\succeq{N}^{-1}\Imat^{-1}_{\rm{FIM}}\left[\pxthetazero\right],
\end{equation}
where equality holds if and only if
\begin{equation}
\label{AchCRB}  
\etavec\left(\Xmat;\thetaveczero\right)=\Imat_{\rm{FIM}}\left[\pxthetazero\right]\Fmat^{-1}_{u}\left(\thetaveczero\right)\psivec_{u}\left(\Xmat;\thetaveczero\right)u\left(\Xmat\right)\hspace{0.2cm}{\textit{w.p. 1}}.
\end{equation}
[A proof is given in Appendix \ref{EffThProof}]
\end{Proposition}
%%%
One can verify that when $\pxtheta$ is a Gaussian measure, the condition (\ref{AchCRB}) is satisfied only for non-zero constant valued MT-functions resulting in the standard GQMLE (\ref{PropEst0}) that only involves first and second-order moments. This implies that in the Gaussian case, non-constant MT-functions  will always lead to asymptotic performance degradation. In the non-Gaussian case, however, as will be illustrated in the sequel, there are many practical scenarios where a non-constant MT-function can significantly decrease the asymptotic MSE as compared to the GQMLE. This results in estimators with weighted mean and covariance that involve higher-order moments, as discussed in Subsection \ref{MT_MEAN_COV}, and can gain robustness against outliers, as will be discussed in Subsection \ref{Robustness}. By solving the differential equation (\ref{AchCRB}) one can obtain the MT-function for which the resulting MT-GQMLE is asymptotically efficient. Unfortunately, in the non-Gaussian case, the solution is highly cumbersome and requires the knowledge of the likelihood function $f\left(\Xmat;\thetavec\right)$. Therefore, as will be described in Subsection \ref{OptChoice}, we propose an alternative technique for optimizing the choice of the MT-function that is based on the following empirical estimate of the asymptotic MSE (\ref{AMSEN}).  
%%%
\begin{Theorem}[Empirical asymptotic MSE]
\label{EAMSETh}  
Define the empirical asymptotic MSE:
\begin{equation}
\label{EMSEN}
\hat{\Cmat}_{u}(\hat{\thetavec}_{u})\triangleq{N}^{-1}\hat{\Rmat}_{u}(\hat{\thetavec}_{u}),
\end{equation}
where
\begin{equation}
\label{EMSE}  
\hat{\Rmat}_{u}(\hat{\thetavec}_{u})\triangleq{{\hat{\Fmat}}^{-1}_{u}(\hat{\thetavec}_{u})\hat{\Gmat}_{u}(\hat{\thetavec}_{u}){\hat{\Fmat}}^{-1}_{u}(\hat{\thetavec}_{u})},
\end{equation}
%%%
\begin{equation}
\label{Ghat}
\hat{\Gmat}_{u}\left(\thetavec\right)\triangleq{N}^{-1}\sum_{n=1}^{N}u^{2}\left(\Xmat_{n}\right)\psivec_{u}\left(\Xmat_{n};\thetavec\right)\psivec^{T}_{u}\left(\Xmat_{n};\thetavec\right)
\end{equation}
and 
\begin{equation}
\label{Fhat}
\hat{\Fmat}_{u}\left(\thetavec\right)\triangleq-{N}^{-1}\sum_{n=1}^{N}u\left(\Xmat_{n}\right)\Gammamat_{u}\left(\Xmat_{n};\thetavec\right).
\end{equation}
%%%
Furthermore, assume that the following conditions are satisfied:
\begin{enumerate}[({D}-1)]
\item
\label{E1}
$\hat{\thetavec}_{u}\xrightarrow{P}\thetaveczero$ as $N\rightarrow\infty$.
\item
\label{E3}  
$\muvec^{\left(u\right)}_{\Xmatsc}\left(\thetavec\right)$ and $\bSigma^{\left(u\right)}_{\Xmatsc}\left(\thetavec\right)$ are twice continuously differentiable in $\Thetasp$ which is assumed to be compact.
\item
\label{E4}
${\rm{E}}\left[u^{2}\left(\Xmat\right);\pxthetazero\right]<\infty$ and ${\rm{E}}\left[\left\|\Xmat\right\|^{4}u^{2}\left(\Xmat\right);\pxthetazero\right]<\infty$.
\end{enumerate}
Then,
\begin{equation}
\label{ChatConv}
N\|\hat{\Cmat}_{u}(\hat{\thetavec}_{u})-{\Cmat}_{u}(\thetaveczero)\|\xrightarrow{P}0\hspace{0.2cm}\text{as $N\rightarrow\infty$}.
\end{equation}
%\begin{equation}
%\label{ChatConv}
%\hat{\Cmat}_{u}(\hat{\thetavec}_{u})\xrightarrow{P}{\Cmat}_{u}({\thetaveczero})\hspace{0.2cm}\text{as $N\rightarrow\infty$}.
%\end{equation}
[A proof is given in Appendix \ref{EAMSEThProof}]
\end{Theorem}
%%%%%%%%%%%%%%%%%%%%%%%%%%%%%%%%%%%%%%%%%%%%%%%%%%%%%%%%%%%%%%%%%%%%%%%%%%%%%%%%%%%%%%%%%%%%%%%%%%%%%%%%%%%%%%%%%%%%%
\subsection{Robustness to outliers}
\label{Robustness}
Here, we study the robustness of the proposed MT-GQMLE (\ref{PropEst}) to outliers using its vector valued influence function \cite{Hampel}.
Define the probability measure 
\begin{equation}
\label{ContProb}
P_{\epsilon}\triangleq(1-\epsilon)\pxthetazero+\epsilon\delta_{\yvec}, 
\end{equation}
where $0\leq\epsilon\leq1$, $\yvec\in\Csp^{p}$, and
$\delta_{\yvec}$ is the Dirac probability measure at $\yvec$. The influence function of a Fisher consistent estimator \cite{Cox} with statistical functional $\textrm{S}[\cdot]$ at probability distribution $\pxthetazero$ is defined as \cite{Hampel}:
\begin{equation}
\label{IFDef}
{\rm{IF}}\left(\yvec;\thetaveczero\right)\triangleq\lim\limits_{\epsilon\rightarrow{0}}\frac{\textrm{S}\left[P_{\epsilon}\right]-\textrm{S}\left[\pxthetazero\right]}{\epsilon}
=\left.\frac{\partial{\textrm{S}}\left[{P}_{\epsilon}\right]}{\partial{\epsilon}}\right|_{\epsilon=0}.
\end{equation}
The influence function describes the effect on the estimator of an infinitesimal contamination at the point $\yvec$. An estimator is said to be B-robust if for any fixed value of $\thetaveczero\in\Thetasp$ its influence function is bounded over $\Csp^p$ \cite{Hampel}. 

The following Proposition states that under some mild regularity conditions, $\hat{\thetavec}_{u}$ is Fisher consistent, i.e., it can be represented as a statistical functional of the empirical probability distribution $\Smat_{u}[\hat{P}_{\Xmatsc}]$ that satisfies $\Smat_{u}[\pxthetazero]=\thetaveczero$.
%%%
\begin{Proposition}[Fisher consistency of $\hat{\thetavec}_{u}$]
\label{FishConsistency}
Assume that conditions A-\ref{AS_1}-A-\ref{AS_4} stated in Theorem \ref{ConsistencyTh} are satisfied. Then, $\hat{\thetavec}_{u}$ is Fisher consistent.
[A proof is given in Appendix \ref{FishConsistencyProof}]
\end{Proposition}
%%%
Assuming that the proposed MT-GQMLE is Fisher consistent, an expression for its corresponding influence function is established in the following Proposition: 
%%%
\begin{Proposition}[Influence function of $\hat{\thetavec}_{u}$]
\label{InfFuncExp}
Assume that $\hat{\thetavec}_{u}$ is Fisher consistent. Furthermore, assume that the following conditions are satisfied:
\begin{enumerate}[({E}-1)]
\item
\label{F1}
For any $\thetavec\in\Thetasp$, the partial derivatives $\frac{\partial^{2}\phi^{(u)}\left(\xvec,\thetavecsc\right)}{\partial\theta_{k}\partial\theta_{j}}$ $k,j=1,\ldots,m$ exist $\pxthetazero-\text{a.e.}$, where $\theta_{k}$, $k=1,\ldots,m$ denote the entries of the vector parameter $\thetavec$.
\item
\label{F2}
For any $\thetavec\in\Thetasp$, $v_{k}\left(\xvec;\thetavec\right)\triangleq\frac{\partial\log\phi^{(u)}\left(\xvec;\thetavecsc\right)}{\partial\theta_{k}}{u}\left(\xvec\right)\in\mathcal{L}_{1}\left(\XCal,\pxthetazero\right)$, $k=1,\ldots,m$, where $\mathcal{L}_{1}\left(\XCal,\pxthetazero\right)$ denotes the space of absolutely integrable functions, defined on $\XCal$, w.r.t. the measure $\pxthetazero$.
\item
\label{F3}
There exist dominating functions $r_{k,j}\left(\xvec\right)\in\mathcal{L}_{1}\left(\XCal,\pxthetazero\right)$, $k,j=1,\ldots,m$, such that for any $\thetavec\in\Thetasp$ $\left|\frac{\partial{v}_{k}\left(\xvec;\thetavecsc\right)}{\partial\theta_{j}}\right|\leq{r_{k,j}}\left(\xvec\right)$ $\pxthetazero-\text{a.e.}$
\item
\label{F4}
The matrix $\Fmat_{u}\left(\thetavec\right)$ (\ref{FDef}) is non-singular at $\thetavec=\thetaveczero$.
\end{enumerate}
Then, the influence function of $\hat{\thetavec}_{u}$ is given by:
\begin{equation}
\label{InfFuncMTQML}
{\bf{IF}}\left(\yvec;\thetaveczero\right)=\Fmat^{-1}_{u}\left(\thetaveczero\right)\psivec_{u}\left(\yvec;\thetaveczero\right)u\left(\yvec\right).
\end{equation}
[A proof is given in Appendix \ref{InfFuncExpProof}]
\end{Proposition}
%%%
Interestingly, using (\ref{VarPhiDef}),  (\ref{MeasureTransformRadNik}), (\ref{GaussDens}), (\ref{FDef}) and (\ref{GammaDef}) one can verify that the influence function (\ref{InfFuncMTQML}) can be rewritten as ${\bf{IF}}\left(\yvec;\thetaveczero\right)=\Imat^{-1}_{\rm{FIM}}[\Phi^{(u)}_{\Xmatsc;\thetavecsc_{0}}]\psivec_{u}\left(\yvec;\thetaveczero\right)\varphi_{u}\left(\yvec;\thetaveczero\right)$, where $\Imat_{\rm{FIM}}[\Phi^{(u)}_{\Xmatsc;\thetavecsc_{0}}]\triangleq-{\rm{E}}\left[\Gammamat_{u}\left(\Xmat;\thetaveczero\right);\Phi^{(u)}_{\Xmatsc;\thetavecsc_{0}}\right]$ is the Fisher information of the Gaussian probability measure $\Phi^{(u)}_{\Xmatsc;\thetavecsc}$ at $\thetavec=\thetaveczero$. We note that the  term $\Imat^{-1}_{\rm{FIM}}[\Phi^{(u)}_{\Xmatsc;\thetavecsc_{0}}]\psivec_{u}\left(\yvec;\thetaveczero\right)$ is the unbounded influence function of the MLE under $\Phi^{(u)}_{\Xmatsc;\thetavecsc_{0}}$ \cite{Hampel}. Hence, we conclude that (\ref{InfFuncMTQML}) is a weighted version of the influence function of the MLE under $\Phi^{(u)}_{\Xmatsc;\thetavecsc_{0}}$, with the weighting function $\varphi_{u}\left(\yvec;\thetaveczero\right)$  (\ref{VarPhiDef}). 

The following proposition states sufficient conditions on the MT-function $u\left(\cdot\right)$ under which the influence function (\ref{InfFuncMTQML}) is bounded over a subset $\mathcal{C}\subseteq{\Csp^{p}}$ and decays to zero over this subset as the outlier norm approaches infinity. When these conditions are satisfied over $\mathcal{C}={\Csp^{p}}$ the proposed MT-GQMLE is B-robust and rejects large norm outliers in any direction. 
\begin{Proposition}
\label{RobustnessConditions}
Let $\mathcal{C}$ denote a subset of $\Csp^{p}$.
\begin{inparaenum}
\item
\label{G1}
If the MT-function $u(\yvec)$ and the product $u(\yvec)\|\yvec\|^{2}$ are bounded over $\mathcal{C}$ then the influence function (\ref{InfFuncMTQML}) is bounded over $\mathcal{C}$ . 
\item
\label{G2}
If over the subset $\mathcal{C}$ $u(\yvec)\rightarrow{0}$ and $u(\yvec)\|\yvec\|^{2}\rightarrow{0}$ as $\|\yvec\|\rightarrow{\infty}$, then $\left\|{\bf{IF}}\left(\yvec;\thetaveczero\right)\right\|\rightarrow{0}$ over $\mathcal{C}$ as $\|\yvec\|\rightarrow{\infty}$.
\end{inparaenum}
\newline
[A proof is given in Appendix \ref{RobustnessConditionsProof}] 
\end{Proposition}
%Clearly, when the condition in part \ref{G1} of the proposition is satisfied over $\mathcal{C}=\Csp^{p}$ the proposed MT-GQMLE is B-robust.
%Additionally, if the condition in part \ref{G2} holds for $\mathcal{C}=\Csp^{p}$ the MT-GQMLE rejects large norm outliers in any direction. 
%%%
%\begin{enumerate}
%\item
%Fisher consistency:
%\item
%The influence function is given by:
%\begin{Proposition}
%\label{RobustnessConditions}
%The influence function (\ref{InfFuncMTQML}) is bounded for any fixed value of $\thetaveczero\in\Thetasp$ if the MT-function $u(\yvec)$ and the product $u(\yvec)\|\yvec\|^{2}_{2}$ are bounded over $\Csp^{p}$. [A proof is given in Appendix \ref{RobustnessConditionsProof}] 
%\end{Proposition}
%\item
%Relation to the asymptotic MSE (\ref{AMSE}): $\Cmat_{u}\left(\thetaveczero\right)=N^{-1}{\rm{E}}\left[{\bf{IF}}\left(\Xmat;\thetaveczero\right){\bf{IF}}^{T}\left(\Xmat;\thetaveczero\right);\pxthetazero\right]$
%\end{enumerate}
%%%
%%%%%%%%%%%%%%%%%%%%%%%%%%%%%%%%%%%%%%%%%%%%%%%%%%%%%%%%%%%%%%%%%%%%%%%%%%%%%%%%%%%%%%%%%%%%%%%%%%%%%%%%%%%%%%%%%%%%%
%%%%%%%%%%%%%%%%%%%%%%%%%%%%%%%%%%%%%%%%%%%%%%%%%%%%%%%%%%%%%%%%%%%%%%%%%%%%%%%%%%%%%%%%%%%%%%%%%%%%%%%%%%%%%%%%%%%%%
\subsection{Optimization of the choice of the MT-function}
\label{OptChoice}
Here we restrict the class of MT-functions to a parametric family $\left\{u\left(\Xmat;\bomega\right), \bomega\in\bOmega\subseteq\Csp^{r}\right\}$ that satisfies the conditions stated in Definition \ref{Def1} and Theorem \ref{EAMSETh}. For example, the Gaussian family of functions that satisfy the conditions in Proposition \ref{RobustnessConditions} is a natural choice for inducing outlier resistance.
An optimal choice of the MT-function parameter $\bomega$ would minimize the trace of the empirical asymptotic MSE matrix (\ref{EMSEN}) that is constructed by the same sequence of samples used for obtaining the MT-GQMLE (\ref{PropEst}). Note that under the conditions of Theorem \ref{EAMSETh} the trace of (\ref{EMSEN}) is a consistent estimator of the mean-squared-deviation ${\rm{E}}[\|\hat{\thetavec}_{u}-\thetavec_{0}\|^{2};P_{\hat{\thetavecsc}_{u}}]$.
%%%%%%%%%%%%%%%%%%%%%%%%%%%%%%%%%%%%%%%%%%%%%%%%%%%%%%%%%%%%%%%%%%%%%%%%%%%%%%%%%%%%%%%%%%%%%%%%%%%%%%%%%%%%%%%%%%%%%%%%%%%%%%%%%%%%%%%%%%%%%%%%%%%%%%%%%%%%%%%%%%%%%%%%%%%%%%%%%%%%%%%%%%%%%%%%%%%%%%%%%%%%%%%%%%%%%%%%%%%%%%%%%%%%%%%%%%
\section{Examples}
\label{NumExamp}
In this section the proposed MT-GQMLE (\ref{PropEst}) is applied to linear regression and source localization problems for the purpose of evaluation of its MSE performance and computational load as compared to other estimators. All simulation studies were performed on an iMac computer with 8 GB RAM and a 2.9 GHz Intel Core i5 processor with 4 cores.
%%%%%%%%%%%%%%%%%%%%%%%%%%%%%%%%%%%%%%%%%%%%%%%%%%%%%%%%%%%%%%%%%%%%%%%%%%%%%%%%%%%%%%%%%%%%%%%%%%%%%%%%%%%%%%%%%%%%%
\subsection{Linear regression}
\label{GainParam}
We consider the following linear observation model:
\begin{equation}
\label{GainModel}
\Xmat_{n}=\Amat\alphavec_{0}+\Wmat_{n},\hspace{0.2cm}n=1,\ldots,N,
\end{equation}
where $\Xmat_{n}\in\Csp^{p}$ is an observation vector, $\Amat\in\Csp^{p\times{q}}$ is a known deterministic matrix of regressors with $q<p$ linearly independent columns,
$\alphavec_{0}\in\Csp^{q}$ is an unknown deterministic vector of regression coefficients and $\Wmat_{n}\in\Csp^{p}$ is an additive noise. The vector parameter to be estimated is defined as:
%%%
\begin{equation}
\label{TrueParam}
\thetavec_{0}\triangleq\left[{\rm{Re}}\left\{\alphavec_{0}\right\}^{T},{\rm{Im}}\left\{\alphavec_{0}\right\}^{T}\right]^{T}\in\Rsp^{m},\hspace{0.2cm}m=2q.
\end{equation}
%%%
We specify the MT-function in the set:
\begin{equation}
\label{OrthSet}
\left\{u\left(\xvec\right)=v(\Pmat^{\bot}_{\Amatsc}\xvec),\hspace{0.1cm}v:\Csp^{p}\rightarrow\Rsp_{+}\right\},
\end{equation}
where $\Pmat^{\bot}_{\Amatsc}$ is the projection matrix onto the subspace orthogonal to the range space of $\Amat$. Assuming that condition (\ref{Cond}) is satisfied, one can verify using (\ref{VarPhiDef}), (\ref{MTMean}),  (\ref{MTCovZ}), (\ref{GainModel}) and  (\ref{OrthSet}) that the MT-mean and MT-covariance under the transformed probability measure $Q^{(u)}_{\Xmatsc;\thetavecsc}$ satisfy the following properties:
\begin{equation}
\label{MT_MEAN_GAIN}
\muvec^{(u)}_{\Xmatsc}\left(\thetavec\right)=\Amat\alphavec + \muvec^{(u)}_{\Wmatsc}
\end{equation}
and
\begin{equation}
\label{MT_COV_GAIN}
\bSigma^{(u)}_{\Xmatsc}\left(\thetavec\right)=\bSigma^{(u)}_{\Wmatsc},
\end{equation}
where $\muvec^{(u)}_{\Wmatsc}$ and $\bSigma^{(u)}_{\Wmatsc}$ are the MT-mean and MT-covariance of the noise component. Hence, by substituting (\ref{MT_MEAN_GAIN}) and (\ref{MT_COV_GAIN}) into (\ref{ObjFun}) the resulting MT-GQMLE (\ref{PropEst}) is given by: 
\begin{equation}
\label{MTQMLE_GAIN_1}
\hat{\thetavec}_{u}=\left[{\rm{Re}}\left\{\hat{\alphavec}^{(u)}\right\}^{T},{\rm{Im}}\left\{\hat{\alphavec}^{(u)}\right\}^{T}\right]^{T},
\end{equation}
where $\hat{\alphavec}^{(u)}\triangleq(\Amat^{H}(\bSigma^{(u)}_{\Wmatsc})^{-1}\Amat)^{-1}\Amat^{H}(\bSigma^{(u)}_{\Wmatsc})^{-1}(\hat{\muvec}^{(u)}_{\Xmatsc}-{\muvec}^{(u)}_{\Wmatsc})$.

We further assume that the noise component is spherically contoured with stochastic representation \cite{Visa}:
\begin{equation}
\label{CompGaussOrth}
\Wmat_{n}=\nu_{n}\Zmat_{n},
\end{equation}
where $\nu_{n}\in\Rsp_{++}$ is a first-order stationary process and $\Zmat_{n}\in\Csp^{p}$ is a proper-complex wide-sense stationary Gaussian process with zero-mean and scaled unit covariance $\sigma^{2}_{\Zmatsc}\Imat$. The processes $\nu_{n}$ and $\Zmat_{n}$ are also assumed to be statistically independent. 

In order to mitigate the effect of outliers and gain sensitivity to parametric variation of the higher-order moments, we specify the MT-function in a subset of (\ref{OrthSet}) that is comprised of zero-centred Gaussian functions parametrized by a width parameter $\omega$, i.e.,
\begin{equation}
\label{GaussMTFuncOrth}
\uGausss\left(\xvec;\omega\right)\triangleq\exp\left(-{\|\Pmat^{\bot}_{\Amatsc}\xvec\|^{2}}/{\omega^{2}}\right),\hspace{0.2cm\omega\in\Rsp_{++}}.
\end{equation}
Using (\ref{MTMean}), (\ref{MTCovZ}), (\ref{CompGaussOrth}) and (\ref{GaussMTFuncOrth}) it can be shown that the MT-mean and MT-covariance of the noise satisfy:
\begin{equation}
\label{MGGain}
\muvec^{(u_{\rm{G}})}_{\Wmatsc}\left(\omega\right)=\zerovec
\end{equation}
and 
\begin{equation}
\label{CGGain}
\bSigma^{(u_{\rm{G}})}_{\Wmatsc}\left(\omega\right)=r_{0}\left(\omega\right)\Pmat_{\Amatsc} + r_{1}\left(\omega\right)\Imat,
\end{equation} 
respectively, where $r_{0}\left(\omega\right)$ and $r_{1}\left(\omega\right)$ are some strictly positive functions of $\omega$ and $\Pmat_{\Amatsc}$ is the projection matrix onto the range space of $\Amat$. Hence, under the spherically contoured noise assumption (\ref{CompGaussOrth}) and the Gaussian MT-function (\ref{GaussMTFuncOrth}), the MT-GQMLE (\ref{MTQMLE_GAIN_1}) reduces to:
\begin{equation}
\label{MTQMLE_GAIN_2}
\hat{\thetavec}_{\uGausss}\left(\omega\right)=\left[{\rm{Re}}\left\{\hat{\alphavec}^{(\uGausss)}\left(\omega\right)\right\}^{T},{\rm{Im}}\left\{\hat{\alphavec}^{(\uGausss)}\left(\omega\right)\right\}^{T}\right]^{T},
\end{equation}
where $\hat{\alphavec}^{(\uGausss)}\left(\omega\right)\triangleq({\Amat^{H}\Amat})^{-1}\Amat^{H}\hat{\muvec}^{(\uGausss)}_{\Xmatsc}\left(\omega\right)$.

The asymptotic MSE (\ref{AMSEN}) of $\hat{\thetavec}_{\uGausss}\left(\omega\right)$ is given by:
\begin{equation}
\label{AMSE_GAIN}
\Cmat_{\uGausss}\left(\thetavec_{0};\omega\right)=\left(\frac{{\rm{E}}\left[\left(\frac{\omega^{2}}{2\sigma^{2}_{\zvec}\nu^{2}+\omega^{2}}\right)^{p-m}\nu^{2};P_{\nu}\right]}
{{\rm{E}}^{2}\left[\left(\frac{\omega^{2}}{\sigma^{2}_{\zvec}\nu^{2}+\omega^{2}}\right)^{p-m};P_{\nu}\right]}\right)\times\left(\frac{\sigma^{2}_{\Zmatsc}}{2N}\Bmat\right),
\end{equation}
where
\begin{equation}
\Bmat\triangleq
\left[{\begin{array}{*{20}c} 
{\rm{Re}}\left\{\Amat^{H}\Amat\right\} & 
-{\rm{Im}}\left\{\Amat^{H}\Amat\right\}
\\
{\rm{Im}}\left\{\Amat^{H}\Amat\right\} &
 {\rm{Re}}\left\{\Amat^{H}\Amat\right\} \end{array}}\right]^{-1}.
\end{equation}
The asymptotic MSE (\ref{AMSE_GAIN}) is comprised of two terms. The first term, which is non-linear in $\sigma^{2}_{\Zmatsc}$ and $\omega^{2}$, arises from the transformation (\ref{MeasureTransform}) of the probability measure $\pxtheta$. The second term is equivalent to the Gaussian CRLB \cite{KayBook} for estimating $\thetaveczero$. Notice that when the noise (\ref{CompGaussOrth}) is Gaussian, i.e. the texture parameter $\nu=1$ w.p. 1, $\Cmat_{\uGausss}\left(\thetavec_{0};\omega\right)\rightarrow\frac{\sigma^{2}_{\Zmatsc}}{2N}\Bmat$ as $\omega\rightarrow\infty$. Since $\uGausss\left(\xvec;\omega\right)\rightarrow{1}$ as $\omega\rightarrow\infty$ this result verifies the conclusion following Proposition \ref{EffTh}, which states that in the Gaussian case the achievability condition (\ref{AchCRB}) of the CRLB is satisfied for non-zero constant MT-function. Using (\ref{EMSEN}) and (\ref{MTQMLE_GAIN_2}) we obtain the following empirical estimate of the asymptotic MSE:
\begin{equation}
\label{PredMSEGain}  
\hat{\Cmat}_{\uGausss}\left(\hat{\thetavec}_{\uGausss}\left(\omega\right);\omega\right)=
\frac{\sum_{n=1}^{N}u^{2}_{\rm{G}}\left(\Xmat_{n};\omega\right)\bzeta\left(\Xmat_{n};\omega\right)\bzeta^{T}\left(\Xmat_{n};\omega\right)}
{(\sum_{n=1}^{N}u_{\rm{G}}\left(\Xmat_{n};\omega\right))^{2}}
\end{equation}
where $\bzeta\left(\xvec;\omega\right)\triangleq\Bmat\left[{\rm{Re}}\left\{\hvec\left(\xvec;\omega\right)\right\}^{T},{\rm{Im}}\left\{\hvec\left(\xvec;\omega\right)\right\}^{T}\right]^{T}
$ and $\hvec\left(\xvec;\omega\right)\triangleq\Amat^{H}\left(\xvec-\hat{\muvec}^{(u_{\rm{G})}}_{\Xmatsc}\left(\omega\right)\right)$. As discussed in Subsection \ref{OptChoice}, (\ref{PredMSEGain}) will be used for optimizing the width parameter $\omega$ of the Gaussian MT-function (\ref{GaussMTFuncOrth}). 

Next, we study the robustness to outliers. The vector valued influence function (\ref{InfFuncMTQML}) of the MT-GQMLE (\ref{MTQMLE_GAIN_2})  is given by:
\begin{equation}
\label{IF_NORM_GAIN}
\textbf{IF}\left(\yvec;\thetavec_{0}\right)={\rm{E}}^{-1}\left[\left(\frac{\omega^{2}}{\sigma^{2}_{\Zmatsc}\nu^{2}+\omega^{2}}\right)^{p-m};P_{\nu}\right]\left(\Bmat
\left[{\begin{array}{*{20}c} 
{\rm{Re}}\left\{\Amat^{H}\yvec\right\}  
\\
{\rm{Im}}\left\{\Amat^{H}\yvec\right\}  \end{array}}\right]
-\thetaveczero\right)\exp\left(-\frac{\left\|\Pmat^{\bot}_{\Amat}\yvec\right\|^2}{\omega^{2}}\right).
\end{equation}
Note that when $\yvec\in\mathcal{A}\triangleq\left\{\yvec\in\Csp^{p}:\frac{\left\|\Pmat_{\Amat}\yvec\right\|^{2}}{\left\|\yvec\right\|^{2}}=1\right\}$, which is the range space of $\Amat$, the influence function can grow unbounded as $\|\yvec\|\rightarrow\infty$. Hence, $\hat{\thetavec}_{\uGausss}\left(\omega\right)$ is not robust against outliers in the range space of $\Amat$. However, as follows from the following Proposition, $\hat{\thetavec}_{\uGausss}\left(\omega\right)$ is outlier robust and rejects large norm outliers over a sufficiently large subset of $\Csp^{p}$.
%%%
\begin{Proposition}
\label{RobGainProp}
Define the set $\mathcal{B}_{\epsilon}\triangleq\left\{\yvec\in\Csp^{p}:\frac{\left\|\Pmat_{\Amat}\yvec\right\|^{2}}{\left\|\yvec\right\|^{2}}\leq1-\epsilon\right\}$, where $\epsilon>0$ is some small constant. For any fixed width parameter $\omega$, the MT-function (\ref{GaussMTFuncOrth}) satisfies the following properties over the set $\mathcal{B}_{\epsilon}$: 
\begin{inparaenum}
\item
$\uGausss\left(\yvec;\omega\right)$ and $\uGausss\left(\yvec;\omega\right)\left\|\yvec\right\|^{2}$ are bounded. 
\item
$\uGausss\left(\yvec;\omega\right)\rightarrow{0}$ and $\uGausss\left(\yvec;\omega\right)\left\|\yvec\right\|^{2}\rightarrow{0}$ as $\|\yvec\|\rightarrow\infty$.
\end{inparaenum}
[A proof is given in Appendix \ref{RobGainPropProof}]
\end{Proposition}
Therefore, by Propositions \ref{RobustnessConditions} and \ref{RobGainProp} the influence function of $\hat{\thetavec}_{\uGausss}\left(\omega\right)$ is bounded over $\mathcal{B}_{\epsilon}$ and decays to zero over this subset as the  outlier norm approaches infinity. Notice that while the probability measure of $\mathcal{A}$ is $P_{\Xmatsc}\left(\mathcal{A}\right)=0$, the probability measure of  $\mathcal{B}_{\epsilon}$ satisfies $P_{\Xmatsc}\left(\mathcal{B}_{\epsilon}\right)\approx1$ for sufficiently small $\epsilon$.
Hence, we conclude that $\hat{\thetavec}_{\uGausss}\left(\omega\right)$ is robust to outliers with sufficiently high probability. 

In the following simulation examples we evaluate the MSE performance of the MT-GQMLE (\ref{MTQMLE_GAIN_2}) as compared to the standard GQMLE (\ref{PropEst0}), the non-linear least squares B-robust M-estimator \cite{Serfling}, \cite{HuberRob} based on Tukey's bi-square loss function \cite{Tukey}, and the omniscient MLE. These are evaluated for a specific choice of true vector parameter $\thetavec_{0}=\left[\left[0.3,0.5\right],\left[0.6,0.8\right]\right]^{T}$, observations dimensionality $p=10$ and matrix of regressors $\Amat=\frac{1}{\sqrt{2}}\left[\avec_{0},\avec_{1}\right]$, where $\avec_{k}\triangleq\frac{1}{\sqrt{p}}\left[1,\exp\left(i\vartheta_{k}\right),\ldots,\exp\left(i(p-1)\vartheta_{k}\right)\right]^{T}$, $k=0,1$, $\vartheta_{0}=\pi/3$ and $\vartheta_{1}=\pi/6$. We considered two types of noise distributions with zero location parameter and isotropic dispersion $\sigma^{2}_{\Zmatsc}\Imat$: 
\begin{inparaenum}
\item
Gaussian and
\item
$t$-distributed noise with $\lambda=0.2$ degrees of freedom.
\end{inparaenum}
Notice that unlike Gaussian noise, $t$-distributed noise is heavy tailed and produces outliers. 

The Tukey bi-square M-estimator we compared to is one that minimizes the following objective function $J_{\rho}\left(\thetavec\right)\triangleq\sum_{n=1}^{N}\rho\left(\frac{\|\Xmat_{n}-\Amat\alphavecsm\|}{\hat{\sigma}}\right)$, where $\rho\left(r\right)\triangleq1-\left(1-\left(\frac{r}{c}\right)^{2}\right)^{3}\mathbbm{1}_{[0,c]}\left(\left|r\right|\right)$ is Tukey's bi-square loss function, $c$ is a tuning constant that controls the asymptotic relative efficiency (ARE) \cite{Serfling} of the estimate relative to the CRLB under nominal Gaussian distribution, and $\mathbbm{1}_{[0,c]}\left(\cdot\right)$ denotes the indicator function of the closed interval $[0,c]$. The robust scale parameter estimate $\hat{\sigma}\triangleq\sqrt{\frac{1}{p}\sum_{k=1}^{p}\hat{\sigma}^{2}_{X_{k}}}$, where $\hat{\sigma}^{2}_{X_{k}} = \gamma^{2}[(\textrm{MAD}(\{\textrm{Re}(X_{k,n})\}_{n=1}^{N}))^{2} + (\textrm{MAD}(\{\textrm{Im}(X_{k,n})\}_{n=1}^{N}))^{2}]$, $\gamma\triangleq{1}/\textrm{erf}^{-1}(3/4)$, is a robust median absolute deviation (MAD) estimate of variance \cite{HuberRob}. The constant $\gamma$ ensures consistency of the scale estimate under normally distributed data \cite{HuberRob}. Similarly to the analysis of M-estimators of location with preliminary estimate of scale \cite{HuberRob} it can be shown that the ARE of the considered Tukey bi-square M-estimator, defined as the ratio between the traces of the CRLB and the asymptotic MSE under Gaussian distribution, is given by ${\rm{ARE}}\left(c\right)=\frac{\left(2c^{-2}p^{-1}{\rm{E}}\left[\left(1-\left(R/c\right)^{2}\right)R^{2}\mathbbm{1}_{[0,c]}\left(R\right);P_{R}\right]-{\rm{E}}\left[\left(1-\left(R/c\right)^{2}\right)^{2}\mathbbm{1}_{[0,c]}\left(R\right);P_{R}\right]\right)^{2}}{p^{-1}{\rm{E}}\left[\left(1-\left(R/c\right)^{2}\right)^{4}R^{2}\mathbbm{1}_{[0,c]}\left(R\right);P_{R}\right]}$, where $\sqrt{2}R$ is a chi distributed random variable with $p$ degrees of freedom. Using this formula, the parameter $c$ was set to achieve ARE of $95\%$ in all simulation examples. For the considered observation vector dimensionality $p=10$ we obtained $c\approx6.2$. By equating the first-order derivative of the objective function $J_{\rho}\left(\thetavec\right)$ to zero, Tukey's bi-weight M-estimator is derived by setting $\hat{\thetavec}=\left[{\rm{Re}}\{{\hat{\alphavec}}\}^{T},{\rm{Im}}\{{\hat{\alphavec}}\}^{T}\right]^{T}$, where $\hat{\alphavec}$ is the numerical solution of the equation ${\alphavec}={\left(\Amat^{H}\Amat\right)^{-1}\Amat^{H}}\frac{\sum_{n=1}^{N}\Xmat_{n}w\left(\Xmat_{n},{\alphavecsm}\right)}{\sum_{n=1}^{N}w\left(\Xmat_{n},{\alphavecsm}\right)}$ obtained by fixed-point iteration, and the weight function $w\left(\Xmat,\alphavec\right)\triangleq(1-({\|\Xmat-\Amat\alphavec\|}/{c\hat{\sigma}})^{2})^{2}\mathbbm{1}_{[0,c]}\left({\|\Xmat-\Amat\alphavec\|}/{c\hat{\sigma}}\right)$. Here, the fixed-point iteration was initialized by $\hat{\alphavec}_{\rm{init}}=\left(\Amat^{H}\Amat\right)^{-1}\Amat^{H}\hat{\muvec}_{{\rm{med}}}$, where $\hat{\muvec}_{{\rm{med}}}$ is the median location estimator. The maximum number of iterations and the stopping criterion in the fixed-point iteration were set to 100 and $\|\hat{\alphavec}_{l}-\hat{\alphavec}_{l-1}\|/\|\hat{\alphavec}_{l-1}\|<10^{-6}$, respectively, where $l$ denotes an iteration index. 

We note that the MLE under the $t$-distributed noise model with $\lambda$ degrees of freedom is obtained using the same iterative procedure as in Tukey's bi-square M-estimator with a different weight function $w\left(\Xmat,\alphavec\right)\triangleq(1+2\|\Xmat-\Amat\alphavec\|^{2}/{\lambda\sigma^{2}_{\Zmatsc}}))^{-1}$. Initialization and stopping criterion were identical to those of Tukey's bi-square M-estimator discussed above.

Notice that unlike Tukey's bi-square M-estimator and the MLE for $t$-distributed noise, the proposed MT-GQMLE (\ref{MTQMLE_GAIN_2}) does not involve iterative numerical optimization procedure that may require proper initialization and stopping criterion.

For each noise type we performed two simulations. In the first simulation example, we compared the traces of the asymptotic MSE matrix (\ref{AMSE_GAIN}) and its empirical estimate (\ref{PredMSEGain}) as a function of $\omega$. The empirical asymptotic MSE was obtained from a single realization of $N=1000$ i.i.d. snapshots. The signal-to-noise-ratio (SNR), defined here as $\rm{SNR}\triangleq{{\rm{tr}}\left[\Amat^{H}\Amat\right]}/{\sigma^{2}_{\Zmatsc}}$, was set to $0$ [dB] and $-10$ [dB] for the Gaussian and $t$-distributed noise, respectively. Observing Figs. \ref{Fig1a} and \ref{Fig2a} one sees that the compared quantities are very close. This indicates that the empirical asymptotic MSE can be reliably used for optimal choice of the MT-function parameter, as discussed in subsection \ref{OptChoice}.

In the second simulation example, we compared the traces of the empirical, asymptotic (\ref{AMSE_GAIN}) and empirical asymptotic (\ref{PredMSEGain}) MSEs of the MT-GQMLE to the traces of the empirical MSEs obtained by the other compared algorithms versus $\rm{SNR}$ and sample size $N$. The reported results represent empirically averaged performance with averaging over $1000$ Monte-Carlo simulations. The performance versus $\rm{SNR}$ was evaluated for $N=1000$ i.i.d. observations. The performance versus sample size were evaluated for $\rm{SNR}=0$ [dB] for both Gaussian and $t$-distributed noise. The optimal Gaussian MT-function parameter $\omega_{\rm{opt}}$ was obtained by minimizing the trace of (\ref{PredMSEGain}) over $K_{\Omega}=30$ equally spaced grid points of the range $\Omega=\left[1,30\right]$. Observing Figs. \ref{Fig1b} and \ref{Fig1c}, one can notice that all compared algorithms perform similarly when the noise is Gaussian. The slight performance gap of Tukey's bi-square M-estimator stems from the fact that its ARE is $95\%$, i.e., it is not asymptotically efficient. In Figs. \ref{Fig2b} and \ref{Fig2c} one sees that for $t$-distributed noise, the proposed MT-GQMLE (\ref{MTQMLE_GAIN_2}) outperforms the standard GQMLE and Tukey's bi-square M-estimator, and performs similarly to the MLE that unlike the propose estimator requires complete knowledge of the likelihood function. Here, the performance gap of the Tukey's bi-square M-estimator stems from the inconsistency of the scale estimate $\hat{\sigma}$ for non-Gaussian data. 

The averaged running times of the compared estimators are reported in Table \ref{Table1} for sample size $N=1000$ and ${\rm{SNR}}=0$ [dB] for both Gaussian and $t$-distributed noise. We note that for a fixed sample size, the averaged running times did not vary across the $\rm{SNR}$. Table \ref{Table1} indicates that in this example the observed performance gain of the MT-GQMLE does not come at the expense of significantly increased computational burden. To give a more comprehensive computational complexity assessment a general asymptotic computational load (ACL) analysis (for the considered estimation problem) is reported in Table \ref{Table2}. One sees that the ACL of all compared algorithms is cubic in the dimension $m$ of the parameter vector $\thetaveczero$. Also notice that unlike the MT-GQMLE, whose ACL is quadratic in the dimension $p$ of the observation vector and linear in the sample size $N$, the ACLs of Tukey's estimator and the non-Gaussian MLE are linear in $p$ and non-linear in $N$.
%%%
\begin{figure}[H] 
  \begin{center}
    {{\subfigure[]{\label{Fig1a}\includegraphics[scale = 0.445]{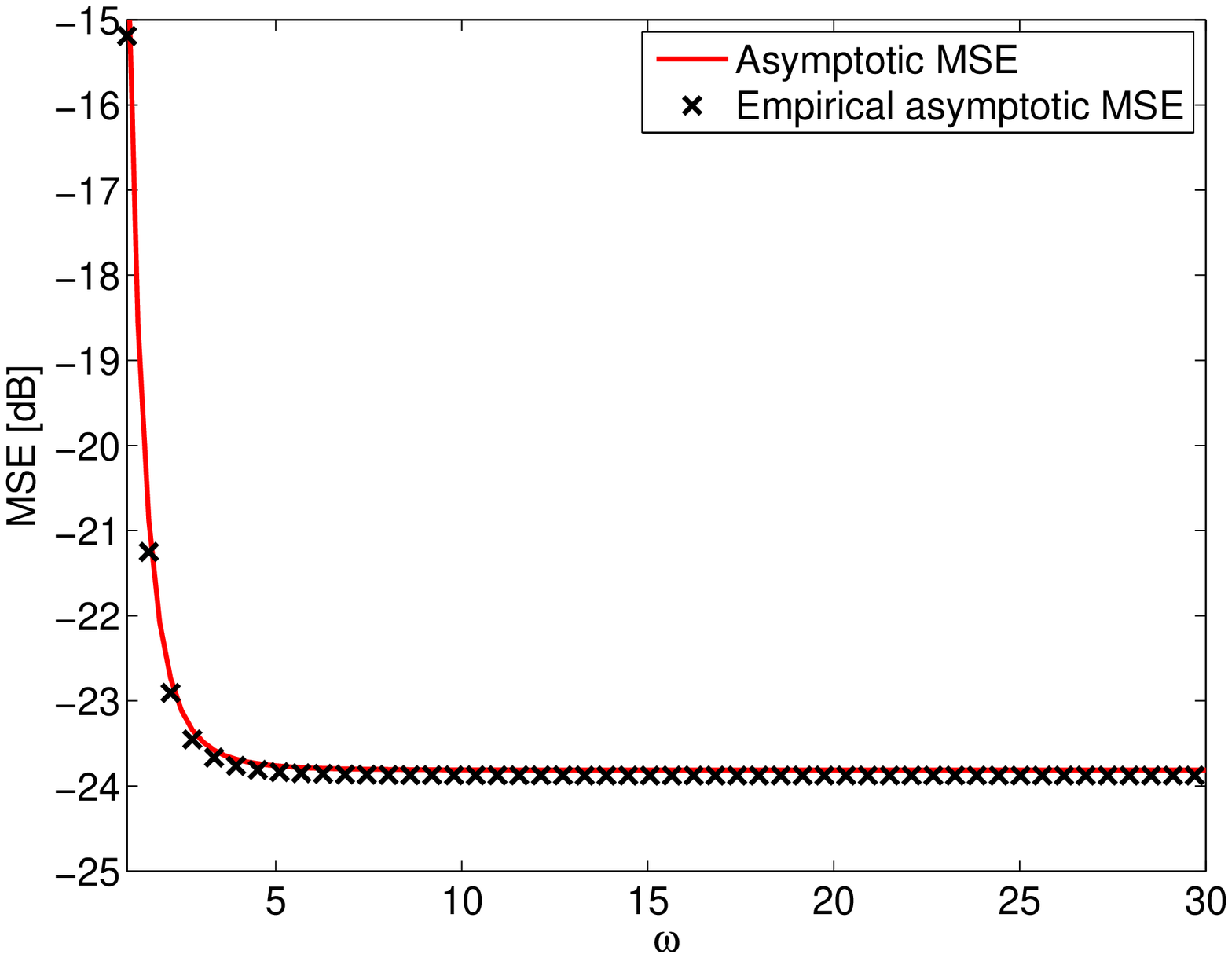}}}}
    {{\subfigure[]{\label{Fig1b}\includegraphics[scale = 0.445]{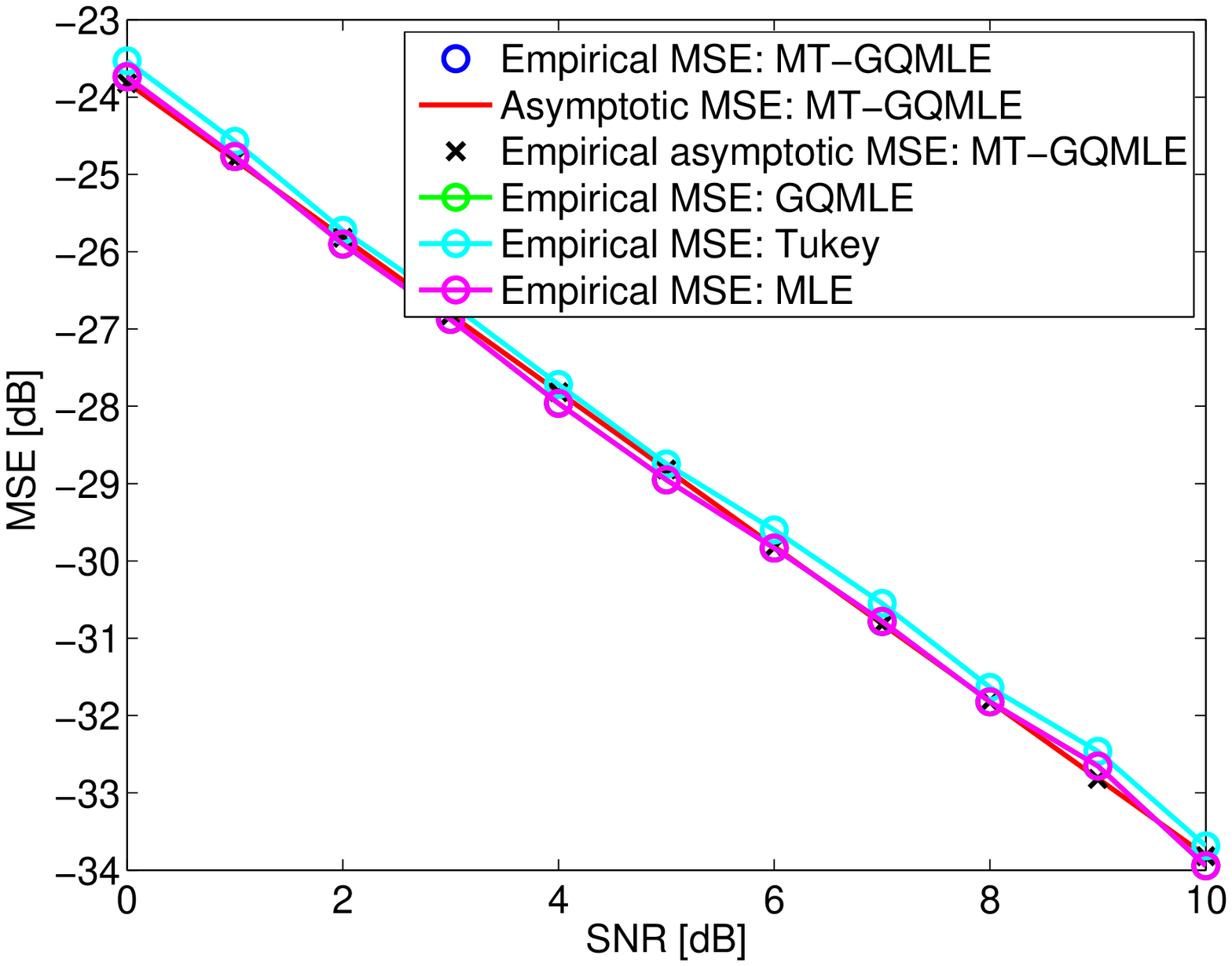}}}}
    {{\subfigure[]{\label{Fig1c}\includegraphics[scale = 0.445]{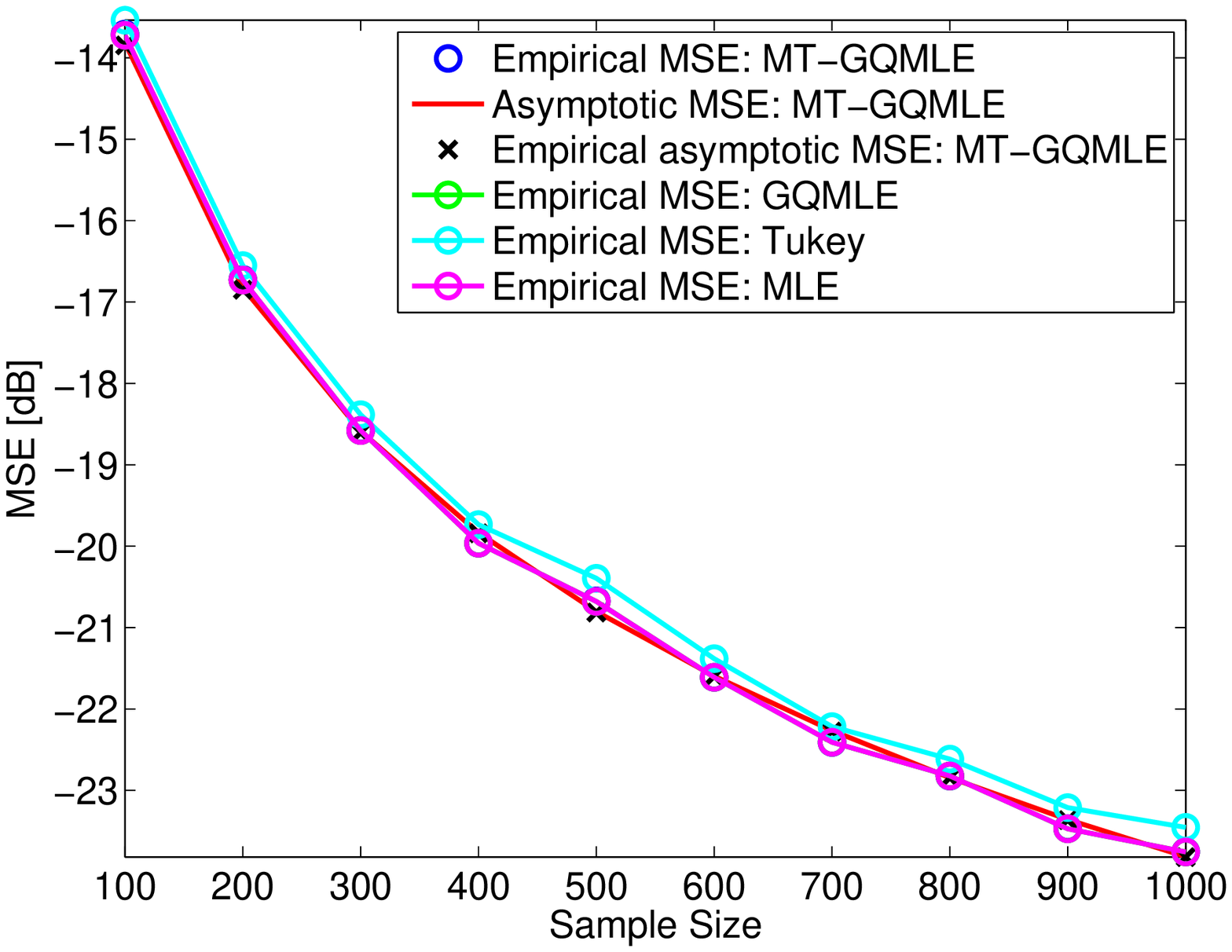}}}}
    %{{\subfigure[]{\label{Fig1c}\includegraphics[scale = 0.445]{Figures/L_REG_GAUSS_RUN_TIME.eps}}}}
      \end{center}  
  \caption{\textbf{Linear regression in Gaussian noise:}
   (a) Traces of the asymptotic MSE matrix (\ref{AMSE_GAIN}) and its empirical estimate (\ref{PredMSEGain}) versus the width parameter $\omega$ of the Gaussian MT-function (\ref{GaussMTFuncOrth}). Notice that due to the consistency of (\ref{PredMSEGain}) the compared quantities are close. (b) + (c) Traces of the empirical, asymptotic (\ref{AMSE_GAIN}) and empirical asymptotic (\ref{PredMSEGain}) MSEs of the MT-GQMLE versus {\rm{SNR}}  (b) and sample size (c)  as compared to the GQMLE, Tukey's bi-square M-estimator and the omniscient MLE. Notice that the compared algorithms perform similarly when the noise is Gaussian.}
\label{Fig1}
\end{figure}
%%%
\begin{figure}[H]
  \begin{center}
    {{\subfigure[]{\label{Fig2a}\includegraphics[scale = 0.445]{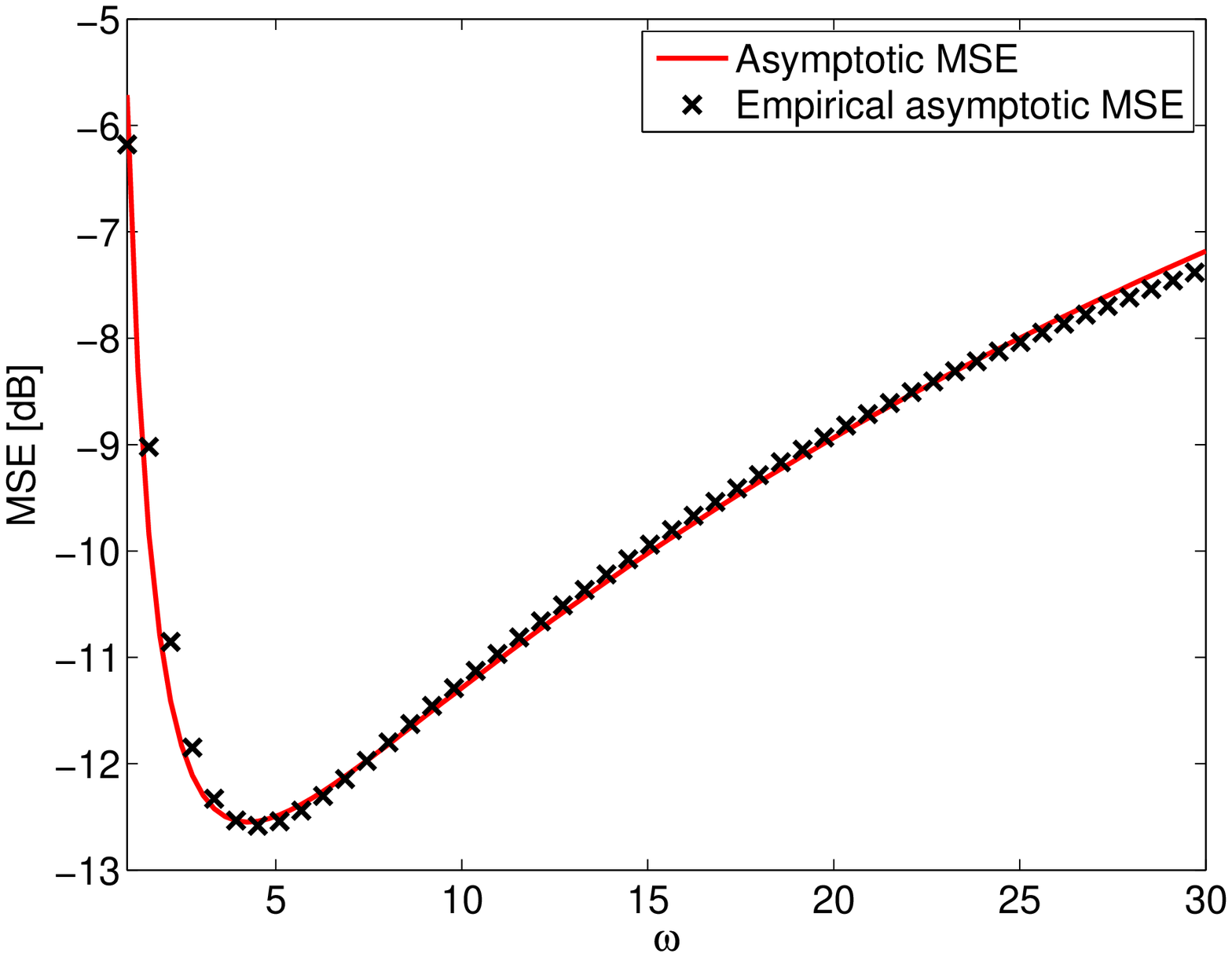}}}}
    {{\subfigure[]{\label{Fig2b}\includegraphics[scale = 0.445]{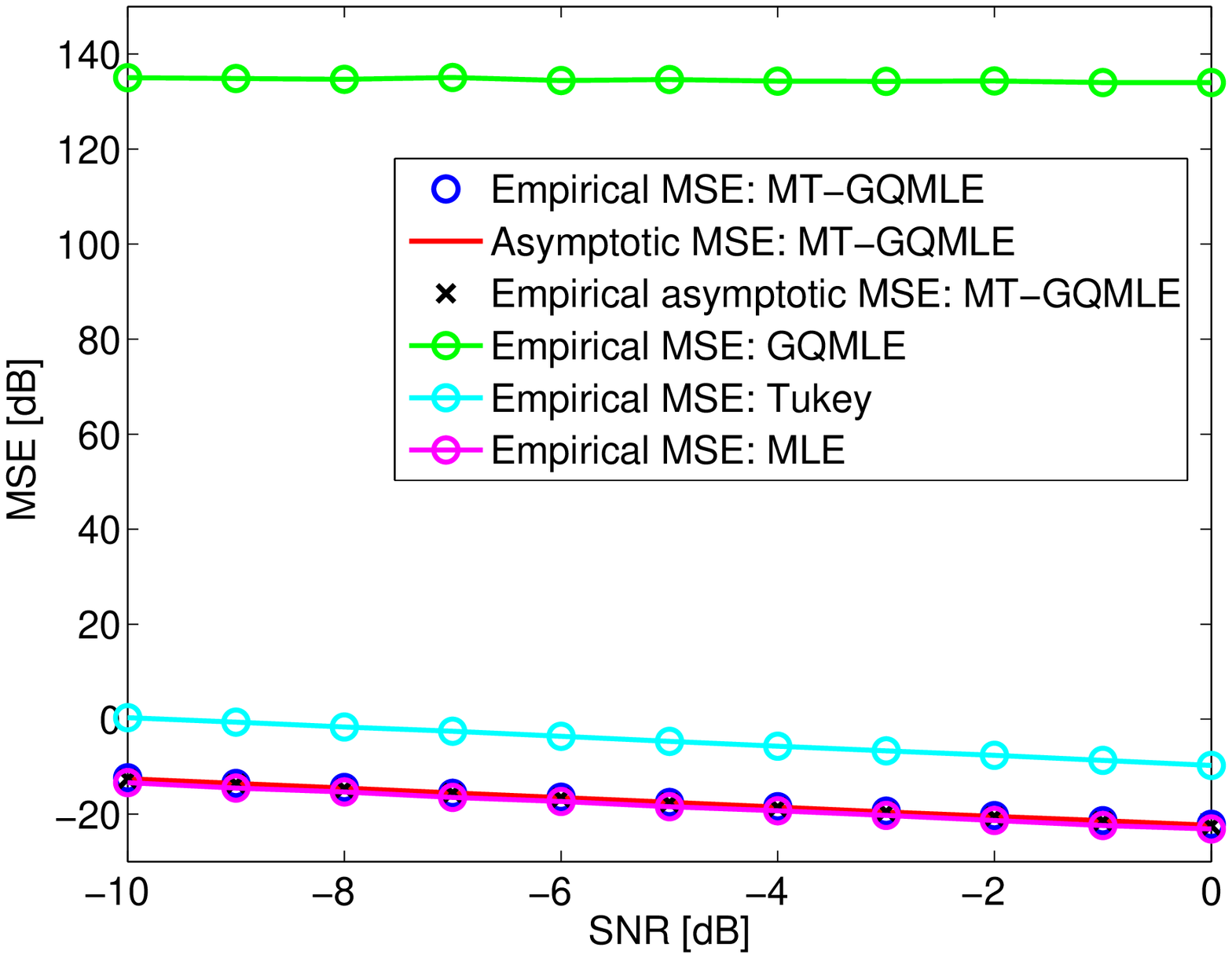}}}}
    {{\subfigure[]{\label{Fig2c}\includegraphics[scale = 0.445]{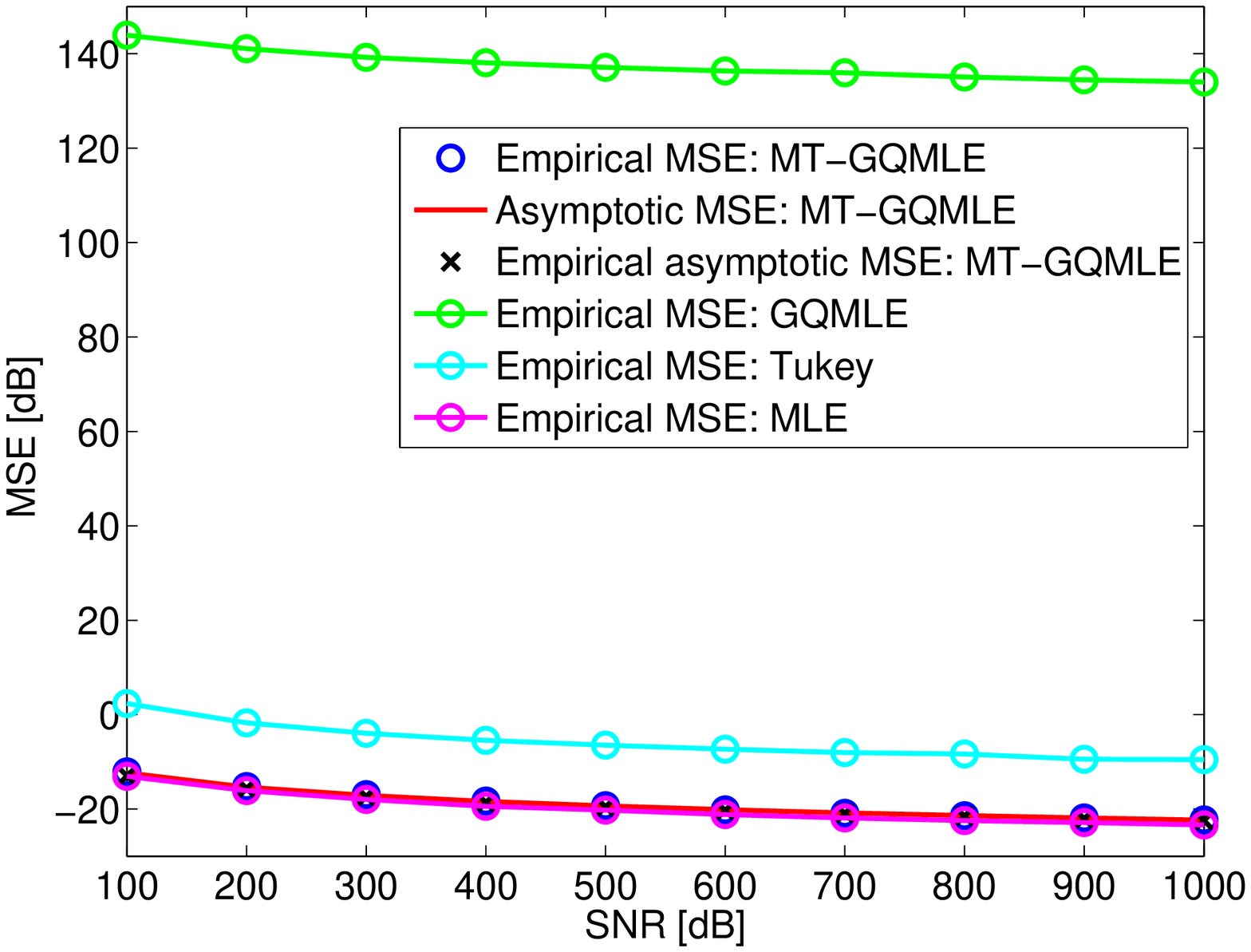}}}}
    %{{\subfigure[]{\label{Fig2c}\includegraphics[scale = 0.445]{Figures/L_REG_T_RUN_TIME.eps}}}}
      \end{center}  
  \caption{\textbf{Linear regression in non-Gaussian noise:}
   (a) Traces of the asymptotic MSE matrix (\ref{AMSE_GAIN}) and its empirical estimate (\ref{PredMSEGain}) versus the width parameter $\omega$ of the Gaussian MT-function (\ref{GaussMTFuncOrth}). Notice that due to the consistency of (\ref{PredMSEGain}) the compared quantities are close. (b) + (c) Traces of the empirical, asymptotic (\ref{AMSE_GAIN}) and empirical asymptotic (\ref{PredMSEGain}) MSEs of the MT-GQMLE versus {\rm{SNR}}  (b) and sample size (c)  as compared to the GQMLE, Tukey's bi-square M-estimator and the omniscient MLE. Notice that the proposed MT-GQMLE (\ref{MTQMLE_GAIN_2}) outperforms the standard GQMLE and Tukey's bi-square M-estimator, and performs similarly to the MLE that unlike the propose estimator requires complete knowledge of the likelihood function.}
\label{Fig2}
\end{figure}
%%%
\begin{table}[H]
\caption{\textbf{Linear Regression:} Averaged running times for parameter vector dimension $m=4$, observation vector dimension $p=10$, sample size $N=1000$ and $\rm{SNR}=0$ [dB].}
\begin{center}
\begin{tabular}{|l|l|l|}\hline
Estimator & Running Time [Sec] - Gaussian Noise & Running Time [Sec] - $t$-distributed Noise\\ \hline
GQMLE & 2e-4 & 2e-4 \\ \hline
MT-GQMLE & 6e-3 & 6e-3  \\ \hline
Tukey & 2e-3 & 2e-3 \\ \hline
MLE   & 2e-4 & 1e-3 \\ \hline
\end{tabular}
\end{center}
\label{Table1}
\end{table}%
%%%
%%%
%\begin{table}[H]
%\caption{Linear regression: asymptotic computational load. Notations: $m$ and $p$ are the dimensions of the parameter and observation vectors, respectively. $N$ denotes the samples size. 
%$K_{\Omega}$ denotes number of grid points of the width parameter space $\Omega$ corresponding to the Gaussian MT-function (\ref{GaussMTFuncOrth}). $L_{\rm{T}}$ and $L_{\rm{M}}$ are the number of fixed-point iterations used in Tukey's and the MLE estimators, respectively.}
%\begin{center}
%\begin{tabular}{|l|l|}\hline
%Estimator & Asymptotic computational load [flops] \\ \hline
%GQMLE & $O\left(m^{3} + pm^{2} + pN\right)$ \\ \hline
% %
%MT-GQMLE & \begin{tabular}{l} MT-function optimization: $O( m^3 + pm^2 + p^2(m + N) + Nm(p + m)K_{\Omega})$ \\  
%Estimation: $O\left(m^{3} + pm^{2} + p^{2}(m + N) \right)$ \end{tabular}  \\ \hline
%%
%Tukey & \begin{tabular}{l} Initialization (via median location estimator): $O\left(m^{3} + pm^{2} + pN\log_{2}{N}\right)$ \\  
%Estimation: $O\left(m^{3} + p(m^{2} + N\log_{2}{N}) + p(m+N)L_{\rm{T}}\right)$ \end{tabular}
%%
%\\ \hline
%%
%MLE - Gaussian Noise &  $O\left(m^{3} + pm^{2} + pN\right)$ \\ \hline
%%
%MLE - $t$-distributed Noise &  \begin{tabular}{l} Initialization (via median location estimator): $O\left(m^{3} + pm^{2} + pN\log_{2}{N}\right)$ \\  
%Estimation: $O\left(m^{3} + pm^{2} + p(m+N)L_{\rm{M}}\right)$ 
%\end{tabular} \\ \hline
%\end{tabular}
%\end{center}
%\label{Table2}
%\end{table}%
\begin{table}[H]
\caption{\textbf{Linear regression:} asymptotic computational load. Notations: $m$ and $p$ are the dimensions of the parameter and observation vectors, respectively. $N$ denotes the samples size. 
$K_{\Omega}$ denotes number of grid points of the $\Omega$-axis (the width parameter space of the Gaussian MT-function (\ref{GaussMTFuncOrth})). $L_{\rm{T}}$ and $L_{\rm{M}}$ are the number of fixed-point iterations used in Tukey's and the non-Gaussian MLE estimators, respectively.}
\begin{center}
\begin{tabular}{|l|l|}\hline
Estimator & Asymptotic computational load [flops] \\ \hline
GQMLE & $O\left(m^{3} + p(m^{2} + N)\right)$ \\ \hline
MT-GQMLE & \begin{tabular}{l} MT-function optimization: $O( m^3 + pm^2 + p^2(m + N) + Nm(p + m)K_{\Omega})$ \\  
Estimation: $O\left(m^{3} + pm^{2} + p^{2}(m + N) \right)$ \end{tabular}  \\ \hline
Tukey & \begin{tabular}{l} Initialization (via median location estimator): $O\left(m^{3} + p(m^{2} + N\log_{2}{N})\right)$ \\  
Estimation: $O\left(m^{3} + p(m^{2} + N\log_{2}{N} + (m+N)L_{\rm{T}})\right)$ \end{tabular}
\\ \hline
MLE - Gaussian Noise &  $O\left(m^{3} + p(m^{2} + N)\right)$ \\ \hline
MLE - $t$-distributed Noise &  \begin{tabular}{l} Initialization (via median location estimator): $O\left(m^{3} + p(m^{2} + N\log_{2}{N})\right)$ \\  
Estimation: $O\left(m^{3} + p(m^{2} + (m+N)L_{\rm{M}})\right)$ 
\end{tabular} \\ \hline
\end{tabular}
\end{center}
\label{Table2}
\end{table}%

%%%
%%%%%%%%%%%%%%%%%%%%%%%%%%%%%%%%%%%%%%%%%%%%%%%%%%%%%%%%%%%%%%%%%%%%%%%%%%%%%%%%%%%%%%%%%%%%%%%%%%%%%%%%%%%%%%%%%%%%%
%\newpage
\subsection{Source localization}
\label{SourceLoc} 
Here, we illustrate the use of the proposed MT-GQMLE (\ref{PropEst}) for source localization. We consider a uniform linear array (ULA) of $p$ sensors with half wave-length spacing that receive a signal generated by a narrowband far-field point source with azimuthal angle of arrival (AOA) $\theta_{0}$. Under this model the array output satisfies \cite{Viberg}:
\begin{equation}
\label{ArrayModel}
\Xmat_{n}=S_{n}\avec\left(\theta_{0}\right)+\Wmat_{n},\hspace{0.2cm}n=1,\ldots,N,
\end{equation}
where $n$ is a discrete time index, $\Xmat_{n}\in\Csp^{p}$ is the vector of received signals, $S_{n}\in\Csp$ is the emitted signal, $\avec\left(\theta_{0}\right)\triangleq\left[1,
\exp\left({-i\pi\sin\left(\theta_{0}\right)}\right),\ldots,\exp\left({-i\pi(p-1)\sin\left(\theta_{0}\right)}\right)\right]^{T}$ is the steering vector of the array toward direction $\theta_{0}$ and $\Wmat_{n}\in\Csp^{p}$ is an additive noise. We assume that the following conditions are satisfied: 
\begin{inparaenum}[1)]
\item
$\theta_{0}\in\Theta=\left[-\pi/2,\pi/2-\delta\right]$ where $\delta$ is a small positive constant,
\item
the emitted signal is symmetrically distributed about the origin,
\item
$S_{n}$ and $\Wmat_{n}$ are statistically independent and first-order stationary, and
\item
the noise component is spherically contoured with the stochastic representation (\ref{CompGaussOrth}).
\end{inparaenum}

In order to gain robustness against outliers, as well as sensitivity to higher-order moments, we specify the MT-function in the zero-centred Gaussian family of functions parametrized by a width parameter $\omega$, i.e.,
\begin{equation}
\label{GaussMTFunc}
\uGausss\left(\xvec;\omega\right)\triangleq\exp\left(-{\left\|\xvec\right\|^{2}}/{\omega^{2}}\right),\hspace{0.2cm\omega\in\Rsp_{++}}.
\end{equation}
Using (\ref{MTMean}), (\ref{MTCovZ}), (\ref{ArrayModel}) and (\ref{GaussMTFunc}) it can be shown that  the MT-mean and MT-covariance under the transformed probability measure $Q^{(\uGausss)}_{\Xmatsc;\theta}$ are given by:
\begin{equation}
\label{MG}
\muvec^{(u_{\rm{G}})}_{\Xmatsc}\left(\theta;\omega\right)=\zerovec
\end{equation}
and 
\begin{equation}
\label{CG}
\bSigma^{(u_{\rm{G}})}_{\Xmatsc}\left(\theta;\omega\right)=r_{S}\left(\omega\right)\avec\left(\theta\right)\avec^{H}\left(\theta\right) + r_{\Wmatsc}\left(\omega\right)\Imat,
\end{equation} 
respectively, where $r_{S}\left(\omega\right)$ and $r_{\Wmatsc}\left(\omega\right)$ are some strictly positive functions of $\omega$.
Hence, by substituting (\ref{GaussMTFunc})-(\ref{CG}) into (\ref{ObjFun}) the resulting MT-GQMLE (\ref{PropEst}) is obtained by solving the following maximization problem: 
\begin{equation}
\label{QML_DOA}
\hat{\theta}_{u_{\rm{G}}}\left(\omega\right)=\arg\max\limits_{\theta\in\Theta}\avec^{H}\left(\theta\right)\hat{\Cmat}^{(u_{\rm{G}})}_{\Xmatsc}\left(\omega\right)\avec\left(\theta\right),
\end{equation}
where $\hat{\Cmat}^{(u_{\rm{G}})}_{\Xmatsc}\left(\omega\right)\triangleq\hat{\bSigma}^{\left(u_{\rm{G}}\right)}_{\Xmatsc}\left(\omega\right)
+\hat{\muvec}^{\left(u_{\rm{G}}\right)}_{\Xmatsc}\left(\omega\right)\hat{\muvec}^{\left(u_{\rm{G}}\right)H}_{\Xmatsc}\left(\omega\right)$. 
Notice that the MT-GQMLE (\ref{QML_DOA}) maximizes a measure-transformed version of the spatial spectrum related to Bartlett's beamformer \cite{Viberg}.

Under the considered settings, it can be shown that the conditions stated in Theorems \ref{ConsistencyTh}-\ref{EAMSETh} are satisfied. The resulting asymptotic MSE (\ref{AMSEN}) is given by:
\begin{equation}
\label{MSE_DOA}
C_{u_{\rm{G}}}\left(\theta_{0};\omega\right)=\frac{{\rm{E}}\left[\left(\nu^{4}\sigma^{4}_{\zvec} + \frac{\nu^{2}\sigma^{2}_{\zvec}\omega^{2}p}{2\nu^{2}\sigma^{2}_{\zvec}
+\omega^{2}}\left|S\right|^{2}\right)h\left(\sqrt{2p}S,\sqrt{2}\nu\sigma_{\zvec},\omega\right);P_{S,\nu}\right]}
{{{\rm{E}}^{2}}\left[p\left|S\right|^{2}h\left(\sqrt{p}S,\nu\sigma_{\zvec},\omega\right);P_{S,\nu}\right]}
\frac{6}{\pi^{2}\cos^{2}\left(\theta_{0}\right)\left(p^{2}-1\right)N},
\end{equation}
where 
$h\left(S,\nu,\omega\right)\triangleq\left(({\nu^{2}+\omega^{2})}/{\omega^{2}}\right)^{-p-2}\exp\left({-{\left|S\right|^{2}}/({\nu^{2}+\omega^{2}}})\right)$.
Notice that when the noise is Gaussian, i.e. the texture component $\nu=1$ in (\ref{CompGaussOrth}), $C_{u_{\rm{G}}}\left(\theta_{0};\omega\right)\rightarrow{C(\theta_{0})}$ as $\omega\rightarrow\infty$, where $C(\theta_{0})\triangleq\frac{6\sigma^{2}_{\zvec}\left(\sigma^{2}_{\zvec}+\sigma^{2}_{S}p\right)}{\sigma^{4}_{S}\pi^{2}\cos^{2}\left(\theta_{0}\right)p^{2}(p^{2}-1)}$ is the CRLB for estimating $\theta_{0}$ under the assumption that the signal and noise are jointly Gaussian and $\sigma^{2}_{S}\triangleq{E}\left[\left|S\right|^{2};P_{S}\right]$ is the variance of the signal. Again, since $\uGausss\left(\xvec;\omega\right)\rightarrow{1}$ as $\omega\rightarrow\infty$ this result verifies the conclusion following Proposition \ref{EffTh}, which states that when the observations are normally distributed the CRLB is achievable for non-zero constant MT-function. Furthermore, the empirical asymptotic MSE (\ref{EMSEN}) takes the form:
\begin{equation}
\label{PredMSE}
\hat{C}_{u_{\rm{G}}}(\hat{\theta}_{u_{\rm{G}}}\left(\omega\right);\omega)=
\frac{\sum_{n=1}^{N}\alpha^{2}(\Xmat_{n};\hat{\theta}_{u_{\rm{G}}}\left(\omega\right))u^{2}_{\rm{G}}(\Xmat_{n};\omega)}
{(\sum_{n=1}^{N}\beta(\Xmat_{n};\hat{\theta}_{u_{\rm{G}}}\left(\omega\right))u_{\rm{G}}(\Xmat_{n};\omega))^{2}},
\end{equation}
where $\alpha\left(\Xmat;\theta\right)\triangleq2{\rm{Re}}\left\{\dot{\avec}^{H}\left(\theta\right)\Xmat\Xmat^{H}\avec\left(\theta\right)\right\},$ 
$\beta\left(\Xmat;\theta\right)\triangleq{2{\rm{Re}}\left\{\ddot{\avec}^{H}\left(\theta\right)\Xmat\Xmat^{H}\avec\left(\theta\right)\\+|\dot{\avec}^{H}\left(\theta\right)\Xmat|^{2}\right\}},$
$\dot\avec\left(\theta\right)\triangleq{d\avec\left(\theta\right)}/{d\theta}$ and $\ddot\avec\left(\theta\right)\triangleq{d^{2}\avec\left(\theta\right)}/{d\theta^{2}}$. 

Next, we study the robustness of the MT-GQMLE (\ref{QML_DOA}) to outliers. Using (\ref{InfFuncMTQML}) the corresponding influence function is given by:
\begin{equation}
\label{IF_DOA}
{\rm{IF}}\left(\yvec;\thetavec_{0}\right)=\frac{{\rm{E}}\left[\left(1+\frac{\nu^{2}\sigma^{2}_{\zvec}}{\omega^{2}}\right)^{2}h\left(\sqrt{p}S,\nu\sigma_{\zvec},\omega\right);P_{S,\nu}\right]}
{{\rm{E}}\left[\left|S\right|^{2}h\left(\sqrt{p}S,\nu\sigma_{\zvec},\omega\right);P_{S,\nu}\right]}\times\frac{12{\rm{Re}}\left\{\dot{\avec}^{H}\left(\theta\right)\yvec\yvec^{H}\avec\left(\theta\right)\right\}\exp\left(-{\left\|\yvec\right\|^{2}}/{\omega^{2}}\right)}{\pi^{2}\cos^{2}\left(\theta_{0}\right)p^{2}\left(p^{2}-1\right)}.
\end{equation}
Notice that for any fixed width parameter $\omega$ the MT-function (\ref{GaussMTFunc}) satisfies the conditions stated in Proposition \ref{RobustnessConditions} for $\mathcal{C}=\Csp^{p}$. Thus, the influence function (\ref{IF_DOA}) is bounded and decays to zero as the outlier norm approaches infinity, i.e, the  MT-GQMLE (\ref{QML_DOA}) is B-robust and rejects large norm outliers.

In the following example, we consider a BPSK signal with power $\sigma^{2}_{S}$ impinging on a $4$-element ULA with half wavelength spacing from AOA $\theta_{0}=30^{\circ}$. We considered two types of noise distributions with zero location parameter and isotropic dispersion $\sigma^{2}_{\Zmatsc}\Imat$: 
\begin{inparaenum}
\item
Gaussian and
\item
heavy-tailed $K$-distributed noise with shape parameter $\lambda=0.75$.
\end{inparaenum}
%The signal-to-noise-ratio (SNR) is defined here as $\textrm{SNR}\triangleq{10}\log_{10}{\sigma^{2}_{S}}/{\sigma^{2}_{\Zmatsc}}$ and is used to index the estimation performance. 

Similarly to the gain estimation example, for each noise type we performed two simulations. In the first simulation example, we compared the the asymptotic MSE (\ref{MSE_DOA}) and its empirical estimate (\ref{PredMSE}) as a function of $\omega$. The empirical asymptotic MSE was obtained from a single realization of $N=5000$ i.i.d. snapshots. The $\rm{SNR}$, defined in this example as $\textrm{SNR}\triangleq{10}\log_{10}{\sigma^{2}_{S}}/{\sigma^{2}_{\Zmatsc}}$, was set to $-5$ [dB] and $-15$ [dB] for the Gaussian and $K$-distributed noise, respectively. Observing Figs. \ref{Fig3a} and \ref{Fig4a} one sees that due to the consistency of (\ref{PredMSE}), which follows from Theorem \ref{EAMSETh}, the compared quantities are very close. This illustrates the reliability of the empirical asymptotic MSE in optimal choice of the MT-function parameter, as discussed in subsection \ref{OptChoice}. 

In the second simulation example, we compared the empirical, asymptotic (\ref{MSE_DOA}) and empirical asymptotic (\ref{PredMSE}) MSEs of the MT-GQMLE (\ref{QML_DOA}) to the empirical MSEs obtained by the GQMLE \cite{Trees}, the omniscient MLE, the measure-transformed MUSIC (MT-MUSIC) \cite{Todros2} and the robust MUSIC generalizations based on the empirical sign-covariance (SGN-MUSIC) \cite{Visuri} and Tyler's scatter M-estimator (TYLER-MUSIC) \cite{Visa}. The width parameter of the Gaussian MT-function in the MT-MUSIC algorithm \cite{Todros2} was set to guarantee relative transform domain Fisher information loss (under nominal Gaussian distribution) of no more than 40$\%$. Initial conditions, maximum number of iterations and stopping criteria of the fixed point iterations used in the MT-MUSIC and TYLER-MUSIC algorithms are identical to those set in \cite{Todros2}. In all compared algorithms estimation of $\theta_{0}$ was carried out via maximization over $K_{\Theta}=10^{4}$ equally spaced grid points of the parameter space $\Theta$. The compared MSE performance are indexed by $\rm{SNR}$ and sample size $N$ with averaging over $1000$ Monte-Carlo simulations. The performance versus $\rm{SNR}$ were evaluated for $N=5000$ i.i.d. observations. The performance versus sample size was evaluated for $\rm{SNR}=0$ [dB] for the Gaussian noise and $\rm{SNR}=-25$ [dB] for the $K$-distributed noise. The optimal Gaussian MT-function parameter $\omega_{\rm{opt}}$ was obtained by minimizing (\ref{PredMSE}) over $K_{\Omega}=30$ equally spaced grid points of the range $\Omega=\left[1,30\right]$. Observing Figs. \ref{Fig3b} and \ref{Fig3c}, one can notice that when the noise is Gaussian the MLE that assumes full knowledge of the likelihood function outperforms the other compared algorithms which attain similar performance. In Figs. \ref{Fig4b} and \ref{Fig4c} one sees that for the $K$-distributed noise, the proposed MT-GQMLE (\ref{MTQMLE_GAIN_2}) outperforms the standard GQMLE and the robust MUSIC generalizations and attains estimation performance that are significantly closer to those obtained by the MLE that, unlike the MT-GQMLE, requires complete knowledge of the likelihood function. The performance gap from the robust MUSIC generalizations stems from the empirical MSE optimization performed by the MT-GQMLE.  

As in the regression example, the averaged running times of the compared estimators are reported in Table \ref{Table3} for sample size $N=5000$ and $\rm{SNR}=-10$ [dB] for both Gaussian and $K$-distributed noise. We note that for a fixed sample size, the averaged running times did not vary across $\rm{SNR}$. Observing Table \ref{Table3}, one sees that the computational burden of the MT-GQMLE is significantly lower than this of the MLE for both Gaussian and $K$-distributed noise. Furthermore, one can notice that, similarly to the linear regression example, the observed performance gain of the MT-GQMLE does not come at the expense of significantly increased running time as compared to the GQMLE and the robust MUSIC generalizations. To support these observations an asymptotic computational load (ACL) analysis (under the considered estimation problem) is reported in Table \ref{Table4}. One sees that unlike the MLE, the ACL of the MT-GQMLE is not affected by product of the sample size $N$ and the number of grid points $K_{\Theta}$ of the parameter space, which is quite large for the considered $N=5000$ and  $K_{\Theta}=10^{4}$. Also notice that when (\ref{PredMSE}) is minimized over a relatively small number of grid points $K_{\Omega}$ the ACL of the MT-GMLE is not significantly higher than this of the GQMLE (here $K_{\Theta}=30$). Finally, note that unlike the robust MUSIC generalizations, whose ACL is cubic in the number of sensors $p$, the ACL of the MT-GQMLE is quadratic in $p$.
%%%
%the significantly increased running time of the MLE arises from the fact that, unlike the MT-GQMLE, its ACL is affected by the product of the sample size $N$ and the number of grid points $K_{\Theta}$ of the parameter space, which can be quite large as in the considered simulation example. We note that although the ACL of the MLE remains the same under the Gaussian and the $K$-distributed noise, the actual running time of the MLE under the $K$-distributed noise is higher, as reported in Table \ref{Table3}, due to the fact that in this case its implementation involves $NK_{\Theta}$ evaluations of the modified Bessel function of the second kind that is more computationaly complex as compared to the exponential function used in the MLE under the Gaussian noise. 
\begin{figure}[H]
  \begin{center}
    {{\subfigure[]{\label{Fig3a}\includegraphics[scale = 0.445]{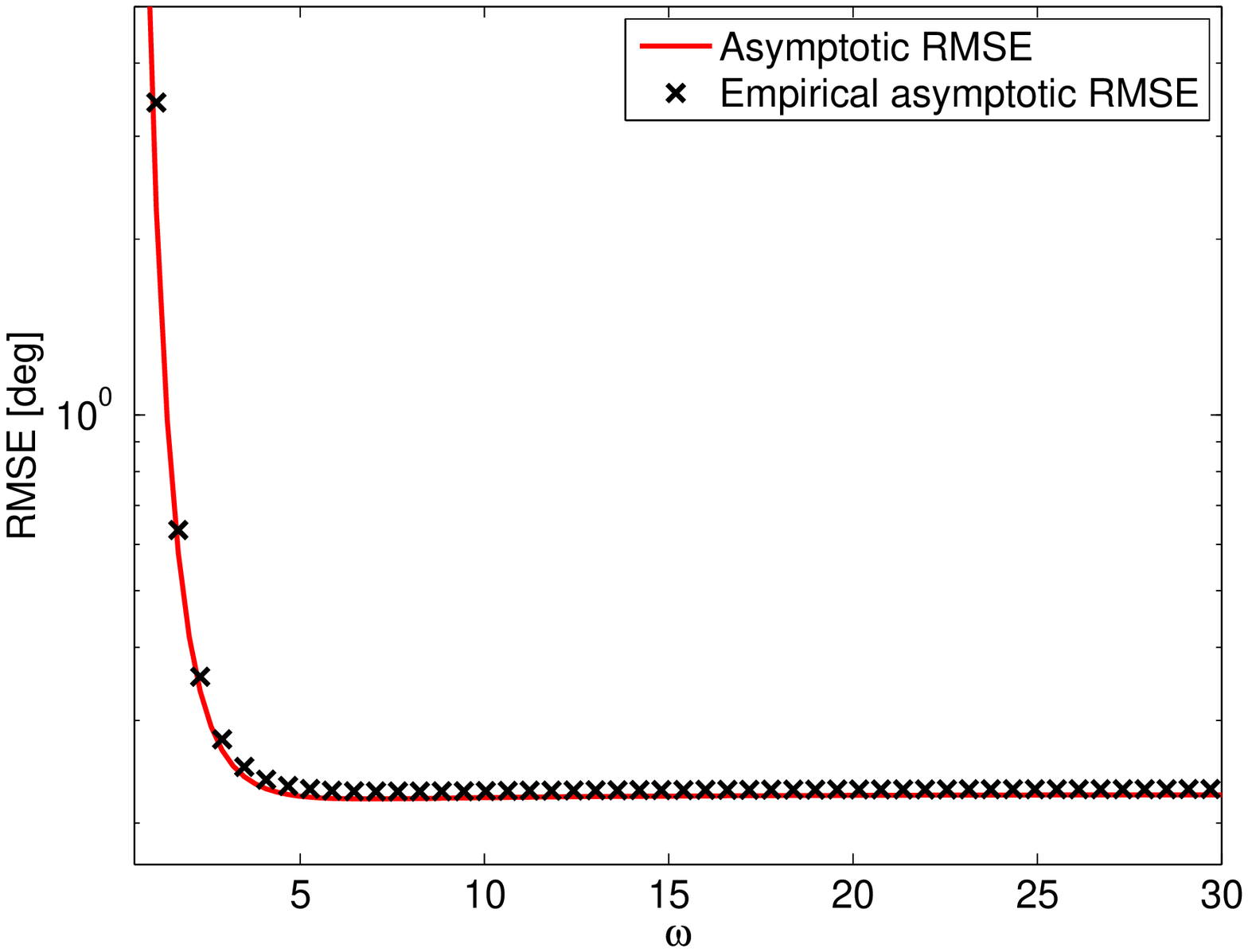}}}}
    {{\subfigure[]{\label{Fig3b}\includegraphics[scale = 0.445]{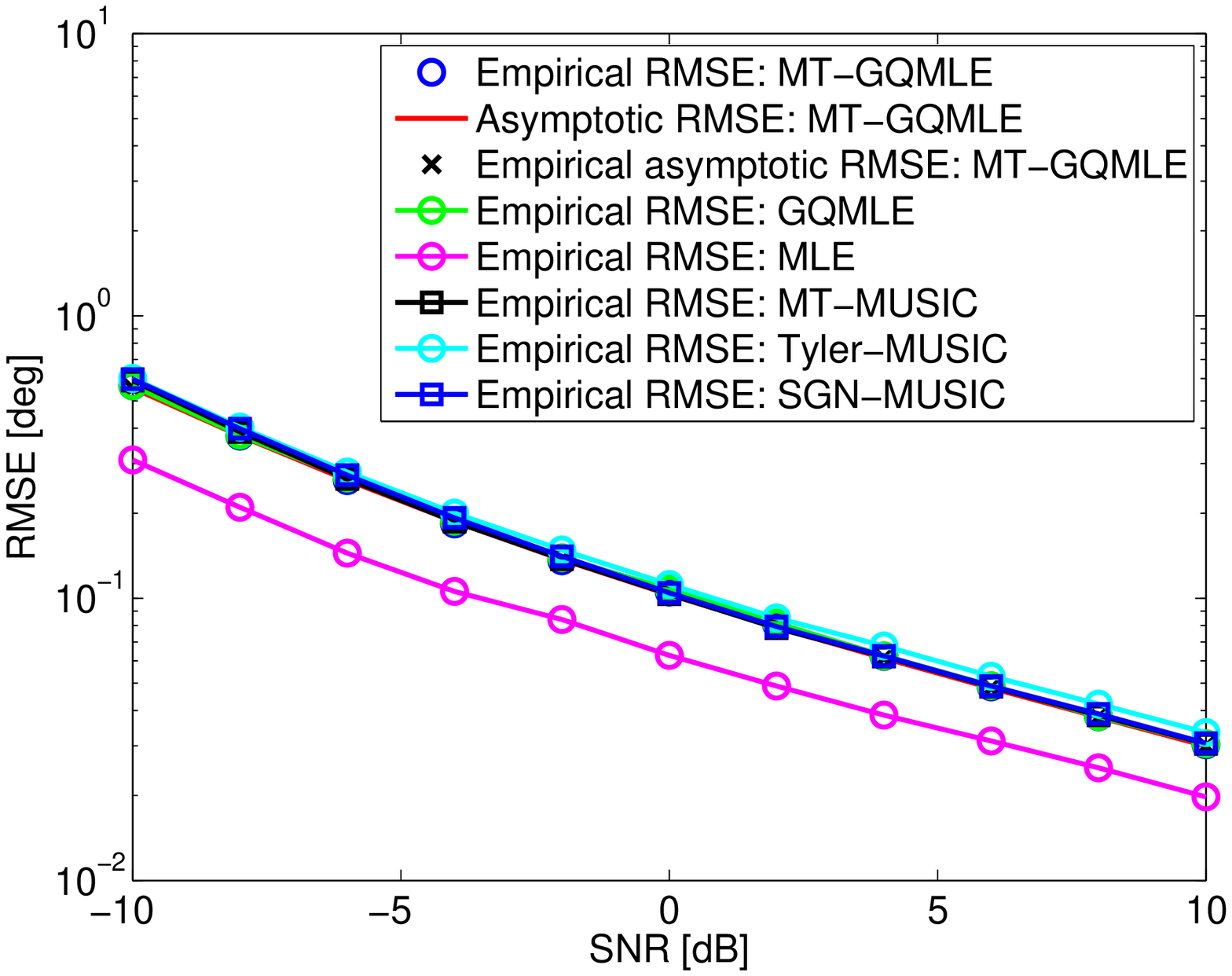}}}}
    {{\subfigure[]{\label{Fig3c}\includegraphics[scale = 0.445]{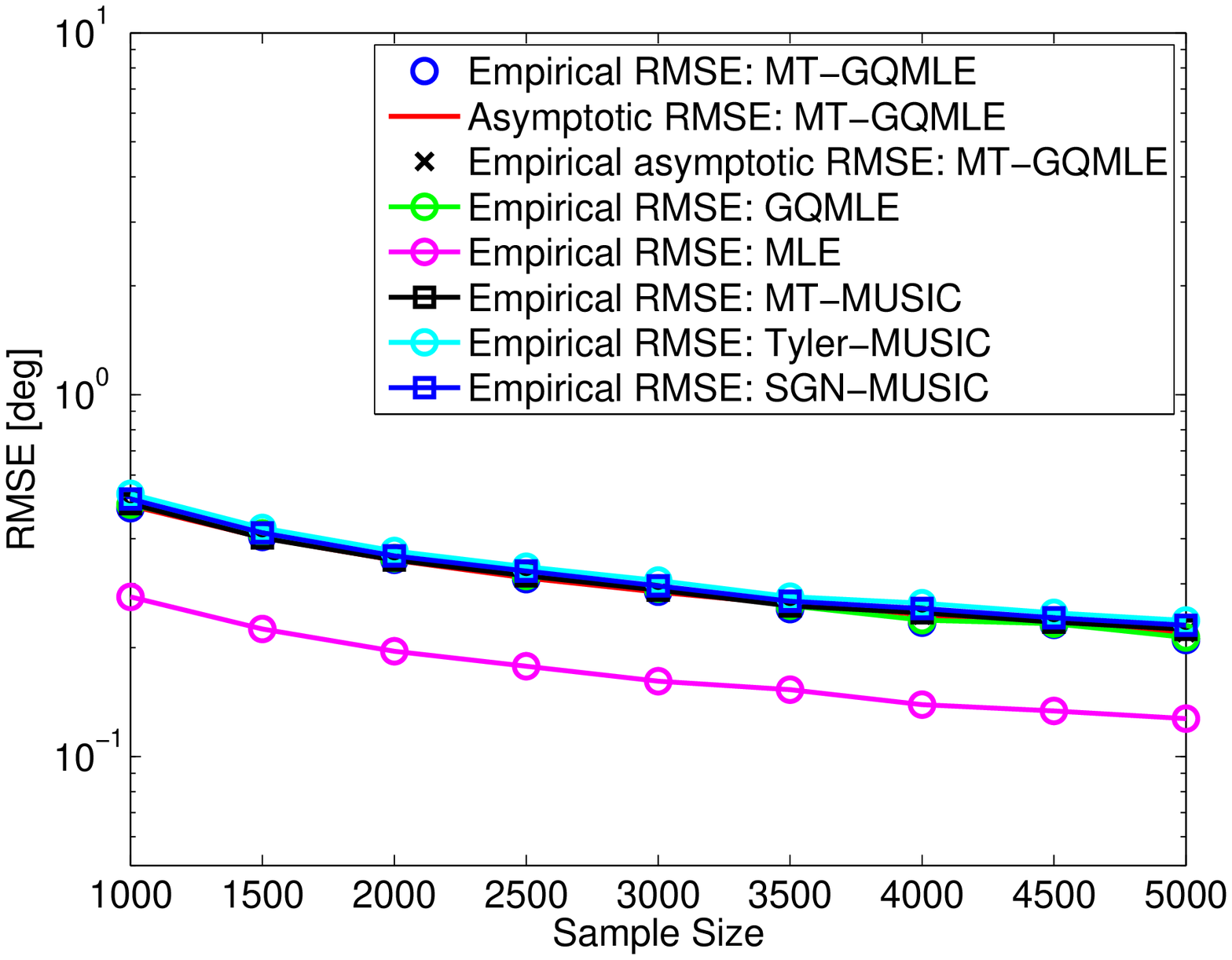}}}}
      \end{center}  
  \caption{\textbf{Source localization in Gaussian noise:}
   (a) Asymptotic RMSE (\ref{MSE_DOA}) and its empirical estimate (\ref{PredMSE}) versus the width parameter $\omega$ of the Gaussian MT-function (\ref{GaussMTFunc}). Notice that due to the consistency of (\ref{PredMSE}) the compared quantities are close. (b) + (c) The empirical, asymptotic (\ref{MSE_DOA}) and empirical asymptotic (\ref{PredMSE}) RMSEs of the MT-GQMLE versus SNR (b) and sample size (c) as compared to the GQMLE, MT-MUSIC, SGN-MUSIC, TYLER-MUSIC and the MLE. Notice that, excluding the omniscient MLE, all algorithms perform similarly when the noise is normally distributed.}
\label{Fig1}
\end{figure}
%%%
\begin{figure}[H]
  \begin{center}
    {{\subfigure[]{\label{Fig4a}\includegraphics[scale = 0.445]{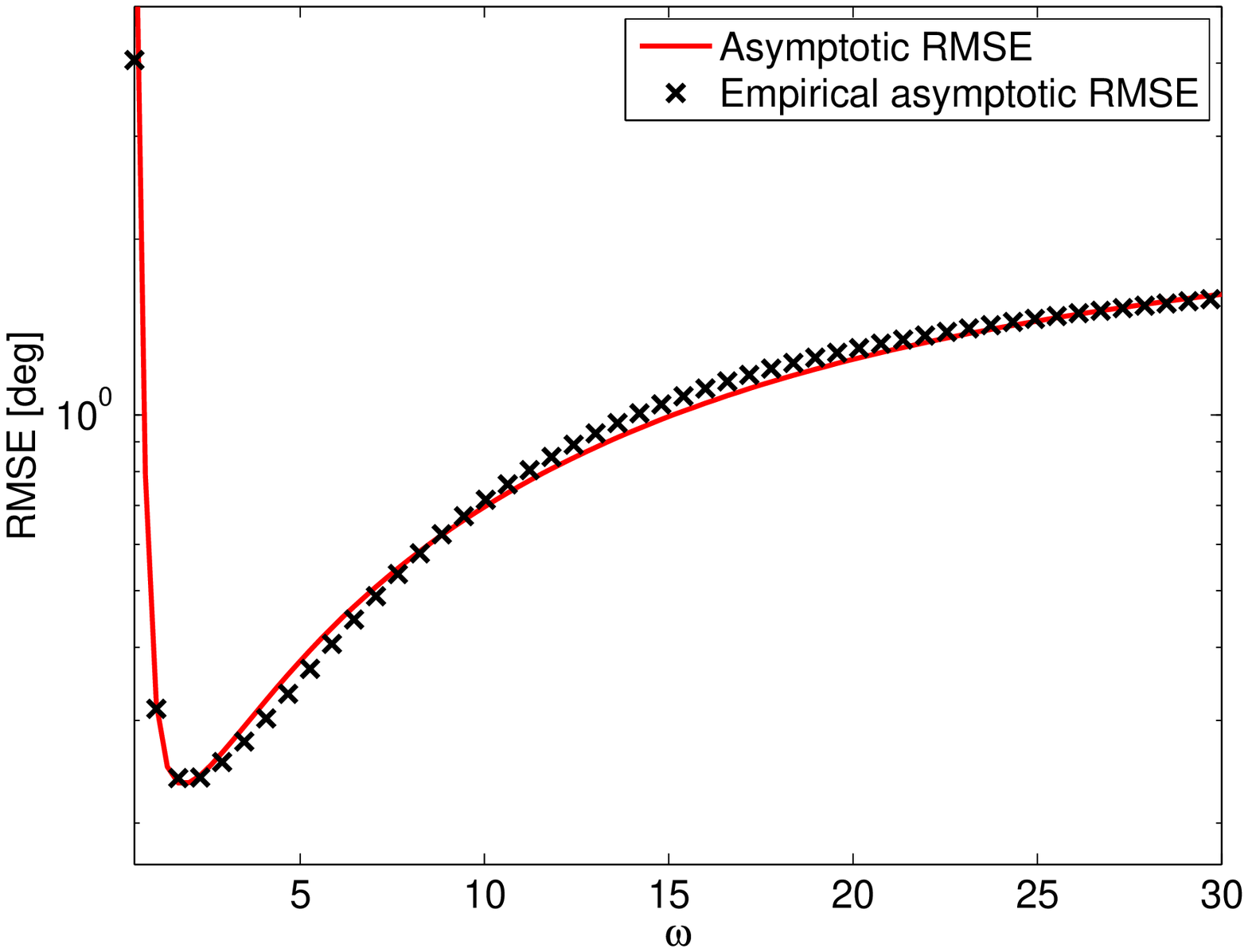}}}}
    {{\subfigure[]{\label{Fig4b}\includegraphics[scale = 0.445]{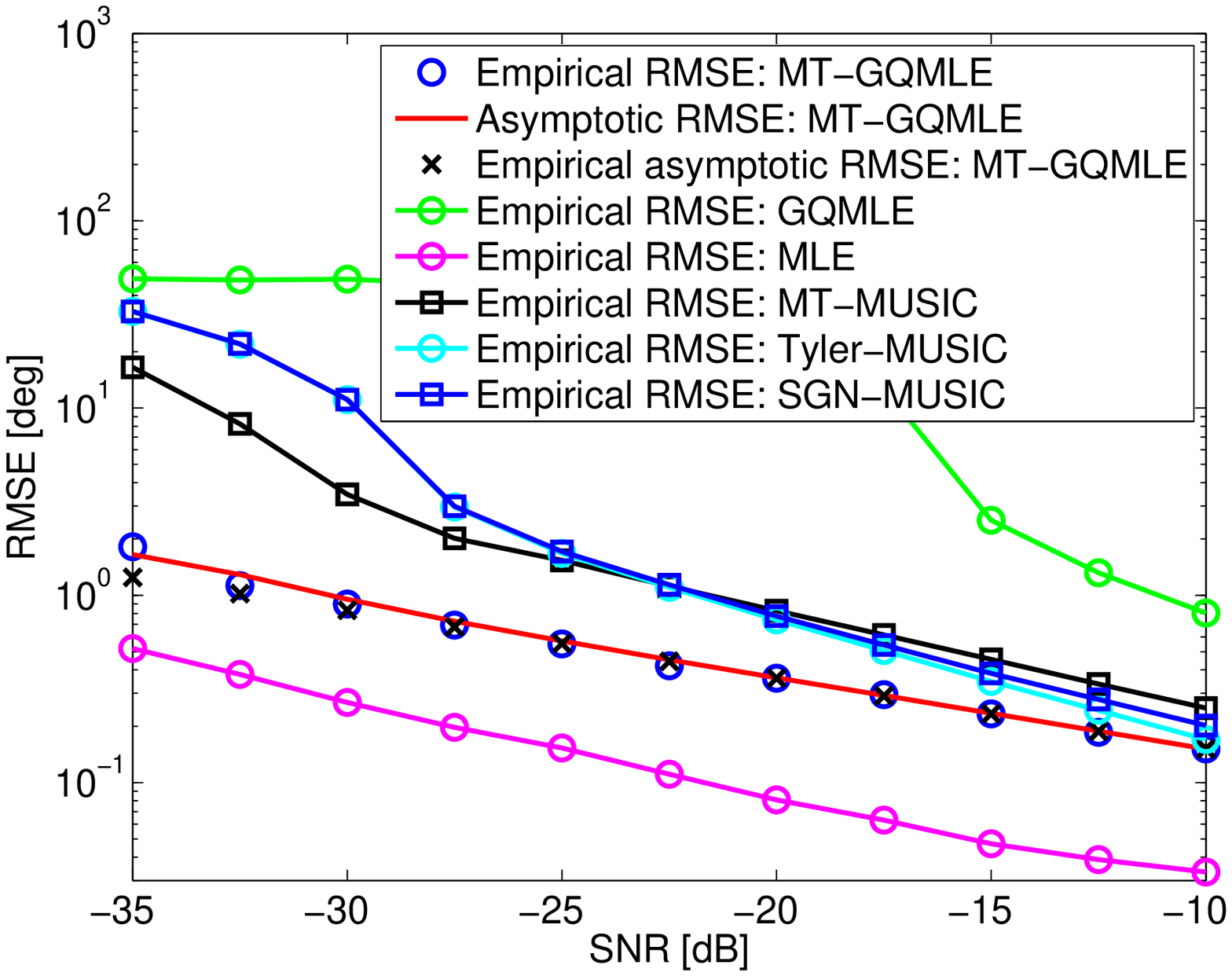}}}}
    {{\subfigure[]{\label{Fig4c}\includegraphics[scale = 0.445]{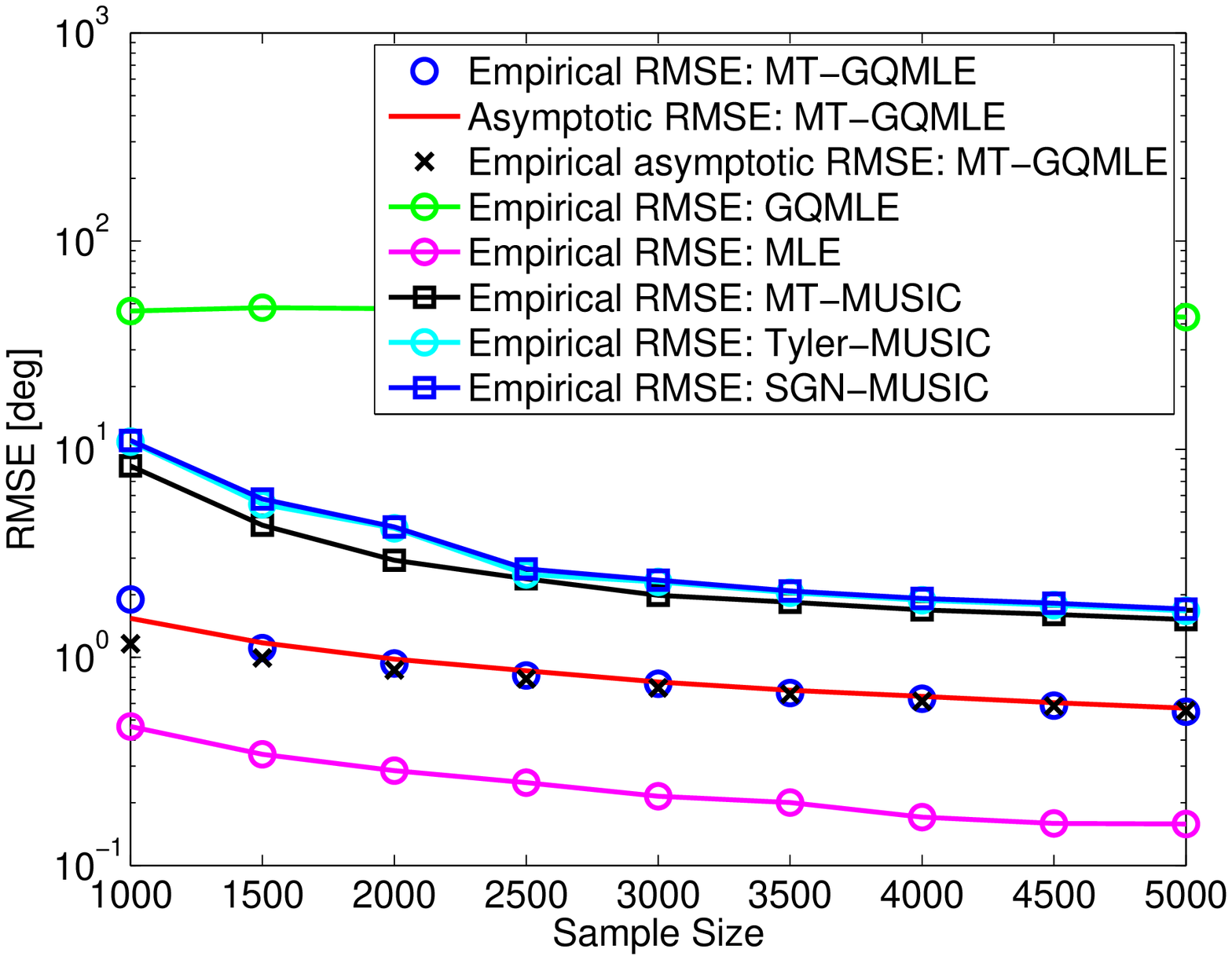}}}}
      \end{center}  
  \caption{\textbf{Source localization in non-Gaussian noise:}
   (a) Asymptotic RMSE (\ref{MSE_DOA}) and its empirical estimate (\ref{PredMSE}) versus the width parameter $\omega$ of the Gaussian MT-function (\ref{GaussMTFunc}). Notice that due to the consistency of (\ref{PredMSE}) the compared quantities are close. (b) + (c) The empirical, asymptotic (\ref{MSE_DOA}) and empirical asymptotic (\ref{PredMSE}) RMSEs of the MT-GQMLE versus SNR (b) and sample size (c) as compared to the GQMLE, MT-MUSIC, SGN-MUSIC, TYLER-MUSIC and the MLE. Notice that the MT-GQMLE outperforms the standard GQMLE and the robust MUSIC generalizations and attains estimation performance that are significantly closer to those obtained by the MLE that, unlike the MT-GQMLE, requires complete knowledge of the likelihood function.}
\label{Fig4}
\end{figure}
%%%
\begin{table}[H] 
\caption{\textbf{Source localization:} averaged running times for number of sensors $p=4$, sample size $N=5000$ and $\rm{SNR}=-10$ [dB].}
\begin{center}
\begin{tabular}{|l|l|l|}\hline
Estimator & Running Time [Sec] - Gaussian Noise & Running Time [Sec] - $K$-distributed Noise\\ \hline
GQMLE & 1e-3 & 1e-3 \\ \hline
MT-GQMLE & 5e-2 & 5e-2  \\ \hline
MT-MUSIC & 1e-2 & 1e-2 \\ \hline
SGN-MUSIC   &1e-3 & 1e-3 \\ \hline
Tyler-MUSIC   & 7e-3 & 7e-3 \\ \hline
MLE   & 2 & 25 \\ \hline
\end{tabular}
\end{center}
\label{Table3}
\end{table}%
%%%
%%%
\begin{table}[htbp]
\caption{\textbf{Source localization:} asymptotic computational load. Notations: $p$ and $N$ are the number of sensors and samples size, respectively. 
$K_{\Theta}$ and $K_{\Omega}$ denote the number of grid points of the parameter space $\Theta$ and the $\Omega$-axis (the width parameter space of the Gaussian MT-function (\ref{GaussMTFunc})), respectively. 
$L_{\rm{M}}$ and $L_{\rm{T}}$ are the number of fixed-point iterations used in the MT-MUSIC and Tyler-MUSIC estimators, respectively.}\begin{center}
\begin{tabular}{|l|l|}\hline
Estimator & Asymptotic computational load [flops] \\ \hline
GQMLE & $O\left(p^{2}(N+K_{\Theta})\right)$ \\ \hline
 %%%
MT-GQMLE & \begin{tabular}{l} MT-function optimization: $O(p^{2}(N+K_{\Theta})K_{\Omega})$ \\  
 Estimation: $O\left(p^{2}(N+K_{\Theta})\right)$  \end{tabular}  \\ \hline
%%%
MT-MUSIC & \begin{tabular}{l} MT-function optimization: $O((p^{3}+Np^{2})L_{M})$ \\  
 Estimation: $O(p^{3} + p^{2}(N+K_{\Theta}))$ \end{tabular}  \\ \hline
%%%
SGN-MUSIC &  $O(p^{3} + p^{2}(N+K_{\Theta}))$  \\ \hline
%%%
Tyler-MUSIC &  $O(p^{3} + p^{2}(N+K_{\Theta}) + (p^{3} + Np^{2})L_{T})$  \\ \hline
%%%
MLE - Gaussian Noise &  $O\left(pNK_{\Theta}\right)$ \\ \hline
%%%
MLE - $K$-distributed Noise &  $O\left(pNK_{\Theta}\right)$ \\ \hline
%%%
\end{tabular}
\end{center}
\label{Table4}
\end{table}
%%%%%%%%%%%%%%%%%%%%%%%%%%%%%%%%%%%%%%%%%%%%%%%%%%%%%%%%%%%%%%%%%%%%%%%%%%%%%%%%%%%%%%%%%%%%%%%%%%%%%%%%%%%%%%%%%%%%%
%%%%%%%%%%%%%%%%%%%%%%%%%%%%%%%%%%%%%%%%%%%%%%%%%%%%%%%%%%%%%%%%%%%%%%%%%%%%%%%%%%%%%%%%%%%%%%%%%%%%%%%%%%%%%%%%%%%%
\section{Conclusion}
\label{Conclusion}
In this paper a new multivariate estimator, called MT-GQMLE, was developed that applies Gaussian quasi maximum likelihood estimation under a transformed probability distribution of the data. 
We have shown that the MT-GQMLE can gain sensitivity to higher-order statistical moments and resilience to outliers when the MT-function associated with the transform is non-constant and satisfies some mild regularity conditions. After specifying the MT-function in the Gaussian family of functions, the proposed estimator was applied to linear regression and source localization in non-Gaussian noise. Exploration of other classes of MT-functions may result in additional estimators in this family that have different useful properties. 

We emphasize that although the MT-MUSIC estimator, presented in \cite{Todros2} and \cite{Todros3}, and the newly proposed MT-GQMLE, presented in this paper, are derived under the same general framework of probability measure transformation, these estimators are totally different. The MT-MUSIC estimator is designed for source localization and operates by finding angle of arrivals with corresponding array steering vectors that have minimal projections onto the empirical noise subspace, whose spanning vectors are obtained via eigendecomposition of the empirical MT-covariance of the array outputs. The proposed MT-GQMLE operates differently by fitting a Gaussian probability model, characterized by the parametric MT-mean and MT-covariance, to a transformed probability distribution of the data. Unlike the MT-MUSIC estimator the MT-GQMLE is much more general and not solely designated to source localization, as illustrated in the linear regression example in Subsection \ref{GainParam}. Furthermore, unlike the MT-MUSIC estimator, optimization of the MT-function in the MT-GQMLE is based on empirical estimate of the asymptotic MSE which can lead to significant performance advantage. Indeed, as shown in the single source localization example in Subsection \ref{SourceLoc} the MT-QMLE outperforms the MT-MUSIC estimator without significantly increased computational burden. However, we note that in multiple source localization, the computational burden of the MT-GQMLE will be significantly higher as, similarly to the MLE and the GQMLE, it involves maximization over a higher dimensional space. 
%\newpage
\appendix
%%%%%%%%%%%%%%%%%%%%%%%%%%%%%%%%%%%%%%%%%%%%%%%%%%%%%%%%%%%%%%%%%%%%%%%%%%%%%%%%%%%%%%%%%%%%%%%%%%%%%%%%%%%%%%%%%%%%%%%%%%%%%%%%%%%%%%%%%%%%%%%%%%%%%%%%%%%%%%%%%%%%%%%%%%%%%%%%%%%%%%%%%%%%%%%%%%%%%%%%%%%%%%%%%%%%%%%%%%%%%%%%%%%%%%%%%%
\subsection{Proof of Theorem \ref{ConsistencyTh}:}
\label{ConsistencyThProof}
Define the deterministic quantity
\begin{equation}
\label{ObjFunDet}
\bar{J}_{u}\left(\thetavec\right)\triangleq-D_{\rm{LD}}\left[{\bSigma}^{\left(u\right)}_{\Xmatsc}\left(\thetaveczero\right)||\bSigma^{\left(u\right)}_{\Xmatsc}\left(\thetavec\right)\right] 
-\left\|{\muvec}^{\left(u\right)}_{\Xmatsc}\left(\thetaveczero\right)-{\muvec}^{\left(u\right)}_{\Xmatsc}\left(\thetavec\right)\right\|^{2}_{{\left(\bSigmasc^{\left(u\right)}_{\xvec}\left(\thetavecsc\right)\right)}^{-1}}.
\end{equation}
According to Lemmas \ref{Lemma1} and  \ref{Lemma2}, stated below, $\bar{J}_{u}\left(\thetavec\right)$ is uniquely maximized at $\thetavec=\thetaveczero$ and the random objective function $J_{u}\left(\thetavec\right)$ defined in (\ref{ObjFun}) converges uniformly w.p. 1 to $\bar{J}_{u}\left(\thetavec\right)$ as $N\rightarrow\infty$. Furthermore, by assumptions {A}-\ref{AS_3} and  {A}-\ref{AS_4} $J_{u}\left(\thetavec\right)$ is continuous over the compact parameter space $\Thetasp$ w.p. 1.
Therefore, by Theorem 3.4 in \cite{White2} we conclude that (\ref{StrongConsistency}) holds.\qed
%%%
\begin{Lemma}
\label{Lemma1}
Let $\bar{J}_{u}\left(\thetavec\right)$ be defined by relation (\ref{ObjFunDet}). If conditions A-\ref{AS_1}-A-\ref{AS_4} are satisfied, then $\bar{J}_{u}\left(\thetavec\right)$ is uniquely maximized at $\thetavec=\thetaveczero$.
\end{Lemma}
\begin{proof}
By (\ref{ObjFunDet}) and assumptions A-\ref{AS_3} and A-\ref{AS_4} $\bar{J}_{u}\left(\thetavec\right)$ is continuous over the compact set $\Thetasp$. Therefore, by the extreme-value theorem \cite{AdCal} it must have a global maximum in $\Thetasp$. We now show that the global maximum is uniquely attained at $\thetavec=\thetaveczero$. According to (\ref{ObjFunDet})
\begin{equation}
\label{DiffEq}
\bar{J}_{u}\left(\thetaveczero\right)-\bar{J}_{u}\left(\thetavec\right)=D_{\rm{LD}}\left[\bSigma^{\left(u\right)}_{\Xmatsc}\left(\thetaveczero\right)||\bSigma^{\left(u\right)}_{\Xmatsc}\left(\thetavec\right)\right] +
\left\|{\muvec}^{\left(u\right)}_{\Xmatsc}\left(\thetaveczero\right)-{\muvec}^{\left(u\right)}_{\Xmatsc}\left(\thetavec\right)\right\|^{2}_{{\left(\bSigmasc^{\left(u\right)}_{\xvec}\left(\thetavecsc\right)\right)}^{-1}}.
\end{equation}
Notice that the log-determinant divergence between positive-definite matrices $\Amat$, $\Bmat$ satisfies  $D_{\rm{LD}}\left[\Amat||\Bmat\right]\geq{0}$ with equality if and only if $\Amat=\Bmat$ \cite{LogDetDiv}. 
Also notice that the weighted Euclidian norm of a vector $\avec$ with positive-definite weighting matrix $\Cmat$ satisfies $\left\|\avec\right\|_{\Cmat}\geq{0}$ with equality if and only if $\avec=\zerovec$.
Therefore, by assumptions A-\ref{AS_2} and A-\ref{AS_3} we conclude that the difference in (\ref{DiffEq}) is finite and strictly positive for all $\thetavec\neq\thetaveczero$, which implies that
$\bar{J}_{u}\left(\thetavec\right)$ is uniquely maximized at $\thetavec=\thetaveczero$.
\end{proof}
%%%
\begin{Lemma}
\label{Lemma2}
Let $J_{u}\left(\thetavec\right)$ and $\bar{J}_{u}\left(\thetavec\right)$ be defined by relations (\ref{ObjFun}) and (\ref{ObjFunDet}), respectively. 
If conditions A-\ref{AS_1} and A-\ref{AS_3}-A-\ref{AS_5} are satisfied, then
\begin{equation}
\label{UniformConv}
\sup\limits_{\thetavecsc\in\Thetaspsc}\left|J_{u}\left(\thetavec\right)-\bar{J}_{u}\left(\thetavec\right)\right|\xrightarrow{\textit{w.p. 1}}{0}\hspace{0.2cm}\text{as $N\rightarrow\infty$}.
\end{equation}
\end{Lemma}
\begin{proof}
Using (\ref{ObjFun}),  (\ref{ObjFunDet}), the Cauchy-Schwartz inequality and its the matrix extension \cite{CS} one can verify that 
\begin{eqnarray}
\label{Lemma2Eq}
\nonumber
\sup\limits_{\thetavecsc\in\Thetaspsc}\left|J_{u}\left(\thetavec\right)-\bar{J}_{u}\left(\thetavec\right)\right|
&\leq&
\left(\left\|\bSigma^{\left(u\right)}_{\Xmatsc}\left(\thetaveczero\right)-\hat{\bSigma}^{\left(u\right)}_{\Xmatsc}\right\|_{\rm{Fro}}
+\left\|{\muvec}^{\left(u\right)}_{\Xmatsc}\left(\thetaveczero\right){\muvec}^{\left(u\right)H}_{\Xmatsc}\left(\thetaveczero\right) 
- \hat{\muvec}^{\left(u\right)}_{\Xmatsc}\hat{\muvec}^{\left(u\right)H}_{\Xmatsc} \right\|_{\rm{Fro}}\right)
\\\nonumber&\times&
\sup\limits_{\thetavecsc\in\Thetaspsc}{\left\|\left(\bSigma^{\left(u\right)}_{\Xmatsc}\left(\thetavec\right)\right)^{-1}\right\|_{\rm{Fro}}}
+\left|\log\det\left[\hat{\bSigma}^{\left(u\right)}_{\Xmatsc}\left(\bSigma^{\left(u\right)}_{\Xmatsc}\left(\thetaveczero\right)\right)^{-1}\right]\right|
\\&+&
2\left\|{\muvec}^{\left(u\right)}_{\Xmatsc}\left(\thetaveczero\right)-\hat{\muvec}^{\left(u\right)}_{\Xmatsc} \right\|
\sup\limits_{\thetavecsc\in\Thetaspsc}{\left\|\left(\bSigma^{\left(u\right)}_{\Xmatsc}\left(\thetavec\right)\right)^{-1}{\muvec}^{\left(u\right)}_{\Xmatsc}\left(\thetavec\right)\right\|}
\hspace{0.4cm}\textrm{w.p. 1}.
\end{eqnarray}
Notice that by assumptions  A-\ref{AS_3} and A-\ref{AS_4} and the compactness of $\Thetasp$, the terms
$\sup\limits_{\thetavecsc\in\Thetaspsc}{\|(\bSigma^{\left(u\right)}_{\Xmatsc}\left(\thetavec)\right)^{-1}\|_{\rm{Fro}}}$ and
$\sup\limits_{\thetavecsc\in\Thetaspsc}{\|(\bSigma^{\left(u\right)}_{\Xmatsc}\left(\thetavec\right))^{-1}{\muvec}^{\left(u\right)}_{\Xmatsc}\left(\thetavec\right)\|}$
are finite. Also notice that since $\Xmat_{n}$, $n=1,\ldots,N$ is a sequence of i.i.d. samples from $\pxthetazero$, then by Proposition \ref{ConsistentEst} and assumption A-\ref{AS_5} $\hat{\muvec}^{\left(u\right)}_{\Xmatsc}\xrightarrow{\textit{w.p. 1}}{\muvec}^{\left(u\right)}_{\Xmatsc}\left(\thetaveczero\right)$ and $\hat{\bSigma}^{\left(u\right)}_{\Xmatsc}\xrightarrow{\textit{w.p. 1}}{\bSigma}^{\left(u\right)}_{\Xmatsc}\left(\thetaveczero\right)$ as $N\rightarrow\infty$. Furthermore, note that the operators $\left\|\cdot\right\|_{\rm{Fro}}$, $\left\|\cdot\right\|$ and $\log\det\left[\cdot\right]$ define real continuous mappings of $\hat{\muvec}^{\left(u\right)}_{\Xmatsc}$ and $\hat{\bSigma}^{\left(u\right)}_{\Xmatsc}$. Hence, by the Mann-Wald Theorem \cite{MannWald} the upper bound in (\ref{Lemma2Eq}) converges to zero w.p. 1 as $N\rightarrow\infty$, and therefore, the relation (\ref{UniformConv}) must hold.
\end{proof}
%%%%%%%%%%%%%%%%%%%%%%%%%%%%%%%%%%%%%%%%%%%%%%%%%%%%%%%%%%%%%%%%%%%%%%%%%%%%%%%%%%%%%%%%%%%%%%%%%%%%%%%%%%%%%%%%%%%%%%%%%%%%%%%%%%%%%%%%%%%%%%%%%%%%%%%%%%%%%%%%%%%%%%%%%%%%%%%%%%%%%%%%%%%%%%%%%%%%%%
\subsection{Proof of Theorem \ref{ANor}:}  
\label{ANorProof}
By assumptions B-\ref{C1} and B-\ref{C2} the estimator $\hat{\thetavec}_{u}$ is weakly consistent and the true parameter vector $\thetaveczero$ lies in the interior of $\Thetasp$, which is assumed to be compact. Therefore, we conclude that $\hat{\thetavec}_{u}$ lies in an open neighbourhood $\mathcal{U}\subset\Thetasp$ of $\thetaveczero$ with sufficiently high probability as $N$ gets large, i.e., it does not lie on the boundary of $\Thetasp$. Hence, $\hat{\thetavec}_{u}$ is a maximum point of the objective function $J_{u}\left(\thetavec\right)$ (\ref{ObjFun}) whose gradient satisfies
 \begin{equation}
\label{Grad}
\nabla_{\thetavecsc}{J}_{u}\left(\hat{\thetavec}_{u}\right)\sum\limits_{n=1}^{N}u\left(\Xmat_{n}\right)=\sum\limits_{n=1}^{N}u\left(\Xmat_{n}\right)\psivec_{u}\left(\Xmat_{n};\hat{\thetavec}_{u}\right)=\zerovec.
\end{equation}
By (\ref{GaussDens}) and assumption B-\ref{C3}, the vector function $\psivec_{u}\left(\Xmat;\thetavec\right)$ defined in (\ref{PsiDef}) is continuous over $\Thetasp$ w.p. 1. Using (\ref{Grad}) and the mean-value theorem \cite{AdCalMult} applied to each entry of $\psivec_{u}\left(\Xmat;\thetavec\right)$ we conclude that
\begin{equation}
\label{GradConc}
\zerovec=\frac{1}{\sqrt{N}}\sum\limits_{n=1}^{N}u\left(\Xmat_{n}\right)\psivec_{u}\left(\Xmat_{n};\hat{\thetavec}_{u}\right)=\frac{1}{\sqrt{N}}\sum\limits_{n=1}^{N}u\left(\Xmat_{n}\right)\psivec_{u}\left(\Xmat_{n};\thetaveczero\right)
-\hat{\Fmat}_{u}\left(\thetavec^{{*}}\right)\sqrt{N}\left(\hat{\thetavec}_{u}-\thetaveczero\right),
\end{equation}
where $\hat{\Fmat}_{u}\left(\thetavec\right)$ defined in (\ref{Fhat}) is an unbiased estimator of the symmetric matrix function $\Fmat_{u}\left(\thetavec\right)$ (\ref{FDef}), $\Gammamat_{u}\left(\Xmat;\thetavec\right)$ is defined in (\ref{GammaDef}) and $\thetavec^{*}$ lies on the line segment connecting $\hat{\thetavec}_{u}$ and $\thetaveczero$. 

Since $\hat{\thetavec}_{u}$ is weakly consistent $\thetavec^{*}$ must be weakly consistent as well.
Therefore, by Lemma \ref{WeakCons}, stated below, $\hat{\Fmat}_{u}\left(\thetavec^{*}\right)\xrightarrow{P}{\Fmat}_{u}\left(\thetaveczero\right)$ as $N\rightarrow\infty$, where $\Fmat_{u}\left(\thetavec\right)$ is non-singular at $\thetavec=\thetaveczero$ by assumption. Hence, by the Mann-Wald theorem \cite{MannWald}
\begin{equation}
\label{FinvCon}
\hat{\Fmat}^{-1}_{u}\left(\thetavec^{*}\right)\xrightarrow{P}{\Fmat}^{-1}_{u}\left(\thetaveczero\right)\hspace{0.2cm}\text{as $N\rightarrow\infty$},
\end{equation}
which implies that $\hat{\Fmat}_{u}\left(\thetavec^{*}\right)$ is invertible with sufficiently high probability as $N$ gets large. Therefore, by (\ref{GradConc}) the equality
\begin{equation}
\sqrt{N}\left(\hat{\thetavec}_{u}-\thetaveczero\right)=\hat{\Fmat}^{-1}_{u}\left(\thetavec^{*}\right)\frac{1}{\sqrt{N}}\sum\limits_{n=1}^{N}u\left(\Xmat_{n}\right)\psivec_{u}\left(\Xmat_{n};\thetaveczero\right)
\end{equation}
holds with sufficiently high probability as $N$ gets large. Furthermore, by Lemma \ref{GaussLem} stated below
\begin{equation}
\label{PsiSumDist}
\frac{1}{\sqrt{N}}\sum\limits_{n=1}^{N}u\left(\Xmat_{n}\right)\psivec_{u}\left(\Xmat_{n};\thetaveczero\right)\xrightarrow{D}\mathcal{N}\left(\zerovec,\Gmat_{u}\left(\thetaveczero\right)\right)\hspace{0.2cm}\text{as $N\rightarrow\infty$},
\end{equation}
where $\Gmat_{u}\left(\thetavec\right)$ is defined in (\ref{GDef}).  
Thus, by (\ref{FinvCon})-(\ref{PsiSumDist}) and Slutsky's theorem \cite{MeasureTheory} the relation (\ref{Asymp}) holds.
%%%
\begin{Lemma}
\label{WeakCons}
Let $\Fmat_{u}\left(\thetavec\right)$ and $\hat{\Fmat}_{u}\left(\thetavec\right)$ be the matrix functions defined in (\ref{FDef}) and (\ref{Fhat}), respectively.
Furthermore, let $\thetavec^{*}$ denote an estimator of $\thetaveczero$.
If $\thetavec^{*}\xrightarrow{P}\thetaveczero$ as $N\rightarrow\infty$ and conditions B-\ref{C3} and B-\ref{C4} are satisfied, then $\hat{\Fmat}_{u}\left(\thetavec^{*}\right)\xrightarrow{P}{\Fmat}_{u}\left(\thetaveczero\right)$ as $N\rightarrow\infty$.
\end{Lemma}
\begin{proof}
By the triangle inequality
\begin{equation}
\|\hat{\Fmat}_{u}\left(\thetavec^{*}\right)-{\Fmat}_{u}\left(\thetaveczero\right)\|\leq 
\|\hat{\Fmat}_{u}\left(\thetavec^{*}\right)-{\Fmat}_{u}\left(\thetavec^{*}\right)\|+
\left\|{\Fmat}_{u}\left(\thetavec^{*}\right)-{\Fmat}_{u}\left(\thetaveczero\right)\right\|.
\end{equation}
Therefore, it suffices to prove that 
\begin{equation}
\|\hat{\Fmat}_{u}\left(\thetavec^{*}\right)-{\Fmat}_{u}\left(\thetavec^{*}\right)\|\xrightarrow{P}0\hspace{0.2cm}\text{as $N\rightarrow\infty$}
\end{equation}
and
\begin{equation}
\label{Conv2}
\left\|{\Fmat}_{u}\left(\thetavec^{*}\right)-{\Fmat}_{u}\left(\thetaveczero\right)\right\|\xrightarrow{P}0\hspace{0.2cm}\text{as $N\rightarrow\infty$}.
\end{equation}
By Lemma \ref{FCont}, stated below, $\Fmat_{u}\left(\thetavec\right)$ is continuous in $\Thetasp$. Therefore, since $\thetavec^{*}\xrightarrow{P}\thetaveczero$ as $N\rightarrow\infty$, by the Mann-Wald Theorem \cite{MannWald} we conclude that (\ref{Conv2}) holds. We further conclude by Lemma  \ref{FCont} that 
$\|\hat{\Fmat}_{u}\left(\thetavec^{*}\right)-{\Fmat}_{u}\left(\thetavec^{*}\right)\|\leq\sup\limits_{\thetavecsc\in{\Thetaspsc}}\left\|\hat{\Fmat}_{u}\left(\thetavec\right)-{\Fmat}_{u}\left(\thetavec\right)\right\|\xrightarrow{P}0,\hspace{0.2cm}\text{as $N\rightarrow\infty$}$.
\end{proof}
%%%
%%%
\begin{Lemma}
\label{FCont}
Let $\Fmat_{u}\left(\thetavec\right)$ and $\hat{\Fmat}_{u}\left(\thetavec\right)$ be the matrix functions defined in (\ref{FDef}) and (\ref{Fhat}), respectively.
If conditions B-\ref{C3} and B-\ref{C4} are satisfied, then $\Fmat_{u}\left(\thetavec\right)$ is continuous in $\Thetasp$ and $\sup\limits_{\thetavecsc\in\Thetaspsc}\left\|\hat{\Fmat}_{u}\left(\thetavec\right)-{\Fmat}_{u}\left(\thetavec\right)\right\|\xrightarrow{P}{0}$ as $N\rightarrow\infty$.
\end{Lemma}
\begin{proof}
Notice that the samples $\Xmat_{n}$, $n=1,\ldots,N$ are i.i.d., the parameter space $\Thetasp$ is compact and by  (\ref{GaussDens}) and assumption B-\ref{C3} the matrix function $\Gammamat_{u}\left(\Xmat;\thetavec\right)$ defined in (\ref{GammaDef}) is continuous in $\Thetasp$ w.p. 1. Furthermore, using (\ref{GaussDens}), (\ref{GammaDef}), the triangle inequality, the Cauchy-Schwartz inequality, assumption B-\ref{C3}, and the compactness of $\Thetasp$ it can be shown that there exists a positive scalar $C$ such that $\left|\left[\Gammamat_{u}\left(\Xmat;\thetavec\right)\right]_{k,j}\right|\leq\left(1+\left\|\Xmat\right\|+\left\|\Xmat\right\|^{2}\right)C$ w.p. 1 $\forall{k,j}=1,\ldots,p$. Therefore, by (\ref{FDef}), (\ref{Fhat}), assumption B-\ref{C4} and the uniform weak law of large numbers \cite{newey} we conclude that $\Fmat_{u}\left(\thetavec\right)$ is continuous in $\Thetasp$ and $\sup\limits_{\thetavecsc\in\Thetaspsc}\left\|\hat{\Fmat}_{u}\left(\thetavec\right)-{\Fmat}_{u}\left(\thetavec\right)\right\|\xrightarrow{P}{0}$ as $N\rightarrow\infty$.
\end{proof}
%%%
\begin{Lemma}
\label{GaussLem}
Given a sequence $\Xmat_{n}$, $n=1,\ldots,N$ of i.i.d. samples from $\pxthetazero$ 
$$\frac{1}{\sqrt{N}}\sum\limits_{n=1}^{N}u\left(\Xmat_{n}\right)\psivec_{u}\left(\Xmat_{n};\thetaveczero\right)\xrightarrow{D}\mathcal{N}\left(\zerovec,\Gmat_{u}\left(\thetaveczero\right)\right)\hspace{0.2cm}{\rm{as}}\hspace{0.2cm}N\rightarrow\infty,$$ where $\Gmat_{u}\left(\thetavec\right)$  and $\psivec_{u}\left(\Xmat;\thetavec\right)$ are defined in (\ref{GDef}) and (\ref{PsiDef}), respectively. 
\end{Lemma}
\begin{proof}
Since $\Xmat_{n}$, $n=1,\ldots,N$ are i.i.d. random vectors and the functions $u\left(\cdot\right)$ and $\psivec_{u}\left(\cdot,\cdot\right)$ are real the products $u\left(\Xmat_{n}\right)\psivec_{u}\left(\Xmat_{n},\thetaveczero\right)$  $n=1,\ldots,N$ are real and i.i.d. 
Furthermore, by Lemmas \ref{PsiuExp} and \ref{GFinit}, stated below, the definition of $\Gmat_{u}\left(\thetavec\right)$ in (\ref{GDef}) and the compactness of $\Thetasp$ we have that ${\rm{E}}\left[u\left(\Xmat\right)\psivec_{u}\left(\Xmat;\thetaveczero\right);\pxthetazero\right]=\zerovec$ and $\textrm{cov}\left[u\left(\Xmat\right)\psivec_{u}\left(\Xmat;\thetaveczero\right);\pxthetazero\right]=\Gmat_{u}\left(\thetaveczero\right)$ is finite. Therefore by the central limit theorem \cite{Serfling} we conclude that 
$\frac{1}{\sqrt{N}}\sum\limits_{n=1}^{N}u\left(\Xmat_{n}\right)\psivec_{u}\left(\Xmat_{n};\thetaveczero\right)$ is asymptotically normal with zero mean and covariance $\Gmat_{u}\left(\thetaveczero\right)$.
\end{proof}
%%%
\begin{Lemma}
\label{PsiuExp}  
The random vector function $\psivec_{u}\left(\Xmat;\thetavec\right)$ defined in (\ref{PsiDef}) satisfies
\begin{equation}
\label{EPsiU}
{\rm{E}}\left[u\left(\Xmat\right)\psivec_{u}\left(\Xmat;\thetavec\right);\pxtheta\right]=\zerovec.
\end{equation}
\end{Lemma}
\begin{proof}
Notice that 
\begin{equation}
\label{EUQ}
\textrm{E}\left[u\left(\Xmat\right)\psivec_{u}\left(\Xmat;\thetavec\right);\pxtheta\right]=\textrm{E}\left[\psivec_{u}\left(\Xmat;\thetavec\right)\varphi_{u}\left(\Xmat;\thetavec\right);\pxtheta\right]{\textrm{E}\left[u\left(\Xmat\right);\pxtheta\right]}, 
\end{equation}
where $\varphi_{u}\left(\cdot;\cdot\right)$ is defined in (\ref{VarPhiDef}) and the expectation ${\textrm{E}}\left[u\left(\Xmat\right);\pxtheta\right]$ is finite by assumption (\ref{Cond}). 
Also notice that by (\ref{GaussDens}) and (\ref{PsiDef}) the $k$-th entry of the vector function $\psivec_{u}\left(\Xmat;\thetavec\right)$ is given by:
\begin{eqnarray}
\label{PsivecExp} 
\left[\psivec_{u}\left(\Xmat;\thetavec\right)\right]_{k}&\triangleq&\frac{\partial\log\phi^{(u)}\left(\Xmat;\thetavec\right)}{\partial\theta_{k}}
=
-{\rm{tr}}\left[\left(\bSigma^{(u)}_{\Xmatsc}\left(\thetavec\right)\right)^{-1}\frac{\partial\bSigma^{(u)}_{\Xmatsc}\left(\thetavec\right)}{\partial\theta_{k}}\right]
\\\nonumber
&+&2{\rm{Re}}\left\{\left(\Xmat-\muvec^{\left(u\right)}_{\Xmatsc}\left(\thetavec\right)\right)^{H}\left(\bSigma^{(u)}_{\Xmatsc}\left(\thetavec\right)\right)^{-1}
\frac{\partial\muvec^{\left(u\right)}_{\Xmatsc}\left(\thetavec\right)}{\partial\theta_{k}}\right\}
\\\nonumber
&+&
{\rm{tr}}\left[\left(\bSigma^{(u)}_{\Xmatsc}\left(\thetavec\right)\right)^{-1}
\frac{\partial\bSigma^{(u)}_{\Xmatsc}\left(\thetavec\right)}{\partial\theta_{k}}
\left(\bSigma^{(u)}_{\Xmatsc}\left(\thetavec\right)\right)^{-1}\left(\Xmat-\muvec^{\left(u\right)}_{\Xmatsc}\left(\thetavec\right)\right)
\left(\Xmat-\muvec^{\left(u\right)}_{\Xmatsc}\left(\thetavec\right)\right)^{H}\right].
\end{eqnarray}
Therefore, by (\ref{MTMean}), (\ref{MTCovZ}), (\ref{EUQ}) and (\ref{PsivecExp}) the relation (\ref{EPsiU}) is easily verified.
\end{proof}
%%%
\begin{Lemma}
\label{GFinit}
Let $\Gmat_{u}\left(\thetavec\right)$ and $\hat{\Gmat}_{u}\left(\thetavec\right)$ be the matrix functions defined in (\ref{GDef}) and (\ref{Ghat}), respectively.
If conditions B-\ref{C3} and B-\ref{C4} are satisfied, then $\Gmat_{u}\left(\thetavec\right)$ is continuous in $\Thetasp$ and $\sup\limits_{\thetavecsc\in\Thetaspsc}\left\|\hat{\Gmat}_{u}\left(\thetavec\right)-{\Gmat}_{u}\left(\thetavec\right)\right\|\xrightarrow{P}{0}$ as $N\rightarrow\infty$.
\end{Lemma}
\begin{proof}
Notice that the samples $\Xmat_{n}$, $n=1,\ldots,N$ are i.i.d., the parameter space $\Thetasp$ is compact and by  (\ref{GaussDens}) and assumption B-\ref{C3} the vector function 
$\psivec_{u}\left(\Xmat;\thetavec\right)$ defined in (\ref{PsiDef}) is continuous in $\Thetasp$ w.p. 1. Using (\ref{GaussDens}), (\ref{PsiDef}), the triangle inequality, the Cauchy-Schwartz inequality, assumption B-\ref{C3}, and the compactness of $\Thetasp$ it can be shown that there exists a positive scalar $C$ such that
\begin{equation}  
\label{vkjbound}
\left|\left[\psivec_{u}\left(\Xmat;\thetavec\right)\psivec^{T}_{u}\left(\Xmat;\thetavec\right)\right]_{k,j}\right|\leq
\left(1+\left\|\Xmat\right\|
+\left\|\Xmat\right\|^{2}+\left\|\Xmat\right\|^{3}
+\left\|\Xmat\right\|^{4}\right)C\hspace{0.2cm}\forall{k,j}=1,\ldots,p.
\end{equation}
Therefore, by (\ref{GDef}), (\ref{Ghat}), assumption B-\ref{C4} and the uniform weak law of large numbers \cite{newey} we conclude that $\Gmat_{u}\left(\thetavec\right)$ is continuous in $\Thetasp$ and $\sup\limits_{\thetavecsc\in\Thetaspsc}\left\|\hat{\Gmat}_{u}\left(\thetavec\right)-{\Gmat}_{u}\left(\thetavec\right)\right\|\xrightarrow{P}{0}$ as $N\rightarrow\infty$.
\end{proof}
%%%%%%%%%%%%%%%%%%%%%%%%%%%%%%%%%%%%%%%%%%%%%%%%%%%%%%%%%%%%%%%%%%%%%%%%%%%%%%%%%%%%%%%%%%%%%%%%%%%%%%%%%%%%%%%%%%%%%%%%%%%%%%%%%%%%%%%%%%%%%%%%%%%%%%%%%%%%%%%%%%%%%%%%%%%%%%%%%%%%%%%%%%%%%%%%%%%%%%
\subsection{Proof of Proposition \ref{EffTh}:}   
\label{EffThProof}
Define $\bxi\left(\Xmat;\thetavec\right)\triangleq{u}(\Xmat)\psivec_{u}(\Xmat;\thetavec)$. Using (\ref{AMSEN}), the definitions of $\Rmat_{u}\left(\cdot\right)$,  $\Gmat_{u}\left(\cdot\right)$
and $\Fmat_{u}\left(\cdot\right)$ in (\ref{AMSE}), (\ref{GDef}) and (\ref{FDef}), respectively, Identity \ref{FIdent} stated below, the symmetricity of  $\Fmat_{u}\left(\thetavec\right)$, the non-singularity of $\Gmat_{u}\left(\thetavec\right)$ at $\thetavec=\thetaveczero$ and the covariance semi-inequality \cite{Lehmann} one can verify that
\begin{eqnarray}
\label{InvC}
\nonumber
\Cmat^{-1}_{u}\left(\thetaveczero\right)&=&{N}{\rm{E}}[\etavec(\Xmat;\thetaveczero)\bxi^{T}\left(\Xmat;\thetaveczero\right);\pxthetazero]
{\rm{E}}^{-1}\left[\bxi\left(\Xmat;\thetaveczero\right)\bxi^{T}\left(\Xmat;\thetaveczero\right);\pxthetazero\right]
\\\nonumber
&\times&{\rm{E}}[\bxi\left(\Xmat;\thetaveczero\right)\etavec^{T}(\Xmat;\thetaveczero);\pxthetazero]
\\&\preceq&{N}{\rm{E}}[\etavec(\Xmat;\thetaveczero)\etavec^{T}(\Xmat;\thetaveczero);\pxthetazero]\triangleq{N}\Imat_{\rm{FIM}}\left[\pxthetazero\right],
\end{eqnarray}
where equality holds if and only if the relation (\ref{AchCRB}) holds. The proof is complete by taking the inverse of (\ref{InvC}).
\qed
%%%
\begin{Identity}
\label{FIdent}
$\Fmat_{u}(\thetavec)={\rm{E}}[u(\Xmat)\psivec_{u}(\Xmat;\thetavec)\etavec^{T}(\Xmat;\thetavec);\pxtheta]$
\end{Identity}
%%%
\begin{proof}
According to Lemma \ref{PsiuExp}, stated in Appendix \ref{ANorProof}, ${\rm{E}}\left[u\left(\Xmat\right)\psivec_{u}\left(\Xmat,\thetavec\right);\pxtheta\right]=\zerovec$.
Therefore, by (\ref{likelihoodfunc}), (\ref{ExpDef}) and (\ref{PsiDef})
\begin{eqnarray}
\label{IdentEq1}
\zerovec
&=&
\frac{\partial}{\partial\thetavec}\int\limits_{\XCalsc}u\left(\xvec\right)\left(\frac{\partial\log\phi^{(u)}\left(\xvec;\thetavec\right)}{\partial\thetavec}\right)^{T}f\left(\xvec;\thetavec\right)d\rho\left(\xvec\right)
\\\nonumber
&=&
\int\limits_{\XCalsc}u\left(\xvec\right)\frac{\partial^{2}\log\phi^{(u)}\left(\xvec;\thetavec\right)}{\partial\thetavec\partial\thetavec^{T}}f\left(\xvec;\thetavec\right)d\rho\left(\xvec\right)
\\\nonumber
&+&
\int\limits_{\XCalsc}u\left(\xvec\right)\left(\frac{\partial\log\phi^{(u)}\left(\xvec;\thetavec\right)}{\partial\thetavec}\right)^{T}{\frac{\partial\log{f}\left(\xvec;\thetavec\right)}{\partial\thetavec}}
{f}\left(\xvec;\thetavec\right)d\rho\left(\xvec\right).
\end{eqnarray}
The second equality in (\ref{IdentEq1}) stems from assumptions C-\ref{D1}-C-\ref{D3} and Theorem 2.40 in \cite{giaquinta2003mathematical} according to which integration and differentiation can be interchanged. Therefore, by (\ref{likelihoodfunc}), (\ref{PsiDef})-(\ref{logfGrad}) and (\ref{IdentEq1})
\begin{equation}
{\rm{E}}[u(\Xmat)\psivec_{u}(\Xmat;\thetavec)\etavec^{T}(\Xmat;\thetavec);\pxtheta]=-{\rm{E}}[u(\Xmat)\Gammamat_{u}(\Xmat;\thetavec);\pxtheta]\triangleq\Fmat_{u}(\thetavec).
\end{equation}
\end{proof}
%%%%%%%%%%%%%%%%%%%%%%%%%%%%%%%%%%%%%%%%%%%%%%%%%%%%%%%%%%%%%%%%%%%%%%%%%%%%%%%%%%%%%%%%%%%%%%%%%%%%%%%%%%%%%%%%%%%%%%%%%%%%%%%%%%%%%%%%%%%%%%%%%%%%%%%%%%%%%%%%%%%%%%%%%%%%%%%%%%%%%%%%%%%%%%%%%%%%%%
\subsection{Proof of Theorem \ref{EAMSETh}:}   
\label{EAMSEThProof}
Under assumption D-\ref{E1}, $\hat{\thetavec}_{u}\xrightarrow{P}\thetaveczero$ as $N\rightarrow\infty$. Therefore, by Lemma \ref{WeakCons}, stated in Appendix \ref{ANorProof}, and Lemma \ref{WeakConG}, stated below, that require conditions 
D-\ref{E3} and D-\ref{E4} to be satisfied, we have that 
$\hat{\Fmat}_{u}(\hat{\thetavec}_{u})\xrightarrow{P}{\Fmat}_{u}\left(\thetaveczero\right)$ as $N\rightarrow\infty$ and 
$\hat{\Gmat}_{u}(\hat{\thetavec}_{u})\xrightarrow{P}{\Gmat}_{u}\left(\thetaveczero\right)$ as $N\rightarrow\infty$.
Hence, by (\ref{AMSE}),  (\ref{AMSEN}), (\ref{EMSEN}), (\ref{EMSE}) and the Mann-Wald theorem \cite{MannWald} we conclude that (\ref{ChatConv}) is satisfied.\qed
%%%
\begin{Lemma}
\label{WeakConG}
Let $\Gmat_{u}\left(\thetavec\right)$ and $\hat{\Gmat}_{u}\left(\thetavec\right)$ be the matrix functions defined in (\ref{GDef}) and (\ref{Ghat}), respectively.
Furthermore, let $\thetavec^{*}$ denote an estimator of $\thetaveczero$.
If $\thetavec^{*}\xrightarrow{P}\thetaveczero$ as $N\rightarrow\infty$ and conditions D-\ref{E3} and D-\ref{E4} are satisfied, then $\hat{\Gmat}_{u}\left(\thetavec^{*}\right)\xrightarrow{P}{\Gmat}_{u}\left(\thetaveczero\right)$ as $N\rightarrow\infty$.
\end{Lemma}
\begin{proof}
The proof is similar to the one of Lemma \ref{WeakCons} stated in Appendix \ref{ANorProof}. Simply replace ${\Fmat}_{u}\left(\thetavec\right)$ and $\hat{\Fmat}_{u}\left(\thetavec\right)$ with ${\Gmat}_{u}\left(\thetavec\right)$ and $\hat{\Gmat}_{u}\left(\thetavec\right)$ and use Lemma \ref{GFinit} instead of Lemma \ref{FCont}.
\end{proof}
%%%%%%%%%%%%%%%%%%%%%%%%%%%%%%%%%%%%%%%%%%%%%%%%%%%%%%%%%%%%%%%%%%%%%%%%%%%%%%%%%%%%%%%%%%%%%%%%%%%
%%%%%%%%%%%%%%%%%%%%%%%%%%%%%%%%%%%%%%%%%%%%%%%%%%%%%%%%%%%%%%%%%%%%%%%%%%%%%%%%%%%%%%%%%%%%%%%%%%%
\subsection{Proof of Proposition \ref{FishConsistency}:}   
\label{FishConsistencyProof}
We first show that $\hat{\thetavec}_{u}$ can be represented as a statistical functional of the empirical probability distribution. 
As also shown in \cite{Todros2}, the empirical MT-mean and MT-covariance (\ref{Mu_u_Est}), (\ref{Rx_u_Est})  can be written as statistical functionals of the empirical probability measure $\hat{P}_{\Xmatsc}$ defined in (\ref{EmpProbMes}), i.e.,
\begin{equation}
\label{StatFuncMean}
\hat{\muvec}^{\left(u\right)}_{\Xmatsc}=\frac{{\rm{E}}[\Xmat{u}\left(\Xmat\right);\hat{P}_{\Xmatsc}]}{{\rm{E}}[{u}\left(\Xmat\right);\hat{P}_{\Xmatsc}]}\triangleq\etavec_{\Xmatsc}^{(u)}[\hat{P}_{\Xmatsc}]
\end{equation}
and
\begin{equation}
\label{StatFuncCov} 
\hat{\bSigma}^{\left(u\right)}_{\Xmatsc}=\frac{{\rm{E}}[(\Xmat-\etavec_{\Xmatsc}^{(u)}[\hat{P}_{\Xmatsc}])(\Xmat-\etavec_{\Xmatsc}^{(u)}[\hat{P}_{\Xmatsc}])^{H}u\left(\Xmat\right);\hat{P}_{\Xmatsc}]}{{\rm{E}}[{u}\left(\Xmat\right);\hat{P}_{\Xmatsc}]}
\triangleq{\Psimat}_{\Xmatsc}^{(u)}[\hat{P}_{\Xmatsc}].
\end{equation} 
Hence, the objective function (\ref{ObjFun}) can be represented as:
\begin{equation}
\label{ObjFunMod}
J_{u}\left(\thetavec\right)=
-D_{\rm{LD}}\left[{\Psimat}_{\Xmatsc}^{(u)}[\hat{P}_{\Xmatsc}]||\bSigma^{\left(u\right)}_{\Xmatsc}\left(\thetavec\right)\right] -
\left\|\etavec_{\Xmatsc}^{(u)}[\hat{P}_{\Xmatsc}]-{\muvec}^{\left(u\right)}_{\Xmatsc}\left(\thetavec\right)\right\|^{2}_{{\left(\bSigmasc^{\left(u\right)}_{\xvec}\left(\thetavecsc\right)\right)}^{-1}}
\triangleq{\rm{H}}_{u}[\hat{P}_{\Xmatsc};\thetavec],
\end{equation}
and therefore, by (\ref{PropEst}) we conclude that 
\begin{equation}
\label{PropEstMod}
\hat{\thetavec}_{u}=\arg\max\limits_{\thetavecsc\in\Thetaspsc}{\rm{H}}_{u}[\hat{P}_{\Xmatsc};\thetavec]\triangleq\Smat_{u}[\hat{P}_{\Xmatsc}]. 
\end{equation}

Next, we prove Fisher consistency of $\hat{\thetavec}_{u}$. According to (\ref{VarPhiDef}),  (\ref{MTMean}),  (\ref{MTCovZ}), (\ref{StatFuncMean}) and (\ref{StatFuncCov}), when $\hat{P}_{\Xmatsc}$ is replaced by $\pxthetazero$ we have $\etavec_{\Xmatsc}^{(u)}[{P}_{\Xmatsc;\thetaveczerosc}]={\muvec}^{\left(u\right)}_{\Xmatsc}\left(\thetaveczero\right)$ and ${\Psimat}_{\Xmatsc}^{(u)}[\pxtheta]={\bSigma}^{(u)}_{\Xmatsc}\left(\thetaveczero\right)$. Thus, by (\ref{ObjFunMod})
\begin{equation}
\label{ObjFunP0}
{\rm{H}}_{u}[\pxthetazero;\thetavec]=
-D_{\rm{LD}}\left[{\bSigma}^{(u)}_{\Xmatsc}\left(\thetaveczero\right)||\bSigma^{\left(u\right)}_{\Xmatsc}\left(\thetavec\right)\right] -
\left\|{\muvec}^{\left(u\right)}_{\Xmatsc}\left(\thetaveczero\right)-{\muvec}^{\left(u\right)}_{\Xmatsc}\left(\thetavec\right)\right\|^{2}_{{\left(\bSigmasc^{\left(u\right)}_{\xvec}\left(\thetavecsc\right)\right)}^{-1}}.
\end{equation}
Finally, since ${\rm{H}}_{u}[\pxthetazero;\thetavec]=\bar{J}_{u}\left(\thetavec\right)$ defined in (\ref{ObjFunDet}), we conclude by Lemma \ref{Lemma1} in Appendix \ref{ConsistencyThProof} that if conditions A-\ref{AS_1}-A-\ref{AS_4} are satisfied, then ${\rm{H}}_{u}[\pxthetazero;\thetavec]$ is uniquely maximized at $\thetavec=\thetaveczero$. Therefore, by (\ref{PropEstMod}) $\Smat_{u}[\pxthetazero]=\thetaveczero$. \qed
%%%%%%%%%%%%%%%%%%%%%%%%%%%%%%%%%%%%%%%%%%%%%%%%%%%%%%%%%%%%%%%%%%%%%%%%%%%%%%%%%%%%%%%%%%%%%%%%%%%
%%%%%%%%%%%%%%%%%%%%%%%%%%%%%%%%%%%%%%%%%%%%%%%%%%%%%%%%%%%%%%%%%%%%%%%%%%%%%%%%%%%%%%%%%%%%%%%%%%%
\subsection{Proof of Proposition \ref{InfFuncExp}}
\label{InfFuncExpProof}
By  Lemma \ref{PsiuExp} stated in Appendix \ref{ANorProof} and the Fisher consistency of $\hat{\thetavec}_{u}$ we conclude that 
\begin{equation}
\label{PsiContProp}
{\rm{E}}\left[u\left(\Xmat\right)\psivec_{u}\left(\Xmat;\Smat_{u}\left[P_{\epsilon}\right]\right);P_{\epsilon}\right]=\zerovec,
\end{equation}
where $\Smat_{u}\left[\cdot\right]$ denotes the statistical functional corresponding to $\hat{\thetavec}_{u}$ and $P_{\epsilon}$ is the contaminated probability measure defined in (\ref{ContProb}).
The first-order derivative of (\ref{PsiContProp}) w.r.t. $\epsilon$ at $\epsilon=0$ satisfies:
\begin{equation}
\left.\frac{\partial{\rm{E}}\left[u\left(\Xmat\right)\psivec_{u}\left(\Xmat;\Smat_{u}\left[P_{\epsilon}\right]\right);\pxthetazero\right]}{\partial\epsilon}\right|_{\epsilon=0}
+u\left(\yvec\right)\psivec_{u}\left(\yvec;\thetavec_{0}\right)=\zerovec.
\end{equation}
Therefore, by the chain-rule for derivatives we obtain that 
\begin{equation}
\label{DerivFunc0}
\left.\frac{\partial{\rm{E}}\left[u\left(\Xmat\right)\psivec_{u}\left(\Xmat;\Smat_{u}\left[P_{\epsilon}\right]\right);\pxthetazero\right]}{\partial\Smat_{u}\left[P_{\epsilon}\right]}\right|_{\epsilon=0}
\left.\frac{\partial\Smat_{u}\left[P_{\epsilon}\right]}{\partial\epsilon}\right|_{\epsilon=0}  
+u\left(\yvec\right)\psivec_{u}\left(\yvec;\thetavec_{0}\right)=\zerovec.
\end{equation}
According to (\ref{ExpDef}) and (\ref{PsiDef})-(\ref{GammaDef})
\begin{eqnarray}
\label{DerivFunc}
\frac{\partial{\rm{E}}\left[u\left(\Xmat\right)\psivec_{u}\left(\Xmat;\Smat_{u}\left[P_{\epsilon}\right]\right);\pxthetazero\right]}{\partial\Smat_{u}\left[P_{\epsilon}\right]}
&=&\frac{\partial}{\partial\Smat_{u}\left[P_{\epsilon}\right]}\int\limits_{\XCalsc}u\left(\xvec\right)\left(\frac{\partial\log\phi^{(u)}\left(\xvec;\Smat_{u}\left[P_{\epsilon}\right]\right)}{\partial\Smat_{u}\left[P_{\epsilon}\right]}\right)^{T}
d\pxthetazero\left(\xvec\right)
\\\nonumber
&=&
\int\limits_{\XCalsc}u\left(\xvec\right)\frac{\partial^{2}\log\phi^{(u)}\left(\xvec;\Smat_{u}\left[P_{\epsilon}\right]\right)}{\partial\Smat_{u}\left[P_{\epsilon}\right]\partial\Smat^{T}_{u}\left[P_{\epsilon}\right]}d\pxthetazero\left(\xvec\right)
\\\nonumber
&=&
-\Fmat_{u}\left(\Smat_{u}\left[P_{\epsilon}\right]\right),
\end{eqnarray}
where the second equality in (\ref{DerivFunc}) stems from assumptions E-\ref{F1}-E-\ref{F3} and Theorem 2.40 in \cite{giaquinta2003mathematical} according to which integration and differentiation can be interchanged. Furthermore, by (\ref{ContProb}) and the Fisher consistency of $\hat{\thetavec}_{u}$  we have that 
\begin{equation}
\label{FFConst}
\left.\Fmat_{u}\left(\Smat_{u}\left[P_{\epsilon}\right]\right)\right|_{\epsilon=0}=\Fmat_{u}\left(\thetaveczero\right),
\end{equation}
which is non-singular by assumption E-\ref{F4}.  Therefore, by relations (\ref{DerivFunc0})-(\ref{FFConst}) and the definition (\ref{IFDef}), according to which ${\bf{IF}}\left(\yvec;\thetaveczero\right)\triangleq\left.\frac{\partial\Smat_{u}\left[P_{\epsilon}\right]}{\partial\epsilon}\right|_{\epsilon=0}$, the relation (\ref{InfFuncMTQML}) is easily verified.\qed
%%%%%%%%%%%%%%%%%%%%%%%%%%%%%%%%%%%%%%%%%%%%%%%%%%%%%%%%%%%%%%%%%%%%%%%%%%%%%%%%%%%%%%%%%%%%%%%%%%%
%%%%%%%%%%%%%%%%%%%%%%%%%%%%%%%%%%%%%%%%%%%%%%%%%%%%%%%%%%%%%%%%%%%%%%%%%%%%%%%%%%%%%%%%%%%%%%%%%%%
\subsection{Proof of Proposition \ref{RobustnessConditions}}
\label{RobustnessConditionsProof}
By (\ref{PsivecExp}) and the triangle inequality the absolute value of the $k$-th entry of the product $\psivec_{u}\left(\Xmat;\thetaveczero\right)u\left(\yvec\right)$ in (\ref{InfFuncMTQML}) satisfies: 
\begin{equation}
\label{Upperbound1}
\left|\left[\psivec_{u}\left(\yvec;\thetaveczero\right)\right]_{k}u\left(\yvec\right)\right|\leq
{c}_{1}u\left(\yvec\right)
+{c}_{2}\left(\yvec\right){u\left(\yvec\right)}{\left\|\yvec-\muvec^{\left(u\right)}_{\Xmatsc}\left(\thetaveczero\right)\right\|}+
{c}_{3}\left(\yvec\right){u\left(\yvec\right)}{\left\|\yvec-\muvec^{\left(u\right)}_{\Xmatsc}\left(\thetaveczero\right)\right\|}^{2},
\end{equation}
where 
\begin{equation}
\nonumber
c_{1}\triangleq\left|{\rm{tr}}\left[\left(\bSigma^{(u)}_{\Xmatsc}\left(\thetaveczero\right)\right)^{-1}\left.\frac{\partial\bSigma^{(u)}_{\Xmatsc}\left(\thetavec\right)}{\partial\theta_{k}}\right|_{\thetavecsc=\thetavecsc_{0}}\right]\right|,
\end{equation}
%%%
\begin{equation}
\nonumber
c_{2}\left(\yvec\right)\triangleq\left|2{\rm{Re}}\left\{\bvec^{H}\left(\yvec;\thetaveczero\right)\left(\bSigma^{(u)}_{\Xmatsc}\left(\thetaveczero\right)\right)^{-1}
\left.\frac{\partial\muvec^{\left(u\right)}_{\Xmatsc}\left(\thetavec\right)}{\partial\theta_{k}}\right|_{\thetavecsc=\thetavecsc_{0}}\right\}\right|, 
\end{equation}
%%%
\begin{equation}
\nonumber
c_{3}\left(\yvec\right)\triangleq\left|{\rm{tr}}\left[\left(\bSigma^{(u)}_{\Xmatsc}\left(\thetaveczero\right)\right)^{-1}
\left.\frac{\partial\bSigma^{(u)}_{\Xmatsc}\left(\thetavec\right)}{\partial\theta_{k}}\right|_{\thetavecsc=\thetavecsc_{0}}
\left(\bSigma^{(u)}_{\Xmatsc}\left(\thetaveczero\right)\right)^{-1}{\Bmat\left(\yvec;\thetaveczero\right)}\right]\right|, 
\end{equation}
$\bvec\left(\yvec;\thetavec\right)\triangleq{\left(\yvec-\muvec^{\left(u\right)}_{\Xmatsc}(\thetavec)\right)}{{\left\|\yvec-\muvec^{\left(u\right)}_{\Xmatsc}(\thetavec)\right\|}^{-1}}$ and $\Bmat\left(\yvec;\thetavec\right)\triangleq\bvec\left(\yvec;\thetavec\right)\bvec^{H}\left(\yvec;\thetavec\right)$. Since ${\left\|\bvec\left(\yvec;\thetavec\right)\right\|}={1}$ for any $\yvec\in\Csp^{p}$, the real and imaginary components of $\bvec\left(\yvec;\thetavec\right)$ and $\Bmat\left(\yvec;\thetavec\right)$ are bounded. Hence, there exists a positive constant $c$ that upper bounds $c_{1}$, $c_{2}\left(\yvec\right)$ and $c_{3}\left(\yvec\right)$. Therefore, by (\ref{Upperbound1}), the triangle inequality and the fact that
$\left\|\yvec\right\|<\left\|\yvec\right\|^{2}+1$ we conclude that
%%%
\begin{eqnarray}
\left|\left[\psivec_{u}\left(\yvec;\thetaveczero\right)\right]_{k}u\left(\yvec\right)\right|&\leq&
c\cdot{u}\left(\yvec\right)\left(1+\left\|\yvec-\muvec^{\left(u\right)}_{\Xmatsc}\left(\thetaveczero\right)\right\|+\left\|\yvec-\muvec^{\left(u\right)}_{\Xmatsc}\left(\thetaveczero\right)\right\|^{2}\right)
\\\nonumber
&\leq&
c\cdot{u}\left(\yvec\right)\left(\left\|\yvec\right\|^{2} + \left\|\yvec\right\|\left(2\left\|\muvec^{\left(u\right)}_{\Xmatsc}\left(\thetaveczero\right)\right\|+1\right) + 
\left\|\muvec^{\left(u\right)}_{\Xmatsc}\left(\thetaveczero\right)\right\|^{2}+ \left\|\muvec^{\left(u\right)}_{\Xmatsc}\left(\thetaveczero\right)\right\| + 1\right)
\\\nonumber
&<&
c\cdot{u}\left(\yvec\right)\left(2\left(\left\|\muvec^{\left(u\right)}_{\Xmatsc}\left(\thetaveczero\right)\right\|+1\right)\left\|\yvec\right\|^{2} + \left( 
\left\|\muvec^{\left(u\right)}_{\Xmatsc}\left(\thetaveczero\right)\right\|^{2} + 3\left\|\muvec^{\left(u\right)}_{\Xmatsc}\left(\thetaveczero\right)\right\|+1\right) \right).
\end{eqnarray}
Thus, the influence function (\ref{InfFuncMTQML}) is bounded over $\mathcal{C}\subseteq\Csp^{p}$ if $u\left(\yvec\right)$ and $u\left(\yvec\right)\left\|\yvec\right\|^{2}$ are bounded over $\mathcal{C}$. 
Furthermore, if over the subset $\mathcal{C}$ $u\left(\yvec\right)\rightarrow0$ and $u\left(\yvec\right)\left\|\yvec\right\|^{2}\rightarrow0$ as $\left\|\yvec\right\|\rightarrow{\infty}$, then   
$\left\|{\bf{IF}}\left(\yvec;\thetaveczero\right)\right\|\rightarrow{0}$ over $\mathcal{C}$ as $\left\|\yvec\right\|\rightarrow{\infty}$.
\qed
%%%%%%%%%%%%%%%%%%%%%%%%%%%%%%%%%%%%%%%%%%%%%%%%%%%%%%%%%%%%%%%%%%%%%%%%%%%%%%%%%%%%%%%%%%%%%%%%%%%
%%%%%%%%%%%%%%%%%%%%%%%%%%%%%%%%%%%%%%%%%%%%%%%%%%%%%%%%%%%%%%%%%%%%%%%%%%%%%%%%%%%%%%%%%%%%%%%%%%%
\subsection{Proof of Proposition \ref{RobGainProp}:}   
\label{RobGainPropProof}
According to (\ref{GaussMTFuncOrth}) $\uGausss\left(\yvec;\omega\right)\leq\exp\left(-\frac{\epsilon\left\|\yvec\right\|^{2}}{\omega^{2}}\right)$ for any $\yvec\in\mathcal{B}_{\epsilon}$.
The proof is now complete by noting that for any fixed width parameter $\omega$ the functions $\exp\left(-\frac{\epsilon\left\|\yvec\right\|^{2}}{\omega^{2}}\right)$ and $\left\|\yvec\right\|^{2}\exp\left(-\frac{\epsilon\left\|\yvec\right\|^{2}}{\omega^{2}}\right)$ are bounded and decay to zero as $\left\|\yvec\right\|\rightarrow\infty$.
\qed

%%%%%%%%%%%%%%%%%%%%%%%%%%%%%%%%%%%%%%%%%%%%%%%%%%%%%%%%%%%%%%%%%%%%%%%%%%%%%%%%%%%%%%%%%%%%%%%%%%%%%%%%%%%%%%%%%%%%%%%%%%%%%%%%%%%%%%%%%%%%%%%%%%%%%%%%%%%%%%%%%%%%%%%%%%%%%%%%%%%%%%%%%%%%%%%%%%%%%%%%%%%%%%%%%%%%%%%%%%%%%%%%%%%%%%%%%%

\bibliographystyle{IEEEbib}

\bibliography{strings,refs}

\begin{thebibliography}{9}

\scriptsize{
\bibitem{Fisher1} R. A. Fisher, ``On the mathematical foundations of theoretical statistics,'' {\em Philosophical Transactions of the Royal Society of London,}  Series A, vol. 222, pp. 309-368, 1922.
 
\bibitem{Fisher2} R. A. Fisher, ``Theory of statistical estimation,'' {\em Proceedings of the Cambridge Philosophical Society,} vol. 22, pp. 700-725, 1925.

\bibitem{Lehmann}  E. L. Lehmann,  and G. Casella, {\em Theory of Point Estimation,} Springer, 1998.

\bibitem{KayBook} S. M. Kay, {\em Fundamentals of statistical signal processing: estimation theory,} Prentice-Hall, 1998.

\bibitem{White} H. White, ``Maximum likelihood estimation of misspecified models,'' {\em Econometrica: Journal of the Econometric Society,} pp. 1-25, 1982.      

\bibitem{Kay} Q. Ding and S. Kay,  ``Maximum likelihood estimator under a misspecified model with high signal-to-noise ratio,'' {\em IEEE Transactions on Signal Processing,} vol. 59, no. 8, pp. 4012-4016, 2011.

\bibitem{Serfling} R. J. Serfling, {\em Approximation theorems of mathematical statistics.}  John Wiley \& Sons, 1980.

\bibitem{HuberRob} P. J. Huber, {\em Robust statistics,} Springer, 2011.

\bibitem{GMM} L. P. Hansen, ``Large sample properties of generalized method of moments estimators,'' {\em Econometrica: Journal of the Econometric Society,} vol. 50, no. 4, pp. 1029-1054, 1982.

\bibitem{Kullback} S. Kullback and R. A. Leibler, ``On information and sufficiency,'' {\em The Annals of Mathematical Statistics,} vol. 22, pp. 79-86, 1951.

\bibitem{rousseaux2005gaussian} O. Rousseaux, G.  Leus, P. Stoica, ``Gaussian maximum-likelihood channel estimation with short training sequences,'' {\em IEEE Transactions on Wireless Communications,} vol. 4, no. 6, pp. 2945-2955, 2005.

\bibitem{Weinstein} E. Weinstein, ``Decentralization of the Gaussian maximum likelihood estimator and its applications to passive array processing,'' {\em IEEE Transactions on Acoustics, Speech and Signal Processing,} vol. 29, no. 5, pp. 945-951, 1981.  

\bibitem{bollerslev1992quasi} T. Bollerslev and J. M. Wooldridge, ``Quasi-maximum likelihood estimation and inference in dynamic models with time-varying covariances,'' {\em Econometric reviews,} vol. 11, no. 2, pp. 143-172, 1992.

\bibitem{Trognon} C. Gourieroux, A. Monfort, and A. Trognon, ``Pseudo maximum likelihood methods: Theory,'' {\em Econometrica,} vol. 52, no. 3, pp. 681-700, 1984.

\bibitem{Ottersten} B. Ottersten, M. Viberg, P. Stoica, and A. Nehorai, ``Exact and large sample maximum likelihood techniques for parameter estimation and detection,'' {\em in Radar Array Processing,} ch. 4, pp. 99-151, Springer-Verlag, 1993.

\bibitem{Viberg} H. Krim and M. Viberg, ``Two decades of array signal processing research: the parametric approach,'' {\em IEEE Signal Processing Magazine,} vol. 13, no. 4, pp. 67-94, 1996.
     
\bibitem{Pearson} K. Pearson, ``Contributions to the mathematical theory of evolution,'' {\em Philosophical Transactions of the Royal Society of London,} vol. A, pp. 71-110, 1894.

\bibitem{bollerslev1987conditionally} T. Bollerslev, ``A conditionally heteroskedastic time series model for speculative prices and rates of return,'' {\em The review of economics and statistics,} vol. 69, no. 3, pp. 542-547, 1987.

\bibitem{baillie2002message} R. Baillie and T. Bollerslev, ``The message in daily exchange rates: a conditional-variance tale,'' {\em Journal of Business \& Economic Statistics,} vol. 7, pp. 297-305, 1989.

\bibitem{hsieh1989modeling} D. A. Hsieh, ``Modeling heteroscedasticity in daily foreign-exchange rates,'' {\em Journal of Business \& Economic Statistics,} vol. 7, no. 3, pp. 307-317, 1989.

\bibitem{fan2014quasi} J. Fan, L. Qi, and D. Xiu, ``Quasi-maximum likelihood estimation of GARCH models with heavy-tailed likelihoods,'' {\em Journal of Business \& Economic Statistics,} vol. 32, no. 2, pp. 178-191, 2014.

\bibitem{williams1993robust} D. B. Williams and D. H. Johnson, ``Robust estimation of structured covariance matrices,'' {\em  IEEE Transactions on Signal Processing,} vol. 41, no. 9, pp. 2891-2906, 1993.

\bibitem{yardimci1998robust} Y. Yardimci, A. E. Cetin, and J. A. Cadzow, ``Robust direction-of-arrival estimation in non-Gaussian noise,'' {\em  IEEE Transactions on Signal Processing,} vol. 46, no. 5, pp. 1443-1451, 1998.

\bibitem{georgiou2006maximum} P. G. Georgiou, and C. Kyriakakis, ``Maximum likelihood parameter estimation under impulsive conditions, a sub-Gaussian signal approach,'' {\em  Signal processing,} vol. 86, no. 10, pp. 3061-3075, 2006.

\bibitem{Huber} P. J. Huber, ``The behaviour of maximum likelihood estimates under nonstandard conditions,'' {\em Proceedings of the fifth Berkeley symposium on mathematical statistics and probability,} vol. 1, no. 1, pp. 221-233, 1967.

\bibitem{newey1997asymptotic} W. K. Newey and D. G. Steigerwald, ``Asymptotic bias for quasi-maximum-likelihood estimators in conditional heteroskedasticity models,'' {\em Econometrica,} vol. 65, no. 3, pp. 587-599, 1997.

\bibitem{MTCCA} K. Todros and A. O. Hero, ``On measure-transformed canonical correlation analysis,'' {\em IEEE Transactions on Signal Processing,} vol. 60, no. 9, pp. 4570-4585, 2012.

\bibitem{MTCCAAPP} K. Todros and A. O. Hero, ``Measure transformed canonical correlation analysis with application to financial data,'' {\em Proceedings of SAM 2012,} pp. 361-364, 2012.

\bibitem{Todros2} K. Todros and A. O. Hero, ``Robust multiple signal classification via probability measure transformation,'' {\em IEEE Transactions on Signal Processing,}  vol. 63, no. 5, pp. 1156-1170, 2015.

\bibitem{Todros3} K. Todros and A. O. Hero, ``Robust Measure transformed MUSIC for DOA estimation,'' {\em Proceeding of ICASSP 2014,} pp. 4190-4194, 2014. 

\bibitem{Ahlswede} R. Ahlswede and M. V. Burnashev, ``On minimax estimation in the presence of side information about remote data,'' {\em Annals of statistics,} vol . 18, no. 1, pp. 141-171, 1990. 

\bibitem{Hampel} F. R. Hampel, E. M. Ronchetti, P. J.  Rousseeuw and W. A. Stahel, {\em Robust statistics: the approach based on influence functions}.  John Wiley \& Sons, 2011.

\bibitem{MTQMLConf} K. Todros and A. O. Hero, ``Measure-transformed quasi maximum likelihood estimation with application to source localization,'' {\em Proceeding of ICASSP 2015,} pp. 3462-3466, 2015. 

\bibitem{Folland} G. B. Folland, {\em Real Analysis,} John Wiley and Sons, 1984.

\bibitem{Schreier} P. Schreier and L. L. Scharf, {\em Statistical signal processing of complex-valued data: the theory of improper and noncircular signals.} p. 39, Cambridge University Press, 2010.

\bibitem{LogDetDiv} I. S. Dhillon and J. A. Tropp,  ``Matrix nearness problems with Bregman divergences,''  {\em SIAM Journal on Matrix Analysis and Applications,} vol. 29, no. 4, pp. 1120-1146, 2007.

\bibitem{MeasureTheory} K. B. Athreya and S. N. Lahiri, {\em Measure theory and probability theory,} Springer-Verlag, 2006.

\bibitem{Cox} D. R. Cox and D. V. Hinkley, {\em Theoretical Statistics,} Chapman \& Hall, 1974.

\bibitem{Cramer} H. Cram\'{e}r, ``A contribution to the theory of statistical estimation," {\em Skand. Akt. Tidskr.,} vol. 29, pp. 85-94, 1946.

\bibitem{Rao} C. R. Rao, ``Information and accuracy attainable in the estimation of  statistical parameters," {\em Bull. Calcutta Math. Soc.,} vol. 37, pp. 81-91, 1945.

\bibitem{Visa} E. Ollila, D. E. Tyler, V. Koivunen and H. V. Poor, ``Complex elliptically symmetric distributions: survey, new results and applications,'' {\em IEEE Transactions on Signal Processing,} vol. 60, no. 1, pp. 5597-5625, 2012.

\bibitem{Tukey} J. W. Tukey, {\em Exploratory Data Analysis,} Addison-Wesley, 1977.

\bibitem{Trees} H. L. Van Trees, {\em Optimum Array Processing,} p. 922, John Wiley \& Sons, 2002.

\bibitem{Visuri} S. Visuri, H. Oja and V. Koivunen, ``Subspace-based direction-of-arrival estimation using nonparametric statistics,'' {\em IEEE Transactions on Signal Processing,} vol. 49, no. 9, pp. 2060-2073, 2001.

%\bibitem{CG1} E. Conte and M. Longo, ``Characterization of radar clutter as a spherically invariant random process,'' {\em IEE Proceedings F (Communications, Radar and Signal Processing),}  vol. 134, no. 2, pp. 191-197, 1987.
%
%\bibitem{CG2} F. Gini and A. Farina, ``Vector subspace detection in compound-Gaussian clutter, part I: Survey and new results,'' {\em IEEE Transactions on Aerospace and Electronic Systems,}  vol. 38, no. 4, pp. 1295-1311, 2002.
%
%\bibitem{CG3} M. Rangaswamy, ``Spherically invariant random processes for modelling non-Gaussian radar clutter,'' {\em Proceeding of the 27th ASILOMAR Conference on Signals, Systems and Computers,} pp. 1106-1110, 1993.

\bibitem{MannWald} H. B. Mann and A. Wald, ``On stochastic limit and order relationships,'' {\em Ann. Math. Stat.,} vol. 14, pp. 217-226, 1943.

\bibitem{White2} H. White, {\em Estimation, inference and specification analysis,} Cambridge university press, 1996.

\bibitem{AdCal} P. M. Fitzpatrick, {\em Advanced calculus,} American Mathematical Society, 2006.

\bibitem{CS} X. Yang, ``Some trace inequalities for operators,'' {\em Journal of the Australian Mathematical Society (Series A),} vol. 58, no. 3, pp. 281-286, 1995.

\bibitem{AdCalMult} C. H. Edwards, {\em Advanced calculus of several variables,} Courier Corporation, 2012.

\bibitem{newey} W. K. Newey and D. McFadden, ``Large sample estimation and hypothesis testing,''  {\em Handbook of econometrics,} vol. 4, pp. 2111-2245, 1994.

\bibitem{giaquinta2003mathematical} M. Giaquinta and G. Modica, {\em Mathematical Analysis: An Introduction to Functions of Several Variables.} Birkh{\"a}user, p. 88, Boston, 2009.


}
\end{thebibliography}
%\newpage

\end{document}